\documentclass{article}

\usepackage[accepted]{icml2024}

\usepackage{amsmath}
\usepackage{amssymb}
\usepackage{mathtools}
\usepackage{amsthm}
\usepackage{amsfonts}       %
\usepackage{nicefrac}       %
\usepackage{xcolor}         %
\usepackage{bbm}
\usepackage{caption}
\usepackage{subcaption}
\usepackage{algorithmic}
\usepackage{algorithm}
\usepackage{hyperref}
\usepackage{dsfont}
\usepackage{wrapfig}
\usepackage{graphicx}
\usepackage{comment}
\usepackage{url}
\usepackage{microtype}
\usepackage{booktabs} %

\DeclareCaptionFont{red}{\color{red}}

\newcommand{\argmax}{\mathop{\rm arg~max}\limits}
\newcommand{\argmin}{\mathop{\rm arg~min}\limits}

\theoremstyle{plain}
\newtheorem{theorem}{Theorem}[section]

\newtheorem{lemma}[theorem]{Lemma}

\newtheorem{assumption}[theorem]{Assumption}

\newtheorem{example}[theorem]{Example}

\makeatletter
\newtheorem*{rep@theorem}{\rep@title}
\newcommand{\newreptheorem}[2]{%
	\newenvironment{rep#1}[1]{%
		\def\rep@title{#2 \ref{##1}}%
		\begin{rep@theorem}}%
		{\end{rep@theorem}}}
\makeatother

\newreptheorem{theorem}{Theorem}
\newreptheorem{lemma}{Lemma}

\definecolor{airforceblue}{rgb}{0.36, 0.54, 0.66}
\definecolor{aliceblue}{rgb}{0.94, 0.97, 1.0}
\definecolor{bluegray}{rgb}{0.4, 0.6, 0.8}
\definecolor{bluefill}{rgb}{0.5, 0.7, 0.9}
\definecolor{cornflowerblue}{rgb}{0.39, 0.58, 0.93}
\definecolor{coolblack}{rgb}{0.0, 0.18, 0.39}
\definecolor{coolblack2}{rgb}{0.05, 0.23, 0.49}
\definecolor{beaublue}{rgb}{0.74, 0.83, 0.9}
\definecolor{dmred300}{HTML}{FF617B}
\hypersetup{
    colorlinks,
    linkcolor={blue},
    citecolor={blue},
    urlcolor={blue}
}

\allowdisplaybreaks[1]

\usepackage[capitalize,noabbrev]{cleveref}

\usepackage[textsize=tiny]{todonotes}

\icmltitlerunning{Adaptively Perturbed Mirror Descent for Learning in Games}

\begin{document}

\twocolumn[
\icmltitle{Adaptively Perturbed Mirror Descent for Learning in Games}

\icmlsetsymbol{equal}{*}

\begin{icmlauthorlist}
\icmlauthor{Kenshi Abe}{ca,uec}
\icmlauthor{Kaito Ariu}{ca}
\icmlauthor{Mitsuki Sakamoto}{ca}
\icmlauthor{Atsushi Iwasaki}{uec}
\end{icmlauthorlist}

\icmlaffiliation{ca}{CyberAgent, Tokyo, Japan}
\icmlaffiliation{uec}{University of Electro-Communications, Tokyo, Japan}

\icmlcorrespondingauthor{Kenshi Abe}{abekenshi1224@gmail.com}

\icmlkeywords{Machine Learning, ICML}

\vskip 0.3in
]

\printAffiliationsAndNotice{}  %

\begin{abstract}
This paper proposes a payoff perturbation technique for the Mirror Descent (MD) algorithm in games where the gradient of the payoff functions is monotone in the strategy profile space, potentially containing additive noise. The optimistic family of learning algorithms, exemplified by optimistic MD, successfully achieves {\it last-iterate} convergence in scenarios devoid of noise, leading the dynamics to a Nash equilibrium. A recent re-emerging trend underscores the promise of the perturbation approach, where payoff functions are perturbed based on the distance from an anchoring, or {\it slingshot}, strategy. In response, we propose {\it Adaptively Perturbed MD} (APMD), which adjusts the magnitude of the perturbation by repeatedly updating the slingshot strategy at a predefined interval. This innovation empowers us to find a Nash equilibrium of the underlying game with guaranteed rates. Empirical demonstrations affirm that our algorithm exhibits significantly accelerated convergence.
\end{abstract}

\section{Introduction}
\label{sec:introduction}
This study delves into a variant of Mirror Descent (MD)~\citep{nemirovskij1983problem,beck2003mirror} in the context of monotone games whose gradient of the payoff functions exhibits monotonicity concerning the strategy profile space. This encompasses diverse games, including Cournot competition \citep{bravo2018bandit}, $\lambda$-cocoercive games \citep{lin2020finite}, concave-convex games, and zero-sum polymatrix games \citep{cai2011minmax,cai2016zero}. Due to their extensive applicability, various learning algorithms have been developed and scrutinized to compute a Nash equilibrium efficiently.

Agents, which are prescribed to play according to MD or its variant,  choose strategies with higher expected payoffs more likely but do not move far away from current strategies via regularization. The dynamics is known to converge to an equilibrium in an average sense, which is referred to as \textit{average-iterate convergence}. In other words, the averaged strategy profile over iterations converges to an equilibrium. Nevertheless, research has shown that the actual trajectory of the updated strategy profiles fails to converge even in two-player zero-sum games and a specific class within monotone games~\citep{mertikopoulos2018cycles,bailey2018multiplicative}.
On the contrary, optimistic learning algorithms, incorporating recency bias, have shown success. The updated strategy profile itself converges to a Nash equilibrium~\citep{daskalakis2017training,daskalakis2018last,mertikopoulos2018optimistic,wei2020linear}, termed {\it last-iterate convergence}.

However, the optimistic approach faces challenges with feedback contaminated by some noise. Typically, each agent updates his or her strategy according to the perfect gradient feedback of the payoff function at each iteration, denoted as \textit{full feedback}. In a more realistic scenario, noise might distort this feedback. With \textit{noisy feedback}, optimistic learning algorithms perform suboptimally. For instance, \citet{abe2022last} empirically demonstrated that optimistic Multiplicative Weights Update (OMWU) fails to converge to an equilibrium, orbiting around it.

Alternatively, perturbation of payoffs has emerged again as a pivotal concept for achieving last-iterate convergence, even under noise~\citep{abe2022last}. Payoff perturbation is a classical technique, as seen in \citet{facchinei2003finite}, and introduces strongly convex penalties to the players' payoff functions to stabilize learning. This leads to convergence to approximate equilibria, not only in the full feedback setting but also in the noisy feedback setting. However, to ensure convergence toward a Nash equilibrium of the underlying game, the magnitude of perturbation requires careful adjustment, which is calculated as the product of a strongly convex penalty function and a perturbation strength parameter.
In fact, \citet{liu2022power} 
shrink the perturbation strength based on the current strategy profile's proximity to an underlying equilibrium.
Similarly, iterative Tikhonov regularization methods~\citep{koshal2010single,tatarenko2019learning} adjust the magnitude of perturbation by using a sequence of perturbation strengths that satisfy certain conditions, such as diminishing as the iteration increases. Although these algorithms admit last-iterate convergence, it becomes challenging to choose an appropriate learning rate for the shrinking perturbation strength, which often leads to a failure in achieving rapid convergence for these algorithms and their variants.

In response to this, we adaptively determine the amount of the penalty from the distance $G(\cdot,\cdot)$ between the current strategy $\pi$ and an anchoring, or {\it slingshot} strategy $\sigma$, while maintaining the perturbation strength parameter $\mu$ constant. Instead of carefully decaying the perturbation strength, the slingshot strategies $\sigma$ are re-initialized at a predefined interval $T_\sigma$ by the current strategies, and thus the magnitude of the perturbation $\mu G(\cdot,\cdot)$ is adjusted. To the best of our knowledge, \citet{perolat2021poincare} were the first to employ this idea and enabled \citet{abe2022last} to modify MWU to achieve last-iterate convergence. However, they have established the convergence only in an asymptotic manner. The significance of our work, in part, lies in extending these two studies and establishing non-asymptotic convergence results.

Our contributions are manifold. First, we identify our algorithm as \textit{Adaptively Perturbed MD}\footnote{An implementation of our method is available at \url{https://github.com/CyberAgentAILab/adaptively-perturbed-md}} (APMD). Second, we analyze the case where both the perturbation function and the proximal regularizer are assumed to be the squared $\ell^2$-distance and provide the last-iterate convergence rates to a Nash equilibrium, $\mathcal{O}(\ln T/\sqrt{T})$ and $\mathcal{O}(\ln T/T^{\frac{1}{10}})$ with full and noisy feedback, respectively. We also discuss the case where different distances from the squared $\ell^2$-distance are utilized. Finally, we empirically reveal that our proposed APMD significantly outperforms MWU and OMWU in two zero-sum polymatrix games, regardless of the feedback type.

\section{Preliminaries}
\paragraph{Monotone games.}
This paper focuses on a continuous game, which is denoted by $([N], (\mathcal{X}_i)_{i\in [N]}, (v_i)_{i\in [N]})$.
$[N]=\{1, 2, \cdots, N\}$ represents the set of $N$ players, $\mathcal{X}_i\subseteq \mathbb{R}^{d_i}$ represents the $d_i$-dimensional compact convex strategy space for player $i\in [N]$, and we write $\mathcal{X}=\prod_{i\in [N]}\mathcal{X}_i$.
Each player $i$ chooses a {\it strategy} $\pi_i$ from $\mathcal{X}_i$ and aims to maximize her differentiable payoff function $v_i:\mathcal{X}\to \mathbb{R}$.
We write $\pi_{-i}\in \prod_{j\neq i}\mathcal{X}_i$ as the strategies of all players except player $i$, and denote the {\it strategy profile} by $\pi=(\pi_i)_{i\in [N]}\in \mathcal{X}$.
This study particularly considers a {\it smooth monotone game}, where the gradient $(\nabla_{\pi_i} v_i)_{i \in [N]}$ of the payoff functions is monotone: $\forall \pi, \pi'\in \mathcal{X},$
\begin{align}
\label{eq:monotone}
    \sum_{i=1}^N\langle \nabla_{\pi_i} v_i(\pi_i,\pi_{-i}) - \nabla_{\pi_i} v_i(\pi_i',\pi_{-i}'), \pi_i - \pi'_i\rangle  \leq 0,
\end{align}
and $L$-Lipschitz:
\begin{align}
\label{eq:smooth}
    \!\!\!\!\sum_{i=1}^N\!\|\nabla_{\pi_i} v_i(\pi_i, \pi_{-i}) \!-\! \nabla_{\pi_i} v_i(\pi_i', \pi_{-i}')\|^2 \!\leq \! L^2\|\pi \!-\! \pi'\|^2\!\!,
\end{align}
where $\|\cdot\|$ is the $\ell^2$-norm.

Monotone games include many common and well-studied classes of games, such as concave-convex games, zero-sum polymatrix games, and Cournot competition.
\begin{example}[Concave-Convex Games] 
\normalfont
Let us consider a max-min game $(\{1, 2\}, (\mathcal{X}_1, \mathcal{X}_2), (v, -v))$, where $v:\mathcal{X}_1\times \mathcal{X}_2\to \mathbb{R}$.
Player $1$ aims to maximize $v$, while player $2$ aims to minimize $v$.
If $v$ is concave in $x_1\in \mathcal{X}_1$ and convex in $x_2\in \mathcal{X}_2$, the game is called a concave-convex game or minimax optimization problem, and it is easy to confirm that the game is monotone.
\end{example}

\begin{example}[Zero-Sum Polymatrix Games] 
\normalfont
In a zero-sum polymatrix game, each player's payoff function can be decomposed as $v_i(\pi) = \sum_{j\neq i}u_i(\pi_i, \pi_j)$, where $u_i:\mathcal{X}_i\times \mathcal{X}_j\to \mathbb{R}$ is represented by $u_i(\pi_i, \pi_j) = \pi_i^{\top}\mathrm{M}^{(i,j)}\pi_j$ with some matrix $\mathrm{M}^{(i,j)}\in \mathbb{R}^{d_i\times d_j}$, and satisfies $u_i(\pi_i, \pi_j) = -u_j(\pi_j, \pi_i)$.
In this game, each player $i$ can be interpreted as playing a two-player zero-sum game with each other player $j\neq i$.
From the linearity and zero-sum property of $u_i$, we can easily show that $\sum_{i=1}^N\langle \nabla_{\pi_i} v_i(\pi_i,\pi_{-i}) - \nabla_{\pi_i} v_i(\pi_i',\pi_{-i}'), \pi_i - \pi'_i\rangle = 0$.
Thus, the zero-sum polymatrix game is a special case of monotone games.
\end{example}

\paragraph{Nash equilibrium and gap function.}
A {\it Nash equilibrium} \citep{nash1951non} is a common solution concept of a game, which is a strategy profile where no player can improve her payoff by deviating from her specified strategy.
Formally, a Nash equilibrium $\pi^{\ast}\in \mathcal{X}$ satisfies the following condition:
\begin{align*}
    \forall i\in [N], \forall \pi_i\in \mathcal{X}_i, ~v_i(\pi_i^{\ast}, \pi_{-i}^{\ast}) \geq v_i(\pi_i, \pi_{-i}^{\ast}).
\end{align*}
We denote the set of Nash equilibria by $\Pi^{\ast}$.
Note that a Nash equilibrium always exists for any smooth monotone game \citep{debreu1952social}.
Furthermore, we measure the proximity to Nash equilibrium for a given strategy profile $\pi$ by its {\it gap function}, which is defined as:
\begin{align*}
\mathrm{GAP}(\pi):=\max_{\tilde{\pi}\in \mathcal{X}}\sum_{i=1}^N\langle \nabla_{\pi_i}v_i(\pi_i, \pi_{-i}), \tilde{\pi}_i - \pi_i\rangle.
\end{align*}
The gap function is a standard metric of proximity to Nash equilibrium for a given strategy profile $\pi$ \citep{cai2023doubly}.
From the definition, $\mathrm{GAP}(\pi) \geq 0$ for any $\pi\in \mathcal{X}$, and the equality holds if and only if $\pi$ is a Nash equilibrium.

\paragraph{Problem setting.}
In this study, we consider the online learning setting where the following process is repeated for $T$ iterations: 1) At each iteration $t\geq 0$, each player $i\in [N]$ chooses her strategy $\pi_i^t\in \mathcal{X}_i$ based on the previously observed feedback; 2) Each player $i$ receives the gradient feedback $\widehat{\nabla}_{\pi_i}v_i(\pi_i^t, \pi_{-i}^t)$ as feedback.
This study considers two feedback models: {\it full feedback} and {\it noisy feedback}.
In the full feedback setting, each player receives the perfect gradient vector as feedback, i.e., $\widehat{\nabla}_{\pi_i}v_i(\pi_i^t, \pi_{-i}^t) = \nabla_{\pi_i}v_i(\pi_i^t, \pi_{-i}^t)$.
In the noisy feedback setting, each player's feedback is given by $\widehat{\nabla}_{\pi_i}v_i(\pi_i^t, \pi_{-i}^t) = \nabla_{\pi_i}v_i(\pi_i^t, \pi_{-i}^t) + \xi_i^t$, where $\xi_i^t\in \mathbb{R}^{d_i}$ is a noise vector.
Specifically, we focus on the zero-mean and bounded-variance noise vectors.

\paragraph{Mirror Descent.}
Mirror Descent (MD) is a widely used algorithm for learning equilibria in games.
Let us define $\psi: \mathbb{R}^{d_i} \to \mathbb{R}\cup \{\infty\}$ as the strictly convex {\it regularization function} and $D_{\psi}(\pi_i, \pi_i')=\psi(\pi_i) - \psi(\pi_i') - \langle \nabla \psi(\pi_i'), \pi-\pi_i'\rangle$ as the associated {Bregman divergence}.
Then, MD updates each player $i$'s strategy $\pi_i^t$ at iteration $t$ as follows:
\begin{align*}
\pi_i^{t+1} &= \argmax_{x\in \mathcal{X}_i} \left\{\eta_t\left\langle \widehat{\nabla}_{\pi_i}v_i(\pi^t), x\right\rangle - D_{\psi}(x, \pi_i^t)\right\},
\end{align*}
where $\eta_t\in (0, \infty)$ is the learning rate at iteration $t$.

\begin{figure}[t!]
\begin{algorithm}[H]
    \caption{APMD for player $i$.}
    \label{alg:mdsp}
    \begin{algorithmic}[1]
    \REQUIRE{Learning rate sequence $\{\eta_t\}_{t \ge 0}$, divergence function for perturbation $G$, perturbation strength $\mu$, update interval $T_{\sigma}$, initial strategy $\pi_i^0$}
    \STATE $k\gets 0, ~\tau \gets 0$
    \STATE $\sigma_i^0 \gets \pi_i^0$
    \FOR{$t=0,1,2,\cdots, T$}
        \STATE Receive the gradient feedback $\widehat{\nabla}_{\pi_i}v_i(\pi^t)$
        \STATE Update the strategy by
        $$\!\!\!\!\!\pi_i^{t+1} \!=\! \argmax_{x\in \mathcal{X}_i} \!\bigg\{\!\eta_t\left\langle \widehat{\nabla}_{\pi_i}v_i(\pi^t) \!-\! \mu \nabla_{\pi_i}G(\pi_i^t, \sigma_i^k), x\!\right\rangle$$
        $$\!\!\!\!\!\!\!\!\!\!\!\! - D_{\psi}(x, \pi_i^t)\bigg\}$$
        \STATE $\tau \gets \tau + 1$
        \IF{$\tau=T_{\sigma}$}
            \STATE $k\gets k+1, ~\tau\gets 0$
            \STATE $\sigma_i^k\gets \pi_i^{t+1}$
        \ENDIF
    \ENDFOR
    \end{algorithmic}
\end{algorithm}
\end{figure}

\paragraph{Other notations.}
We denote a $d$-dimensional probability simplex by $\Delta^{d}=\{p\in [0,1]^d ~|~ \sum_{j=1}^dp_j = 1\}$.
We define $\mathrm{diam}(\mathcal{X}):= \sup_{\pi, \pi'\in \mathcal{X}}\|\pi - \pi'\|$ as the diameter of $\mathcal{X}$. 
The {\it Kullback-Leibler} (KL) divergence is defined by $\mathrm{KL}(\pi_i, \pi_i') = \sum_{j=1}^{d_i}\pi_{ij}\ln \frac{\pi_{ij}}{\pi_{ij}'}$.
Besides, with a slight abuse of notation, we represent the sum of Bregman divergences and the sum of KL divergences by $D_{\psi}(\pi, \pi')=\sum_{i=1}^N D_{\psi}(\pi_i, \pi_i')$, and $\mathrm{KL}(\pi, \pi')=\sum_{i=1}^N \mathrm{KL}(\pi_i, \pi_i')$, respectively.
We finally denote the domain of $\psi$ by $\mathrm{dom~}\psi := \{x: \psi(x) < \infty\}$, and corresponding interior by $\mathrm{int}(\mathrm{dom}~\psi)$.

\section{Adaptively Perturbed Mirror Descent}
In this section, we present Adaptively Perturbed Mirror Descent (APMD), which is an extension of the standard MD algorithms. Algorithm~\ref{alg:mdsp} describes the pseudo-code. APMD employs two pivotal techniques: {\it slingshot payoff perturbation} and {\it slingshot strategy update}. Each of them corresponds to line 5 and line 9 in Algorithm~\ref{alg:mdsp}, respectively.

\begin{figure}[t!]
    \centering
    \includegraphics[width=.4\textwidth]{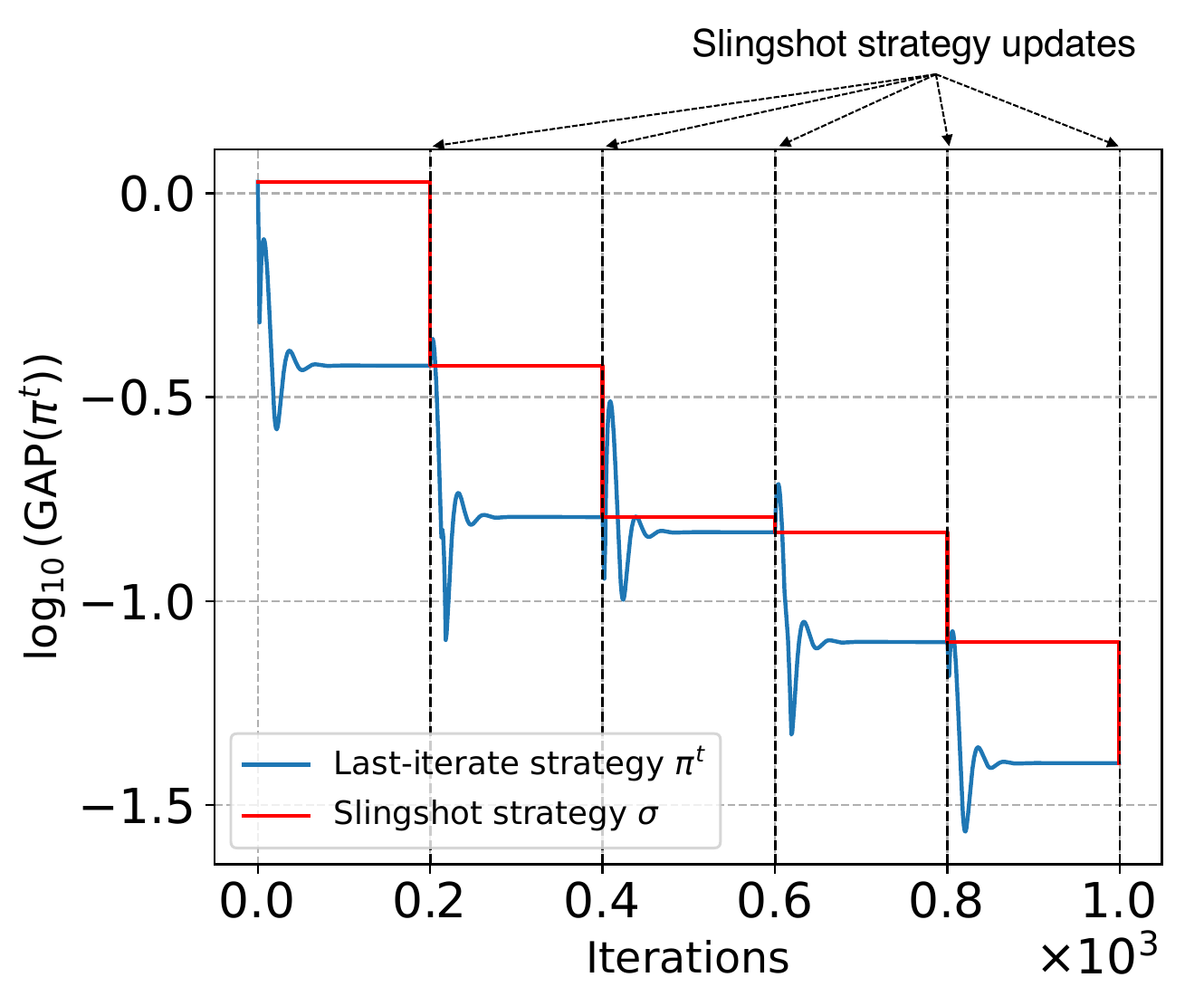}
    \caption{
        Illustration of the impact of the slingshot strategy updates on the gap function for $\pi^t$ updated by APMD.
    }
    \label{fig:exploitability_compare_pi_sigma}
\end{figure}

\subsection{Slingshot Payoff Perturbation}
Letting us define the differentiable divergence function $G(\cdot, \cdot): \mathbb{R}^{d_i}\times \mathbb{R}^{d_i} \to \mathbb{R}\cup \{\infty\}$ and the {\it slingshot} strategy $\sigma_i\in \mathcal{X}_i$, APMD perturbs each player's payoff by the divergence between the current strategy $\pi_i^t$ and the slingshot strategy $\sigma_i$, i.e., $G(\pi_i^t, \sigma_i)$. 
Specifically, APMD updates each player's strategy according to 
\begin{align}
    \label{eq:gm-md}
    \pi_i^{t+1} &= \argmax_{x\in \mathcal{X}_i} \bigg\{\eta_t\left\langle \widehat{\nabla}_{\pi_i}v_i(\pi^t) - \mu \nabla_{\pi_i}G(\pi_i^t, \sigma_i), x\right\rangle \nonumber\\
    &\phantom{\argmax_{x\in \mathcal{X}_i}==} - D_{\psi}(x, \pi_i^t)\bigg\}.
\end{align}
where $\mu \in (0, \infty)$ is the {\it perturbation strength} and $\nabla_{\pi_i}G$ denotes differentiation with respect to first argument.
We assume that $G(\cdot, \sigma_i)$ is strictly convex for every $\sigma_i\in \mathcal{X}_i$, and takes a minimum value of $0$ at $\sigma_i$.
Furthermore, we assume that $\psi$ is differentiable and $\rho$-strongly convex on $\mathcal{X}_i$ with $\rho\in (0, \infty)$.

The conventional MD updates its strategy based on the gradient feedback of the payoff function and the regularization term. The regularization term adjusts the next strategy so that it does not deviate significantly from the current strategy.
APMD perturbs the gradient payoff vector by the divergence between the current strategy $\pi_i^t$ and a predefined slingshot strategy $\sigma_i$. 
If the current strategy is far away from the slingshot strategy, the magnitude of the perturbation increases. 
Note that if both strategies are equivalent, no perturbation occurs, and APMD just seeks a strategy with a higher expected payoff. 
As Figure~\ref{fig:exploitability_compare_pi_sigma} illustrates, the divergence first fluctuates, and then the current strategy profile comes close to a stationary point where the gradient of the expected payoff vector is equal to the gradient of the magnitude of the perturbation so that the perturbation via slingshot stabilizes the learning dynamics.
Indeed, Mutant Follow the Regularized Leader (Mutant FTRL) instantiated in Example \ref{exm:mftrl} encompasses replicator-mutator dynamics, which is guaranteed to converge to an approximate equilibrium in two-player zero-sum games~\citep{abe2022mutationdriven}. We can argue that APMD inherits this nice feature. 

\subsection{Slingshot Strategy Update}
The perturbation via slingshot enables $\pi^t$ to converge quickly to a stationary point (Lemmas~\ref{lem:md_sp_full} and \ref{lem:md_sp_noisy}). Different slingshot strategy profiles induce different stationary points. Of course, when the slingshot strategy profile is set to be an equilibrium, the corresponding stationary point also becomes an equilibrium. However, it is virtually impossible to identify such an ideal slingshot strategy profile beforehand.
To this end, APMD adjusts a slingshot strategy profile $\sigma$ by replacing it with the (nearly) stationary point that is reached after predefined iterations $T_{\sigma}$.

Figure~\ref{fig:exploitability_compare_pi_sigma} illustrates how our slingshot strategy update brings the corresponding stationary points closer to an equilibrium.
x- and y-axis indicate the number of iterations and the logarithm of the gap function of the last-iterate strategy profile.
We here assume that the learning rate and the perturbation strength are $\eta=0.1$ and $\mu=1$, respectively.
The initial slingshot strategy $\sigma_i$ is given as a uniform distribution on the action space.
After the first interval of $T_{\sigma}=200$ iterations, APMD finds a stationary point. The slingshot strategy profile for the second interval is replaced with the stationary point. 
Figure~\ref{fig:exploitability_compare_pi_sigma} clearly shows the stationary point, i.e., the last-iterate strategy profile comes close to an equilibrium every time the slingshot strategy profile is updated. 
We also theoretically justify our slingshot strategy update in Theorem \ref{thm:exploitability_upper_bound}, i.e., when the slingshot strategy profile is close to an equilibrium $\pi^{\ast}\in \Pi^{\ast}$, the stationary point is close to $\pi^{\ast}$.

\section{Last-Iterate Convergence Rates}
\label{sec:last_iterate_convergence_rates}
In this section, we establish the last-iterate convergence rates of APMD.
More specifically, we examine a setting where both $D_{\psi}$ and $G$ is set to the squared $\ell^2$-distance, i.e., $D_{\psi}(\pi_i, \pi_i') = G(\pi_i, \pi_i') = \frac{1}{2}\|\pi_i - \pi_i'\|^2$.
This instance can be considered as an extended version of Gradient Descent, which incorporates our techniques of payoff perturbation and slingshot strategy update.
We also assume that the gradient vector of $v_i$ is bounded.
We emphasize that we have obtained the overall last-iterative convergence rates of APMD for the entire $T$ iterations in both full and noisy feedback settings.

\subsection{Full Feedback Setting}
First, we demonstrate the last-iterate convergence rate of APMD with {\it full feedback} where each player receives the perfect gradient vector $\widehat{\nabla}_{\pi_i}v_i(\pi_i^t, \pi_{-i}^t)=\nabla_{\pi_i}v_i(\pi_i^t, \pi_{-i}^t)$, at each iteration $t$.
Theorem~\ref{thm:lic_rate_full} provides the APMD's convergence rate of $\mathcal{O}(\ln T/\sqrt{T})$ in terms of the gap function. Note that the learning rate is constant, and its upper bound is specified by perturbation strength $\mu$ and smoothness parameter $L$.
\begin{theorem}
\label{thm:lic_rate_full}
If we use the constant learning rate $\eta_t = \eta \in (0, \frac{2\mu}{3\mu^2 + 8L^2})$, and set $D_{\psi}$ and $G$ as the squared $\ell^2$-distance $D_{\psi}(\pi_i, \pi_i') = G(\pi_i, \pi_i') = \|\pi_i - \pi_i'\|^2/2$, and set $T_{\sigma} = \Theta(\ln T)$, then the strategy profile $\pi^T$ updated by APMD satisfies:
\begin{align*}
    \mathrm{GAP}(\pi^T) &= \mathcal{O}\left(\frac{\ln T}{\sqrt{T}}\right).
\end{align*}
\end{theorem}
The obtained rate herein is competitive with optimistic gradient and extra-gradient methods~\citep{cai2022finite} whose rates are $\mathcal{O}(1/\sqrt{T})$. Although it is open whether our convergence rate matches its lower bound, it closely aligns with the lower bound for the class of algorithms that includes $1$-SCLI algorithms~\citep{golowich2020last}, which is different from our APMD. 
To the best of our knowledge, the fastest rate of $\mathcal{O}(1/T)$ is achieved by \textit{Accelerated Optimistic Gradient} (AOG) \citep{cai2023doubly}, which is an optimistic variant of the \textit{Halpern iteration} \citep{bams/1183529119}. 
At first sight, the update rule of AOG looks as if the perturbation strength was linearly decayed. However, it does not perturb the payoff functions and instead adjusts the regularization term by using the convex combination of the current strategy and the anchoring strategy as the proximal point in MD. Unlike our APMD, the anchoring strategy is never updated through the iterations.\footnote{Regarding the rate of $\mathcal{O}(1/T)$, a companion paper is in preparation.}

\subsection{Proof Sketch of Theorem~\ref{thm:lic_rate_full}}
This section sketches the proof for Theorem \ref{thm:lic_rate_full}.
We present the complete proofs for the theorem and associated lemmas in Appendix~\ref{sec:appx_exact_nash_conv}.

\paragraph{(1) Convergence rates to a stationary point with $k$-th slingshot strategy profile.}
We denote $\sigma^k$ as the slingshot strategy profile after $k$ updates.
Since the slingshot strategy profile is overwritten by the current strategy profile $\pi^t$ every $T_{\sigma}$ iterations, we can write $\sigma^k = \pi^{kT_{\sigma}}$.
We first prove that, as $T_{\sigma}$ increases, the next $k+1$-th slingshot strategy profile $\sigma^{k+1}$ approaches to the stationary point $\pi^{\mu, \sigma^k}$ under the slingshot strategy profile $\sigma^k$, which satisfies the following condition: $\forall i\in [N]$,
\begin{align}
\label{eq:perturbed_nash}
     \pi_i^{\mu,\sigma^k} &= \argmax_{\pi_i\in \mathcal{X}_i} \left\{ v_i(\pi_i, \pi_{-i}^{\mu,\sigma^k}) - \mu G(\pi_i, \sigma_i^k)\right\}.
\end{align}
Note that $\pi^{\mu,\sigma^k}$ always exists since the perturbed game is still monotone.
Using the strong convexity of $G(\pi_i, \sigma_i^k) = \frac{1}{2}\|\pi_i - \sigma_i^k\|^2$, we show that $\sigma^{k+1}\to \pi^{\mu, \sigma^k}$ as $T_{\sigma}\to \infty$:
\begin{lemma}
\label{lem:gd_sp_full}
Assume that $D_{\psi}$ and $G$ are set as the squared $\ell^2$-distance.
If we use the constant learning rate $\eta_t = \eta \in (0, \frac{2\mu}{3\mu^2 + 8L^2})$, the $k+1$-th slingshot strategy profile $\sigma^{k+1}$ of APMD under the full feedback setting satisfies that:
\begin{align*}
    \|\pi^{\mu,\sigma^k} -  \sigma^{k+1}\|^2 = \|\pi^{\mu,\sigma^k} - \sigma^k\|^2\cdot \exp\left(-\mathcal{O}(T_{\sigma})\right).
\end{align*}
\end{lemma}

\paragraph{(2) Upper bound on the gap function.}
Next, we derive the upper bound on the gap function for $\sigma^{k+1}$.
From the bounding technique of the gap function using the tangent residuals by \citet{cai2022finite} and the first-order optimality condition for $\pi^{\mu, \sigma^k}$, the gap function for $\sigma^{k+1}$ can be upper bounded by the distance between the slingshot strategy profile and the stationary point $\pi^{\mu, \sigma^k}$:
\begin{lemma}
\label{lem:gap_fn_slingshot_strategy}
If $G$ is set as the squared $\ell^2$-distance, the gap function for $\sigma^{k+1}$ of APMD satisfies for $k\geq 0$:
\begin{align*}
\mathrm{GAP}(\sigma^{k+1}) = \mathcal{O}\left( \|\pi^{\mu, \sigma^k} - \sigma^{k+1}\| + \|\pi^{\mu, \sigma^k} - \sigma^k\|\right).
\end{align*}
\end{lemma}
By Lemma \ref{lem:gd_sp_full}, if we set $D_{\psi}$ as the squared $\ell^2$-distance, the first term in this lemma can be bounded as:
\begin{align}
\label{eq:gm_gd_full}
\|\pi^{\mu,\sigma^k} - \sigma^{k+1}\|^2 = \|\pi^{\mu,\sigma^k} -  \sigma^k\|^2\exp\left(-\mathcal{O}(T_{\sigma})\right).
\end{align}
Therefore, it is enough to derive the convergence rate on the $\ell^2$-distance between $\pi^{\mu, \sigma^k}$ and $\sigma^k$.

\paragraph{(3) Last-iterate convergence results for the slingshot strategy profile.}
Let us denote $K:= \lfloor T / T_{\sigma}\rfloor$ as the total number of the slingshot strategy updates over the entire $T$ iterations.
Then, by adjusting $T_{\sigma} = \Omega(\ln T)$, we show that the $\ell^2$-distance between the $K-1$-th slingshot strategy profile $\sigma^{K-1}$ and the corresponding stationary point $\pi^{\mu, \sigma^{K-1}}$ decreases as $K$ increases:
\begin{lemma}
\label{lem:slingshot_diff_full}
In the same setup of Theorem \ref{thm:lic_rate_full}, the $K-1$-th slingshot strategy profile $\sigma^{K-1}$ of APMD satisfies:
\begin{align*}
\|\pi^{\mu,\sigma^{K-1}} - \sigma^{K-1}\| = \mathcal{O}(1/\sqrt{K}).
\end{align*}
\end{lemma}
Lemma \ref{lem:slingshot_diff_full} implies that as $K$ increases, the variation of $\sigma^K$ becomes negligible, signifying convergence in its behavior.

By combining \eqref{eq:gm_gd_full} and Lemmas \ref{lem:gap_fn_slingshot_strategy}-\ref{lem:slingshot_diff_full}, we can derive the last-iterate convergence rate of the slingshot strategy profile $\sigma^K$:
\begin{align*}
\mathrm{GAP}(\sigma^K) = \mathcal{O}(1/\sqrt{K}).
\end{align*}
Thus, since $\pi^T = \sigma^K$ and $K=\lfloor T / T_{\sigma}\rfloor= \Theta(T / \ln T)$, the statement of the theorem is concluded. \qed

\subsection{Noisy Feedback Setting}
Next, we consider the {\it noisy feedback} setting, where each player $i$ receives a gradient vector with additive noise: $\nabla_{\pi_i}v_i(\pi_i^t, \pi_{-i}^t) + \xi^t_i$.
Define the sigma-algebra generated by the history of the observations:
$\mathcal{F}_t = \sigma\left((\widehat{\nabla}_{\pi_i}v_i(\pi_i^0, \pi_{-i}^0))_{i \in [N]}, \ldots, ( \widehat{\nabla}_{\pi_i}v_i(\pi_i^{t-1}, \pi_{-i}^{t-1}))_{i \in [N]}\right)$, $\forall t\ge1$.
We assume that the noise vectors $(\xi_i^t)_{t\geq 1}$ are with zero-mean and bounded variances.
We also suppose that the noise vectors $(\xi_i^t)_{t\geq 1}$ are independent over $t$.
In this setting, the last-iterate convergence rate is achieved by APMD using a decreasing learning rate sequence $\eta_t$.
The convergence rate obtained by APMD is $\mathcal{O}(\ln T/T^{\frac{1}{10}})$:
\begin{theorem}
\label{thm:lic_rate_noisy}
Let $\theta =  \frac{3\mu^2 + 8 L^2}{2 \mu} $ and $\kappa = \frac{\mu}{2}$.
Assume that $D_{\psi}$ and $G$ are set as the squared $\ell^2$-distance $D_{\psi}(\pi_i, \pi_i') = G(\pi_i, \pi_i') = \frac{1}{2}\|\pi_i - \pi_i'\|^2$, and $T_{\sigma} = \Theta(T^{4/5})$.
If we choose the learning rate sequence of the form $\eta_t = 1/(\kappa (t - T_{\sigma} \cdot \lfloor t / T_{\sigma} \rfloor) + 2\theta)$, then the strategy profile $\pi^T$ updated by APMD satisfies:
\begin{align*}
    \mathbb{E}\left[\mathrm{GAP}(\pi^T)\right] &= \mathcal{O}\left(\frac{\ln T}{T^{\frac{1}{10}}}\right).
\end{align*}
\end{theorem}
It should be noted that Theorem \ref{thm:lic_rate_noisy} provides a non-asymptotic convergence guarantee with a rate.
This is a significant departure from the existing convergence results \citep{koshal2010single,koshal2013regularized,tatarenko2019learning}, which focus on the asymptotic convergence of iterative Tikhonov regularization methods in the noisy or bandit feedback settings.

\subsection{Proof Sketch of Theorem \ref{thm:lic_rate_noisy}}
As in Theorem \ref{thm:lic_rate_full}, we first derive the convergence rate of $\sigma^{k+1}$ for the noisy feedback setting:
\begin{lemma}
\label{lem:gd_sp_noisy}
Let $\theta =  \frac{3\mu^2 + 8L^2}{2 \mu} $ and $\kappa = \frac{\mu}{2}$.
Suppose that both $D_{\psi}$ and $G$ are defined as the squared $\ell^2$-distance.
Under the noisy feedback setting, if we use the learning rate sequence of the form $\eta_t = 1/(\kappa (t - T_{\sigma} \cdot \lfloor t / T_{\sigma}\rfloor) + 2\theta)$, the $k+1$-th slingshot strategy profile $\sigma^{k+1}$ of APMD for each $k\geq 0$ satisfies that:
\begin{align*}
     &\mathbb{E}\left[\|\pi^{\mu,\sigma^k} - \sigma^{k+1}\|^2\right] = \mathcal{O}\left(\frac{\ln T_{\sigma}}{T_{\sigma}}\right).
\end{align*}
\end{lemma}
The proof is given in Appendix~\ref{sec:appx_gd_sp_noisy} and is based on the standard argument of stochastic optimization, e.g., \citet{nedic2014stochastic}.
However, the proof is made possible by taking into account the monotonicity of the game and the relative (strong and smooth) convexity of the divergence function.

Next, in a similar manner for Lemma \ref{lem:slingshot_diff_full}, we show the upper bound on the expected distance between $\pi^{\mu, \sigma^{K-1}}$ and $\sigma^{K-1}$ under the noisy feedback setting:
\begin{lemma}
In the same setup of Theorem \ref{thm:lic_rate_noisy}, the $K-1$-th slingshot strategy profile $\sigma^{K-1}$ satisfies:
\label{lem:slingshot_diff_noisy}
\begin{align*}
\mathbb{E}\left[\|\pi^{\mu,\sigma^{K-1}} - \sigma^{K-1}\|\right] = \mathcal{O}\left(\frac{\ln K}{\sqrt{K}}\right).
\end{align*}
\end{lemma}

We note that Lemma~\ref{lem:gap_fn_slingshot_strategy} holds for any combination of $\sigma^k$ and $\sigma^{k+1}$, regardless of the presence of noise.
By combining Lemmas \ref{lem:gap_fn_slingshot_strategy}, \ref{lem:gd_sp_noisy}, and \ref{lem:slingshot_diff_noisy}, we can derive the following last-iterate convergence rate of $\sigma^K$ in terms of the gap function:
\begin{align*}
    \mathbb{E}\left[\mathrm{GAP}(\sigma^K)\right] = \mathcal{O}\left(\frac{\ln K}{\sqrt{K}}\right).
\end{align*}
By setting $K=\lfloor T/T_{\sigma} \rfloor = \Theta(T^{1/5})$, we conclude the statement of the theorem. \qed

\section{Beyond Squared $\ell^2$-Payoff Perturbation}
\label{sec:discussion}

Section \ref{sec:last_iterate_convergence_rates} assumes that both $G$ and $D_{\psi}$ are the squared $\ell^2$-distance.
This section considers more general choices of $G$ and $D_{\psi}$.

\subsection{Instantiation of Payoff-Perturbed Algorithms}
First, we would like to emphasize that choosing appropriate combinations of $G$ and $D_{\psi}$ enables APMD to reproduce some existing learning algorithms that incorporate payoff perturbation.
For example, the following learning algorithms can be viewed as instantiations of APMD.

\begin{example}[Boltzmann Q-Learning \citep{Tuyls2006}]
\label{exm:boltzman_q}
\normalfont
Assume that $\mathcal{X}_i=\Delta^{d_i}$, the regularize is entropy: $\psi(\pi_i)=\sum_{j=1}^{d_i}\pi_{ij}\ln \pi_{ij}$.
Let us set $G$ as the KL divergence and the slingshot strategy $\sigma_i$ as a uniform distribution, i.e., $G(\pi_i, \sigma_i) = \mathrm{KL}(\pi_i, \sigma_i)$ and $\sigma_i = (1/d_i)_{j \in [d_i]}$.
Then, the corresponding continuous-time APMD dynamics can be expressed as:
\begin{align*}
    \frac{d}{dt}\pi_{ij}^t &= \pi_{ij}^t\left( q_{ij}^{\pi^t} - \sum_{k=1}^{d_i}\pi_{ik}^tq_{ik}^{\pi^t}\right) \\
    &\phantom{=} - \mu \pi_{ij}^t\!\left(\ln \pi_{ij}^t - \sum_{k=1}^{d_i}\pi_{ik}^t\ln \pi_{ik}^t\right),
\end{align*}
which is equivalent to Boltzman Q-learning \citep{Tuyls2006, Bloembergen2015}.
\end{example}

\begin{example}[Reward transformed FTRL \citep{perolat2021poincare}]
\label{exm:reward_transformed}
\normalfont
Consider the continuous-time APMD dynamics where $N=2$, $\mathcal{X}_i=\Delta^{d_i}$, $\psi$ is Legendre \citep{rockafellar1997convex,lattimore2020bandit} with $\mathrm{dom}~\psi\subseteq \mathcal{X}_i$, and $G(\pi_i, \sigma_i) = \mathrm{KL}(\pi_i, \sigma_i)$.
Then, APMD dynamics can be described as follows:
\begin{align*}
    \pi_{ij}^t &= \argmax_{x\in \Delta^{d_i}} \left\{ \int_0^t\left(\sum_{k=1}^{d_i}x_{k}q_{ik}^{\pi^s} - \mu\sum_{k=1}^{d_i}x_{k} \ln \frac{\pi_{ik}^s}{\sigma_{ik}}\right.\right. \\
    &\left.\left. \phantom{======} + \mu\sum_{k=1}^{d_{-i}}\pi_{-ik}^s \ln \frac{\pi_{-ik}^s}{\sigma_{-ik}}\right)ds - \psi(x)\right\}.
\end{align*}
This algorithm is equivalent to FTRL with reward transformation \citep{perolat2021poincare}.
\end{example}

\begin{example}[Mutant FTRL \citep{abe2022mutationdriven}]
\label{exm:mftrl}
\normalfont
Let us define $\mathcal{X}_i=\Delta^{d_i}$, and assume that the regularizer $\psi$ is  Legendre with $\mathrm{dom}~\psi\subseteq \mathcal{X}_i$.
If we set $G$ as the reverse KL divergence, i.e., $G(\pi_i, \sigma_i) = \mathrm{KL}(\sigma_i, \pi_i) =  \sum_{j=1}^{d_i}\sigma_{ij}\ln \frac{\sigma_{ij}}{\pi_{ij}}$, we can rewrite \eqref{eq:gm-md} as:
\begin{align*}
    \pi_i^{t+1} \!\!=\! \argmax_{x\in \Delta^{d_i}}\!\Bigg\{\!\! \sum_{s=0}^t \eta \!\sum_{j=1}^{d_i} \!x_j\!\!\left(\!q_{ij}^{\pi^s} \!\!+\! \frac{\mu}{\pi_{ij}^s}\!\left(\sigma_{ij} \!-\! \pi_{ij}^s\right)\!\!\right) \!\!-\! \psi(x)\!\Bigg\}\!,
\end{align*}
where $q_{ij}^{\pi^s} = (\widehat{\nabla}_{\pi_i} v_i(\pi_i^s, \pi_{-i}^s))_j$.
This algorithm is equivalent to Mutant FTRL \citep{abe2022mutationdriven}.
\end{example}

\subsection{Convergence Results with General $G$ and $D_{\psi}$}
Next, we establish the convergence results for APMD with general combinations of $G$ and $D_{\psi}$.
For theoretical analysis, we assume a specific condition on $G$:
\begin{assumption}
\label{asm:rel_smooth}
$G(\cdot, \sigma_i)$ is differentiable over $\mathrm{int}(\mathrm{dom}~\psi)$.
Moreover, $G(\cdot, \sigma_i)$ is $\beta$-smooth and $\gamma$-strongly convex relative to $\psi$, i.e., for any $\pi_i, \pi_i'\in \mathrm{int}(\mathrm{dom}~\psi)$, $\gamma D_{\psi}(\pi_i', \pi_i) \leq G(\pi_i', \sigma_i) - G(\pi_i, \sigma_i) - \langle \nabla_{\pi_i}G(\pi_i, \sigma_i), \pi_i' - \pi_i\rangle\leq  \beta D_{\psi}(\pi_i', \pi_i)$ holds.
\end{assumption}

Note that these assumptions are always satisfied with $\beta=\gamma=1$ whenever $G$ is identical to $D_{\psi}$; thus, these are not strong assumptions. 
We also assume that $\pi^t$ is well-defined over iterations:
\begin{assumption}
\label{asm:well_defined}
$\pi^t$ updated by APMD satisfies $\pi^t\in \mathrm{int}(\mathrm{dom}~\psi)$ for any $t\in \{0, 1, \cdots, T\}$.
\end{assumption}

Using Assumptions \ref{asm:rel_smooth} and \ref{asm:well_defined}, we derive the convergence rate to $\pi^{\mu, \sigma^k}$ in \eqref{eq:perturbed_nash}:
\begin{lemma}
\label{lem:md_sp_full}
Suppose that Assumptions \ref{asm:rel_smooth} with $\beta,\gamma \in (0, \infty)$ and \ref{asm:well_defined} hold.
If we use the constant learning rate $\eta_t = \eta \in (0, \frac{2\mu \gamma \rho^2}{\mu^2 \gamma \rho^2 (\gamma + 2\beta) + 8L^2})$, the $k+1$-th slingshot strategy profile $\sigma^{k+1}$ of APMD under the full feedback setting satisfies that:
\begin{align*}
    D_{\psi}(\pi^{\mu,\sigma^k}, \sigma^{k+1}) = D_{\psi}(\pi^{\mu,\sigma^k}, \sigma^k)\exp\left(-\mathcal{O}(T_{\sigma})\right).
\end{align*}
\end{lemma}
\begin{lemma}
\label{lem:md_sp_noisy}
Let $\theta =  \frac{\mu^2 \gamma \rho^2 (\gamma + 2\beta) + 8 L^2}{2 \mu \gamma \rho^2} $ and $\kappa = \frac{\mu\gamma}{2}$.
Suppose that Assumptions \ref{asm:rel_smooth} with $\beta,\gamma \in (0, \infty)$ and \ref{asm:well_defined} hold, and the learning rate sequence of the form $\eta_t = 1/(\kappa (t - T_{\sigma} \cdot \lfloor t / T_{\sigma}\rfloor) + 2\theta)$ is used.
Then, the $k+1$-th slingshot strategy profile $\sigma^{k+1}$ of APMD under the noisy feedback setting satisfies that:
\begin{align*}
     &\mathbb{E}\left[D_{\psi}(\pi^{\mu,\sigma^k}, \sigma^{k+1})\right] = \mathcal{O}\left(\frac{\ln T_{\sigma}}{T_{\sigma}}\right).
\end{align*}
\end{lemma}

\begin{figure*}[t!]
    \centering
    \includegraphics[width=1.0\textwidth]{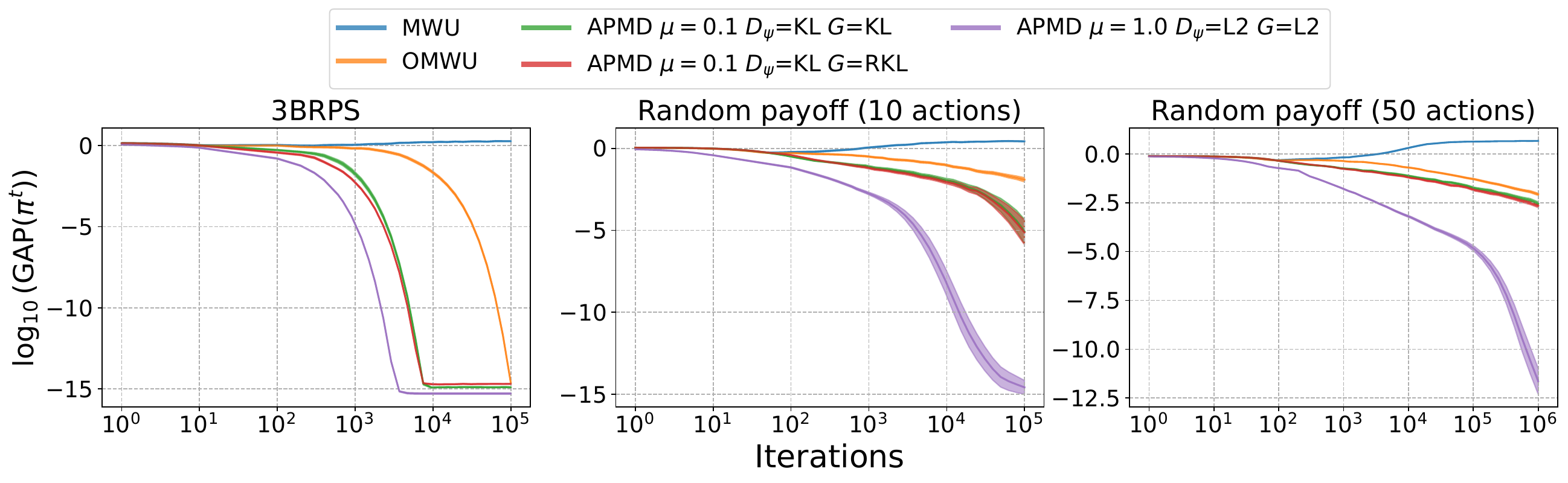}
    \caption{
    The gap function for $\pi^t$ for APMD, MWU, and OMWU with full feedback.
    The shaded area represents the standard errors. Note that the KL divergence, reverse KL divergence, and squared $\ell^2$-distance are abbreviated to KL, RKL, and L2, respectively.
    }
    \label{fig:exploitability_full}
\end{figure*}
\begin{figure*}[t!]
    \centering
    \includegraphics[width=1.0\textwidth]{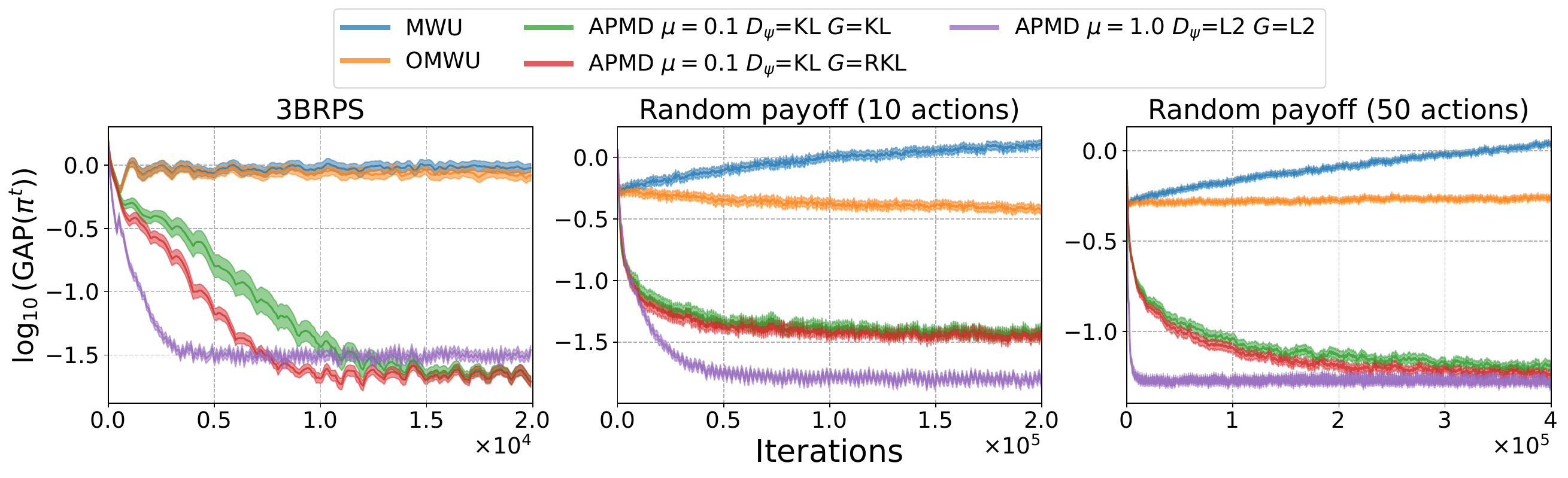}
    \caption{
    The gap function for $\pi^t$ for APMD, MWU, and OMWU with noisy feedback.
    }
    \label{fig:exploitability_noisy}
\end{figure*}

These lemmas imply that $\sigma^{k+1}\to \pi^{\mu, \sigma^k}$ as $T_{\sigma} \to \infty$.
Therefore, when $T_{\sigma}$ is sufficiently large, $k+1$-th slingshot strategy profile becomes almost equivalent to the stationary point $\pi^{\mu, \sigma^k}$.
From this, we anticipate that $\mathbb{E}[\mathrm{GAP}(\sigma^k)]\to 0$ as $k\to \infty$ even in the noisy feedback setting.
The subsequent theorems provide asymptotic last-iterate convergence results for this ideal scenario.
In particular, we show that the slingshot strategy profile $\sigma^{k+1}$ converges to equilibrium when using the following divergence functions as $G$: 1) Bregman divergence; 2) $\alpha$-divergence; 3) R\'{e}nyi-divergence; 4) reverse KL divergence.
\begin{theorem}
\label{thm:exact_conv_bregman}
Assume that $G$ is a Bregman divergence $D_{\psi'}$ for some strongly convex function $\psi'$, and $\sigma^{k+1}=\pi^{\mu, \sigma^k}$ for $k\geq 0$.
Then, there exists $\pi^{\ast}\in \Pi^{\ast}$ such that $\sigma^k\to \pi^{\ast}$ as $k\to \infty$.
\end{theorem}

\begin{theorem}
\label{thm:exact_conv_alpha}
Let us define $\mathcal{X}_i = \Delta^{d_i}$.
Assume that $\sigma^{k+1}=\pi^{\mu, \sigma^k}$ for $k\geq 0$, and $G$ is one of the following divergence: 1) $\alpha$-divergence with $\alpha\in (0, 1)$; 2) R\'{e}nyi-divergence with $\alpha\in (0, 1)$; 3) reverse KL divergence.
If the initial slingshot strategy profile $\sigma^0$ is in the interior of $\mathcal{X}$, the sequence $\{\sigma^k\}_{k\geq 1}$ converges to the set of Nash equilibria $\Pi^{\ast}$ of the underlying game.
\end{theorem}

We remark that these results cover the algorithms in Example \ref{exm:boltzman_q}, \ref{exm:reward_transformed}, and \ref{exm:mftrl}.
Furthermore, we can incorporate our payoff perturbation techniques into FTRL, detailed in Appendix \ref{sec:appx_extension_to_ftrl}.

\section{Experiments}
\label{sec:experiments}
This section empirically compares the representative instance of MD, namely Multiplicative Weight Update (MWU) and its Optimistic version (OMWU), with our framework.
Specifically, we consider the following three instances of APMD: (i) both are the squared $\ell^2$-distance; (ii) both $G$ and $D_{\psi}$ are the KL divergence, which is also an instance of Reward transformed FTRL in Example~\ref{exm:reward_transformed}. Note that if the slingshot strategy is fixed to a uniformly random strategy, this algorithm corresponds to Boltzmann Q-Learning in Example~\ref{exm:boltzman_q}; (iii) the divergence function $G$ is the reverse KL divergence, and the Bregman divergence $D_{\psi}$ is the KL divergence, which matches Mutant FTRL in Example~\ref{exm:mftrl}. 

We focus on two zero-sum polymatrix games: Three-Player Biased Rock-Paper-Scissors (3BRPS) and three-player random payoff games with $10$ and $50$ actions.
For the 3BRPS game, each player participates in two instances of the game in Table~\ref{tab:biased-rps} in Appendix \ref{sec:appx_experimental_detail} simultaneously with two other players.
For the random payoff games, each player $i$ participates in two instances of the game with two other players $j$ simultaneously.
The payoff matrix for each instance is denoted as $M^{(i,j)}$. Each entry of $M^{(i,j)}$ is drawn independently from a uniform distribution on the interval $[-1,1]$.

Figures~\ref{fig:exploitability_full} and \ref{fig:exploitability_noisy} illustrate the logarithm of the gap function averaged over $100$ instances with different random seeds.
We assume that the initial slingshot strategy profile $\pi^0$ is chosen uniformly at random in the interior of the strategy space $\mathcal{X}=\prod_{i=1}^3 \Delta^{d_i}$ in each instance for 3BRPS, 
while $\pi^0$ is chosen as $(1/d_i)_{j\in [d_i]}$ for $i\in [3]$ in every instances for the random payoff games. 

First, Figure~\ref{fig:exploitability_full} depicts the case of full feedback. Unless otherwise specified, we use a constant learning rate $\eta=0.1$ and a perturbation strength $\mu=0.1$ for APMD. Further details and additional experiments can be found in Appendix~\ref{sec:appx_experimental_detail}. Figure \ref{fig:exploitability_full} shows that APMD outperforms MWU and OMWU in all three games. %
Notably, APMD exhibits the fastest convergence in terms of the gap function when using the squared $\ell^2$-distance as both $G$ and $D_{\psi}$.
Next, Figure~\ref{fig:exploitability_noisy} depicts the case of noisy feedback. We assume that the noise vector $\xi_i^t$ is generated from the multivariate Gaussian distribution $\mathcal{N}(0,~ 0.1^2\mathbf{I})$ in an i.i.d. manner. To account for the noise, we use a lower learning rate $\eta=0.01$ than the full feedback case. In OMWU, we use the noisy gradient vector $\widehat{\nabla}_{\pi_i}v_i(\pi_i^{t-1}, \pi_{-i}^{t-1})$ at the previous step $t-1$ as the prediction vector for the current iteration $t$. We observe the same trends as with full feedback. While MWU and OMWU exhibit worse performance, APMD maintains its fast convergence, as predicted by the theoretical analysis.

\section{Related Literature}
Recent progress in achieving no-regret learning with full feedback has been driven by optimistic learning~\citep{rakhlin2013online,rakhlin2013optimization}. 
Optimistic versions of well-known algorithms like Follow the Regularized Leader~\citep{shalev2006convex} and Mirror Descent~\citep{zhou2017mirror,hsieh2021adaptive} have been proposed to admit last-iterate convergence in a wide range of game settings. These optimistic algorithms have been successfully applied to various classes of games, including bilinear games \citep{daskalakis2017training, daskalakis2018last, liang2019interaction, de2022convergence}, cocoercive games \citep{lin2020finite}, and saddle point problems \citep{daskalakis2018limit, mertikopoulos2018optimistic, golowich2020last, wei2020linear, lei2021last, yoon2021accelerated, Lee2021FastExtraGrad, cevher2023min}. The advancements have provided solutions to monotone games and have established convergence rates~\citep{golowich2020tight, cai2022finite, cai2022tight, gorbunov2022last, cai2023doubly}. 

The exploration of literature with noisy feedback poses significant challenges, in contrast to full feedback.
In situations where feedback is imprecise or limited, algorithms must estimate action values at each iteration.
There have been two trends in achieving last-iterate convergence: restricting the class of games and perturbing the payoff functions. 
On one hand, particularly noticeable works lie in potential games~\citep{Heliou2017Learning}, normal-form games with strict Nash equilibria~\citep{giannou2021convergence,giannou2021Survival}, and two-player zero-sum games~\citep{abe2022last}. Also, noisy feedback is handled with games whose payoff functions are assumed to be strictly (or strongly) monotone ~\citep{bravo2018bandit,kannan2019optimal,hsieh2019ontheconvergence,Anagnostides2022frequency}, while to be strictly variational stable~\citep{mertikopoulos2019learning,mertikopoulos2018optimistic,mertikopoulos2022learning,azizian2021LastIterate}.
Note that variationally stable games, often referred to in control theory, are a slightly broader class of monotone games.
These studies require the payoff functions to be strictly or strongly convex.
When these restrictions are not imposed, convergence is primarily guaranteed only in an asymptotic sense, and the rate is not quantified \citep{hsieh2020explore,hsieh2022no}.

On the other hand, payoff-perturbed algorithms have recently regained attention for their ability to demonstrate convergence in unrestricted games when noise is present. As described in Section \ref{sec:introduction}, payoff perturbation is a textbook technique \citep{facchinei2003finite} that has been extensively studied~\citep{koshal2010single,koshal2013regularized,yousefian2017smoothing,tatarenko2019learning,abe2022last,cen2021fast, cen2022faster, cai2023uncoupled, pmlr-v206-pattathil23a}. It is known that carefully adjusting the magnitude of perturbation ensures convergence to a Nash equilibrium. This magnitude is computed as the product of a strongly convex penalty and a perturbation strength parameter. \citet{liu2022power} shrink the perturbation strength based on a predefined hyper-parameter and the gap function of the current strategy. Likewise, \citet{koshal2010single} and \citet{tatarenko2019learning} have identified somewhat complex conditions that the sequence of the perturbation strength parameters and learning rates should satisfy. Roughly speaking, as we have implied in Lemma~\ref{lem:gd_sp_full}, a smaller strength would require a lower learning rate. This potentially decelerates the convergence rate and complicates the task of finding an appropriate learning rate. For practicality, we have opted to keep the perturbation strength constant, independent of the iteration in APMD. Moreover, it must be emphasized that the existing literature has primarily provided asymptotic convergence results, while we have successfully provided the non-asymptotic convergence rate.

Finally, the idea of the slingshot strategy update was initiated by \citet{perolat2021poincare} and later extended by \citet{abe2022last}. Our contribution partly owes to the significance of quantifying the convergence rate for the first time. We must also mention that \citet{sokota2022unified} have proposed a very similar, but essentially different update rule to ours. It just adds an additional regularization term based on an anchoring strategy, which they call a \textit{magnetic} one, and this means that it directly perturbs the (expected) payoff functions. In contrast, our APMD indirectly perturbs the payoff functions, i.e., perturbs the gradient vector.
Furthermore, we have established non-asymptotic convergence results toward a Nash equilibrium, while \citet{sokota2022unified} have only shown convergence toward a quantal response equilibrium \citep{mckelvey1995quantal,mckelvey1998quantal}, which is just equivalent to an approximate equilibrium. 
Similar results to them have been obtained with the Boltzmann Q-learning dynamics \citep{Tuyls2006} and penalty-regularized dynamics \citep{coucheney2015penalty} in continuous-time settings \citep{leslie2005individual,hussain2023asymptotic}.

\section{Conclusion}
This paper proposes a novel variant of MD that achieves last-iterate convergence even when the noise is present, by adaptively adjusting the magnitude of the perturbation. 
This research could lead to several intriguing future studies, such as finding the best perturbation strength for the optimal convergence rate and achieving convergence with more limited feedback, for example, using bandit feedback \citep{bravo2018bandit,drusvyatskiy2022improved}.

\section*{Acknowledgments}
Kaito Ariu is supported by JSPS KAKENHI Grant Number 23K19986.
Atsushi Iwasaki is supported by JSPS KAKENHI Grant Numbers 21H04890 and 23K17547.

\section*{Impact Statement}
This paper focuses on discussing the problem of computing equilibria in games. 
There are potential social implications associated with our work, but none that we believe need to be particularly emphasized here.

\bibliography{references}
\bibliographystyle{icml2024}

\newpage
\appendix
\onecolumn
\tableofcontents

\section{Notations}
\label{sec:notations}
In this section, we summarize the notations we use in Table~\ref{tb:notation}.
\begin{table}[h!]
    \centering
    \caption{Notations}
    \label{tb:notation}
    \begin{tabular}{cc} \hline
        Symbol & Description \\ \hline
        $N$ & Number of players \\
        $\mathcal{X}_i$ & Strategy space for player $i$ \\ 
        $\mathcal{X}$ & Joint strategy space: $\mathcal{X}=\prod_{i=1}^N \mathcal{X}_i$ \\
        $v_i$ & Payoff function for player $i$ \\
        $\pi_i$ & Strategy for player $i$ \\
        $\pi$ & Strategy profile: $\pi=(\pi_i)_{i\in [N]}$ \\
        $\xi_i^t$ & Noise vector for player $i$ at iteration $t$ \\
        $\pi^{\ast}$ & Nash equilibrium \\
        $\Pi^{\ast}$ & Set of Nash equilibria \\
        $\mathrm{GAP}(\pi)$ & Gap function of $\pi$: $\mathrm{GAP}(\pi)=\max_{\tilde{\pi}\in \mathcal{X}}\sum_{i=1}^N\langle \nabla_{\pi_i}v_i(\pi), \tilde{\pi}_i - \pi_i\rangle$ \\
        $\Delta^d$ & $d$-dimensional probability simplex: $\Delta^d=\{p\in [0,1]^d ~|~ \sum_{j=1}^d p_j=1\}$ \\
        $\mathrm{diam}(\mathcal{X})$ & Diameter of $\mathcal{X}$: $\mathrm{diam}(\mathcal{X})=\sup_{\pi, \pi'\in \mathcal{X}}\|\pi - \pi'\|$ \\
        $\mathrm{KL}(\cdot, \cdot)$ & Kullback-Leibler divergence \\
        $D_{\psi}(\cdot, \cdot)$ & Bregman divergence associated with $\psi$ \\
        $\nabla_{\pi_i}v_i$ & Gradient of $v_i$ with respect to $\pi_i$ \\
        $\eta_t$ & Learning rate at iteration $t$ \\
        $\mu$ & Perturbation strength \\
        $\sigma$ & Slingshot strategy profile \\
        $G(\cdot,\cdot)$ & Divergence function for payoff perturbation \\
        $\nabla_{\pi_i}G$ & Gradient of $G$ with respect to first argument \\
        $T_{\sigma}$ & Update interval for the slingshot strategy \\
        $K$ & Total number of the slingshot strategy updates \\
        $\pi^{\mu,\sigma}$ & Stationary point satisfies \eqref{eq:perturbed_nash} for given $\mu$ and $\sigma$ \\
        $\pi^t$ & Strategy profile at iteration $t$ \\
        $\sigma^k$ & Slingshot strategy profile after $k$ updates \\
        $L$ & Smoothness parameter of $(v_i)_{i\in [N]}$ \\
        $\rho$ & Strongly convex parameter of $\psi$ \\
        $\beta$ & Smoothness parameter of $G(\cdot, \sigma_i)$ relative to $\psi$ \\
        $\gamma$ & Strongly convex parameter of $G(\cdot, \sigma_i)$ relative to $\psi$ \\ \hline
    \end{tabular}
\end{table}
\vfill

\section{Formal Theorems and Lemmas}

\subsection{Full Feedback Setting}
\begin{theorem}[Formal version of Theorem \ref{thm:lic_rate_full}]
\label{thm:formal_lic_rate_full}
Assume that $\sqrt{\sum_{i=1}^N \|\nabla_{\pi_i} v_i(\pi)\|^2} \leq \zeta$ for any $\pi \in \mathcal{X}$.
If we use the constant learning rate $\eta_t = \eta \in (0, \frac{2\mu \rho^2}{3\mu^2\rho^2 + 8L^2})$, and set $D_{\psi}$ and $G$ as squared $\ell^2$-distance $D_{\psi}(\pi_i, \pi_i') = G(\pi_i, \pi_i') = \frac{1}{2}\|\pi_i - \pi_i'\|^2$, and set $T_{\sigma} = c\cdot \max(\frac{6}{\ln 2 - \ln (2 - \eta \mu) }\ln T + \frac{2 \ln 64}{\ln 2 - \ln (2 - \eta \mu)}, 1)$ for some constant $c\geq 1$, then the strategy profile $\pi^T$ updated by APMD satisfies:
\begin{align*}
    &\mathrm{GAP}(\pi^T) \\
    &\leq \frac{2\sqrt{2}c\left((\mu + L) \cdot \mathrm{diam}(\mathcal{X}) + \zeta\right)\cdot  \left(\frac{6}{\ln 2 - \ln (2 - \eta \mu) }\ln T + \frac{2 \ln 64}{\ln 2 - \ln (2 - \eta \mu)} + 1\right)}{\sqrt{T}}\sqrt{\mathrm{diam}(\mathcal{X})\left(8 \cdot \mathrm{diam}(\mathcal{X}) + \frac{\zeta}{\mu}\right)}.
\end{align*}
\end{theorem}

\begin{lemma}[Formal version of Lemma \ref{lem:gd_sp_full}]
\label{lem:formal_gd_sp_full}
Assume that $D_{\psi}$ and $G$ are set as the squared $\ell^2$-distance.
If we use the constant learning rate $\eta_t = \eta \in (0, \frac{2\mu \rho^2}{3\mu^2\rho^2 + 8L^2})$, the $k+1$-th slingshot strategy profile $\sigma^{k+1}$ of APMD under the full feedback setting satisfies that:
\begin{align*}
    \|\pi^{\mu,\sigma^k} -  \sigma^{k+1}\|^2 \leq \|\pi^{\mu,\sigma^k} - \sigma^k\|^2\left(1 - \frac{\eta \mu}{2}\right)^{T_{\sigma}}.
\end{align*}
\end{lemma}

\begin{lemma}[Formal version of Lemma \ref{lem:gap_fn_slingshot_strategy}]
\label{lem:formal_gap_fn_slingshot_strategy}
Assume that $\sqrt{\sum_{i=1}^N \|\nabla_{\pi_i} v_i(\pi)\|^2} \leq \zeta$ for any $\pi \in \mathcal{X}$.
If $G$ is set as the squared $\ell^2$-distance, the gap function for $\sigma^{k+1}$ of APMD satisfies for $k\geq 0$:
\begin{align*}
    \mathrm{GAP}(\sigma^{k+1}) \leq \mu  \cdot \mathrm{diam}(\mathcal{X})\cdot \|\pi^{\mu, \sigma^k} - \sigma^k\| + \left(L\cdot \mathrm{diam}(\mathcal{X}) + \zeta\right) \cdot \|\pi^{\mu, \sigma^k} - \sigma^{k+1}\|.
\end{align*}
\end{lemma}

\begin{lemma}[Formal version of Lemma \ref{lem:slingshot_diff_full}]
\label{lem:formal_slingshot_diff_full}
Assume that $T_{\sigma} \geq \max(\frac{6}{\ln 2 - \ln (2 - \eta \mu) }\ln T + \frac{2 \ln 64}{\ln 2 - \ln (2 - \eta \mu)}, 1)$.
In the same setup of Theorem \ref{thm:lic_rate_full}, the $K-1$-th slingshot strategy profile $\sigma^{K-1}$ of APMD satisfies:
\begin{align*}
    \|\pi^{\mu, \sigma^{K-1}} - \sigma^{K-1}\| \leq \frac{2\sqrt{2}}{\sqrt{K}}\sqrt{\mathrm{diam}(\mathcal{X})\left(8 \cdot \mathrm{diam}(\mathcal{X}) + \frac{\zeta}{\mu}\right)}.
\end{align*}
\end{lemma}

\subsection{Noisy Feedback Setting}
For the noisy feedback setting, we assume that $\xi^t_i \in \mathbb{R}^{d_i}$ is a zero-mean independent random vector with bounded variance.
\begin{assumption}\label{asm:noise}
For all $t\ge1$ and $i \in [N]$, the noise vector $\xi^t_i$ satisfies the following properties: (a) Zero-mean: $ \mathbb{E}[\xi^t_i | \mathcal{F}_t] = (0, \cdots, 0)^{\top}$;  (b)  Bounded variance: $\mathbb{E}[\|\xi^t_i \|^2 | \mathcal{F}_t] \le C^2$ with some constant $C>0$.
\end{assumption}
This is a standard assumption in learning in games with noisy feedback \citep{mertikopoulos2019learning,hsieh2019ontheconvergence} and stochastic optimization \citep{nemirovski2009robust,nedic2014stochastic}. 
Under Assumption \ref{asm:noise}, we can obtain the following convergence results for ADMP under the noisy feedback setting.

\begin{theorem}[Formal version of Theorem \ref{thm:lic_rate_noisy}]
\label{thm:formal_lic_rate_noisy}
Let $\theta =  \frac{3\mu^2 \rho^2 + 8 L^2}{2 \mu \rho^2} $ and $\kappa = \frac{\mu}{2}$.
Suppose that Assumption \ref{asm:noise} holds and $\sqrt{\sum_{i=1}^N \|\nabla_{\pi_i} v_i(\pi)\|^2} \leq \zeta$ for any $\pi \in \mathcal{X}$.
We also assume that $D_{\psi}$ and $G$ are set as squared $\ell^2$-distance $D_{\psi}(\pi_i, \pi_i') = G(\pi_i, \pi_i') = \frac{1}{2}\|\pi_i - \pi_i'\|^2$, and $T_{\sigma} = c\cdot \max(T^{4/5} + 2, 3)$ for some constant $c\geq 1$.
If we choose the learning rate sequence of the form $\eta_t = 1/(\kappa (t - T_{\sigma} \cdot \lfloor t / T_{\sigma} \rfloor) + 2\theta)$, then the strategy profile $\pi^T$ updated by APMD satisfies:
\begin{align*}
&\mathbb{E}\left[\mathrm{GAP}(\pi^T)\right] \leq \frac{\sqrt{6c} \mu \cdot \mathrm{diam}(\mathcal{X})^2}{T^{1/10}} \\
&+ \frac{L \cdot \mathrm{diam}(\mathcal{X}) + \zeta + \mu \cdot \mathrm{diam}(\mathcal{X})\sqrt{18c\left(\mathrm{diam}(\mathcal{X}) + \frac{\zeta}{\mu} + 1\right)}}{T^{1/10}}\sqrt{\frac{\rho(2 \theta - \kappa)\mathrm{diam}(\mathcal{X})^2 + N C^2\left( \frac{1}{\kappa }\ln \left(\frac{\kappa}{2 \theta} T + 1\right) +  \frac{1}{2 \theta}\right)}{\rho \kappa}}.
\end{align*}
\end{theorem}

\begin{lemma}[Formal version of Lemma \ref{lem:gd_sp_noisy}]
\label{lem:formal_gd_sp_noisy}
Let $\theta =  \frac{3\mu^2 \rho^2 + 8L^2}{2 \mu \rho^2} $ and $\kappa = \frac{\mu}{2}$.
Suppose that Assumption \ref{asm:noise} holds, and both $D_{\psi}$ and $G$ are defined as the squared $\ell^2$-distance.
Under the noisy feedback setting, if we use the learning rate sequence of the form $\eta_t = 1/(\kappa (t - T_{\sigma} \cdot \lfloor t / T_{\sigma}\rfloor) + 2\theta)$, the $k+1$-th slingshot strategy profile $\sigma^{k+1}$ of APMD for each $k\geq 0$ satisfies that:
\begin{align*}
\mathbb{E}\left[\|\pi^{\mu,\sigma^k} - \sigma^{k+1}\|^2\right] \le \frac{2 \theta - \kappa}{\kappa (T_{\sigma} - 1) + 2 \theta} \mathbb{E}\left[\|\pi^{\mu,\sigma^k}-\sigma^k\|^2\right] + \frac{N C^2 }{\rho (\kappa (T_{\sigma} - 1) + 2 \theta)} \left( \frac{1}{\kappa }\ln \left(\frac{\kappa}{2 \theta} (T_{\sigma} - 1) + 1\right) +  \frac{1}{2 \theta}\right).
\end{align*}
\end{lemma}

\begin{lemma}[Formal version of Lemma \ref{lem:slingshot_diff_noisy}]
Assume that $T_{\sigma}\geq \max(T^{4/5} + 2, 3)$.
In the same setup of Theorem \ref{thm:lic_rate_noisy}, the $K-1$-th slingshot strategy profile $\sigma^{K-1}$ satisfies:
\label{lem:formal_slingshot_diff_noisy}
\begin{align*}
&\mathbb{E}\left[\|\pi^{\mu,\sigma^{K-1}} - \sigma^{K-1}\|\right] \\
&\leq \sqrt{\frac{\|\pi^{\ast} - \sigma^0\|^2  + \frac{3}{\rho\kappa}\left(\mathrm{diam}(\mathcal{X}) + \frac{\zeta}{\mu} + 1\right)\left(\rho(2 \theta - \kappa)\mathrm{diam}(\mathcal{X})^2 + N C^2\left( \frac{1}{\kappa }\ln \left(\frac{\kappa}{2 \theta} T + 1\right) +  \frac{1}{2 \theta}\right)\right)}{K}}.
\end{align*}
\end{lemma}

\subsection{Convergence Results with General $G$ and $D_{\psi}$}
\begin{lemma}[Formal version of Lemma \ref{lem:md_sp_full}]
\label{lem:formal_md_sp_full}
Suppose that Assumptions \ref{asm:rel_smooth} with $\beta,\gamma \in (0, \infty)$ and \ref{asm:well_defined} hold.
If we use the constant learning rate $\eta_t = \eta \in (0, \frac{2\mu \gamma \rho^2}{\mu^2 \gamma \rho^2 (\gamma + 2\beta) + 8L^2})$, the $k+1$-th slingshot strategy profile $\sigma^{k+1}$ of APMD under the full feedback setting satisfies that:
\begin{align*}
    D_{\psi}(\pi^{\mu,\sigma^k}, \sigma^{k+1}) \leq D_{\psi}(\pi^{\mu,\sigma^k}, \sigma^k)\left(1 - \frac{\eta\mu \gamma}{2}\right)^{T_{\sigma}}.
\end{align*}
\end{lemma}

\begin{lemma}[Formal version of Lemma \ref{lem:md_sp_noisy}]
\label{lem:formal_md_sp_noisy}
Let $\theta =  \frac{\mu^2 \gamma \rho^2 (\gamma + 2\beta) + 8 L^2}{2 \mu \gamma \rho^2} $ and $\kappa = \frac{\mu\gamma}{2}$.
Suppose that Assumptions \ref{asm:rel_smooth}, \ref{asm:well_defined}, and \ref{asm:noise} hold, and the learning rate sequence of the form $\eta_t = 1/(\kappa (t - T_{\sigma} \cdot \lfloor t / T_{\sigma}\rfloor) + 2\theta)$ is used.
Then, the $k+1$-th slingshot strategy profile $\sigma^{k+1}$ of APMD under the noisy feedback setting satisfies that:
\begin{align*}
\mathbb{E}[D_{\psi}(\pi^{\mu,\sigma^k}, \sigma^{k+1})] \le \frac{2 \theta - \kappa}{\kappa (T_{\sigma} - 1) + 2 \theta} \mathbb{E}[D_\psi(\pi^{\mu,\sigma^k}, \sigma^k)] + \frac{N C^2 }{\rho (\kappa (T_{\sigma} - 1) + 2 \theta)} \left( \frac{1}{\kappa }\ln \left(\frac{\kappa}{2 \theta} (T_{\sigma} - 1) + 1\right) +  \frac{1}{2 \theta}\right).
\end{align*}
\end{lemma}

\section{Extension to Follow the Regularized Leader}
\label{sec:appx_extension_to_ftrl}
In Sections \ref{sec:last_iterate_convergence_rates} and \ref{sec:discussion}, we introduced and analyzed APMD, which extends the standard MD approach.
Similarly, it is possible to extend the FTRL approach as well.
In this section, we present Adaptively Perturbed Follow the Regularized Leader (APFTRL), which incorporates the perturbation term $\mu G(\cdot,\sigma_i)$ into the conventional FTRL algorithm:
\begin{align*}
\pi_i^{t+1} = \argmax_{x\in \mathcal{X}_i} \left\{ \sum_{s=0}^t \eta_s\left\langle \widehat{\nabla}_{\pi_i}v_i(\pi^s) - \mu \nabla_{\pi_i}G(\pi_i^s, \sigma_i), x\right\rangle - \psi(x)\right\}.
\end{align*}
In this section, we make the assumption that $\mathcal{X}_i$ is an affine subset, which means there exists a matrix $A\in \mathbb{R}^{k_i\times d_i}$ and a vector $b\in \mathbb{R}^{k_i}$ such that $A\pi_i = b$ for all $\pi_i\in \mathcal{X}_i$.
Additionally, we also assume that $\psi$ is a Legendre function, as described in \citep{rockafellar1997convex,lattimore2020bandit}.
Another assumption we make is that $\pi^t$ is consistently well-defined over iterations:
\begin{assumption}
\label{asm:well_defined_ftrl}
$\pi^t$, updated by APFTRL, satisfies the condition $\pi^t\in \mathrm{int}(\mathrm{dom}~\psi)$ for every $t\in \{0, 1, \cdots, T\}$.
\end{assumption}
Then, APFTRL enjoys the last-iterate convergence of $\pi^T$ in full and noisy feedback settings by proving the following lemmas:
\begin{lemma}
\label{lem:formal_ftrl_sp_full}
Suppose that Assumptions \ref{asm:rel_smooth} with $\beta,\gamma \in (0, \infty)$ and \ref{asm:well_defined_ftrl} hold.
If we use the constant learning rate $\eta_t = \eta \in (0, \frac{2\mu \gamma \rho^2}{\mu^2 \gamma \rho^2 (\gamma + 2\beta) + 8L^2})$, the $k+1$-th slingshot strategy profile $\sigma^{k+1}$ of APFTRL under the full feedback setting satisfies that:
\begin{align*}
    D_{\psi}(\pi^{\mu,\sigma^k}, \sigma^{k+1}) \leq D_{\psi}(\pi^{\mu,\sigma^k}, \sigma^k)\left(1 - \frac{\eta\mu \gamma}{2}\right)^{T_{\sigma}}.
\end{align*}
\end{lemma}

\begin{lemma}
\label{lem:formal_ftrl_sp_noisy}
Let $\theta =  \frac{\mu^2 \gamma \rho^2 (\gamma + 2\beta) + 8 L^2}{2 \mu \gamma \rho^2} $ and $\kappa = \frac{\mu\gamma}{2}$.
Suppose that Assumptions \ref{asm:rel_smooth}, \ref{asm:noise}, and \ref{asm:well_defined_ftrl} hold, and the learning rate sequence of the form $\eta_t = 1/(\kappa (t - T_{\sigma} \cdot \lfloor t / T_{\sigma}\rfloor) + 2\theta)$ is used.
Then, the $k+1$-th slingshot strategy profile $\sigma^{k+1}$ of APFTRL under the noisy feedback setting satisfies that:
\begin{align*}
\mathbb{E}[D_{\psi}(\pi^{\mu,\sigma^k}, \sigma^{k+1})] \le \frac{2 \theta - \kappa}{\kappa (T_{\sigma} - 1) + 2 \theta} \mathbb{E}[D_\psi(\pi^{\mu,\sigma^k}, \sigma^k)] + \frac{N C^2 }{\rho (\kappa (T_{\sigma} - 1) + 2 \theta)} \left( \frac{1}{\kappa }\ln \left(\frac{\kappa}{2 \theta} (T_{\sigma} - 1) + 1\right) +  \frac{1}{2 \theta}\right).
\end{align*}
\end{lemma}

The proofs of these theorems can be found in Appendix~\ref{sec:appx_frtl_sp_full}, \ref{sec:appx_frtl_sp_noisy}.

\section{Additional Theorems}
\label{sec:approximate_nash_conv_full}
Based on this theorem, we can show that the gap function for $\pi^t$ converges to the value of $\mathcal{O}(\mu)$.
\begin{theorem}
\label{thm:exploitability_upper_bound}
In the same setup of Lemma \ref{lem:md_sp_full}, the gap function for APMD is bounded as:
\begin{align*}
\mathrm{GAP}(\pi^t) &\leq \mu \cdot \mathrm{diam}(\mathcal{X})\sqrt{\sum_{i=1}^N\| \nabla_{\pi_i}G(\pi_i^{\mu,\sigma}, \sigma_i)\|^2} + \mathcal{O}\left(\left(1 - \frac{\eta \mu\gamma}{2}\right)^{\frac{t}{2}}\right).
\end{align*}
\end{theorem}
We note that lower $\mu$ reduces the gap function for $\pi^{\mu,\sigma}$ (as in the first term of Theorem \ref{thm:exploitability_upper_bound}), whereas higher $\mu$ makes $\pi^t$ converge faster (as in the second term of Theorem \ref{thm:exploitability_upper_bound}).
That is, $\mu$ controls a trade-off between the speed of convergence and the gap function.

\section{Proofs for Section \ref{sec:last_iterate_convergence_rates}}
\label{sec:appx_exact_nash_conv}

\subsection{Proof of Theorem \ref{thm:formal_lic_rate_full} (Formal Version of Theoremf \ref{thm:lic_rate_full})}
\begin{proof}[Proof of Theorem \ref{thm:formal_lic_rate_full}]
First, from Lemma \ref{lem:formal_gap_fn_slingshot_strategy}, we have for any $k\geq 0$:
\begin{align*}
    \mathrm{GAP}(\sigma^{k+1}) \leq \mu  \cdot \mathrm{diam}(\mathcal{X})\cdot \|\pi^{\mu, \sigma^k} - \sigma^k\| + \left(L\cdot \mathrm{diam}(\mathcal{X}) + \zeta\right) \cdot \|\pi^{\mu, \sigma^k} - \sigma^{k+1}\|.
\end{align*}

Using Lemma \ref{lem:formal_gd_sp_full}, we can upper bound the term of $\|\pi^{\mu, \sigma^k} - \sigma^{k+1}\|$ as follows:
\begin{align*}
    \|\pi^{\mu,\sigma^k} -  \sigma^{k+1}\|^2 \leq \|\pi^{\mu,\sigma^k} - \sigma^k\|^2\left(1 - \frac{\eta \mu}{2}\right)^{T_{\sigma}}.
\end{align*}

Combining these inequalities, we have for any $k\geq 0$:
\begin{align*}
    \mathrm{GAP}(\sigma^{k+1}) &\leq \mu  \cdot \mathrm{diam}(\mathcal{X})\cdot \|\pi^{\mu, \sigma^k} - \sigma^k\| + \left(1 - \frac{\eta \mu}{2}\right)^{\frac{T_{\sigma}}{2}} \cdot \left(L\cdot \mathrm{diam}(\mathcal{X}) + \zeta\right) \cdot \|\pi^{\mu,\sigma^k} - \sigma^k\| \\
    &\leq \left((\mu + L)\cdot \mathrm{diam}(\mathcal{X}) + \zeta\right) \cdot \|\pi^{\mu,\sigma^k} - \sigma^k\|,
\end{align*}
where the second inequality follows from $T_{\sigma}\geq 1$.
Let us denote $K:= \lfloor T/T_{\sigma}\rfloor$ as the total number of the slingshot strategy updates over the entire $T$ iterations.
By letting $k=K-1$ in the above inequality, we get:
\begin{align}
\mathrm{GAP}(\sigma^K) \leq \left((\mu + L)\cdot \mathrm{diam}(\mathcal{X}) + \zeta\right) \cdot \|\pi^{\mu,\sigma^{K-1}} - \sigma^{K-1}\|.
\label{eq:gap_fn_last_slingshot_strategy}
\end{align}

Next, we derive the following upper bound on $\|\pi^{\mu, \sigma^{K-1}} - \sigma^{K-1}\|$ from Lemma \ref{lem:formal_slingshot_diff_full}:
\begin{align}
    \|\pi^{\mu, \sigma^{K-1}} - \sigma^{K-1}\| \leq \frac{2\sqrt{2}}{\sqrt{K}}\sqrt{\mathrm{diam}(\mathcal{X})\left(8 \cdot \mathrm{diam}(\mathcal{X}) + \frac{\zeta}{\mu}\right)}.
\label{eq:slingshot_diff_full}
\end{align}

By combining \eqref{eq:gap_fn_last_slingshot_strategy} and \eqref{eq:slingshot_diff_full}, we get:
\begin{align*}
    \mathrm{GAP}(\sigma^K) &\leq \frac{2\sqrt{2}\left((\mu + L) \cdot \mathrm{diam}(\mathcal{X}) + \zeta\right)}{\sqrt{K}}\sqrt{\mathrm{diam}(\mathcal{X})\left(8 \cdot \mathrm{diam}(\mathcal{X}) + \frac{\zeta}{\mu}\right)}.
\end{align*}

Finally, since $\pi^T = \sigma^K$, $K=\lfloor T / T_{\sigma}\rfloor$, and $T_{\sigma}=c \cdot \max(\frac{6}{\ln 2 - \ln (2 - \eta \mu) }\ln T + \frac{2 \ln 64}{\ln 2 - \ln (2 - \eta \mu)}, 1)$, we have:
\begin{align*}
    &\mathrm{GAP}(\pi^T) \\
    &\leq \frac{2\sqrt{2}\left((\mu + L) \cdot \mathrm{diam}(\mathcal{X}) + \zeta\right)}{\sqrt{T / T_{\sigma}}}\sqrt{\mathrm{diam}(\mathcal{X})\left(8 \cdot \mathrm{diam}(\mathcal{X}) + \frac{\zeta}{\mu}\right)} \\
    &\leq \frac{2\sqrt{2}c\left((\mu + L) \cdot \mathrm{diam}(\mathcal{X}) + \zeta\right)\cdot  \left(\frac{6}{\ln 2 - \ln (2 - \eta \mu) }\ln T + \frac{2 \ln 64}{\ln 2 - \ln (2 - \eta \mu)} + 1\right)}{\sqrt{T}}\sqrt{\mathrm{diam}(\mathcal{X})\left(8 \cdot \mathrm{diam}(\mathcal{X}) + \frac{\zeta}{\mu}\right)}.
\end{align*}
This concludes the statement of the theorem.
\end{proof}

\subsection{Proof of Lemma \ref{lem:formal_gd_sp_full} (Formal Version of Lemma \ref{lem:gd_sp_full})}
\begin{proof}[Proof of Lemma \ref{lem:formal_gd_sp_full}]
From the definition of the Bregman divergence, we have for all $\pi_i, \pi_i' \in \mathcal{X}_i$:
\begin{align*}
    &D_{\psi}(\pi_i', \sigma_i) - D_{\psi}(\pi_i, \sigma_i) - \langle \nabla_{\pi_i}D_{\psi}(\pi_i, \sigma_i), \pi_i' - \pi_i\rangle \\
    &= \psi(\pi_i') - \psi(\sigma_i) - \langle \nabla \psi(\sigma_i), \pi_i' - \sigma_i\rangle -\psi(\pi_i) + \psi(\sigma_i) + \langle \nabla \psi(\sigma_i), \pi_i - \sigma_i\rangle  - \langle \nabla\psi(\pi_i) - \nabla\psi(\sigma_i), \pi_i' - \pi_i\rangle \\
    &= \psi(\pi_i') -\psi(\pi_i) - \langle \nabla\psi(\pi_i), \pi_i' - \pi_i\rangle \\
    &= D_{\psi}(\pi_i', \pi_i).
\end{align*}
Hence, assuming that $G$ is identical to $D_{\psi}$, Assumption \ref{asm:rel_smooth} is satisfied with $\beta=\gamma=1$.
Furthermore, since $\psi(x)=\frac{1}{2}\left\|x\right\|^2$, both $\rho=1$ and $\mathrm{int}(\mathrm{dom}~\psi)=\mathbb{R}^{d_i}$ hold.
Therefore, Assumption \ref{asm:well_defined} is also satisfied.
Consequently, we can obtain the following convergence result from Lemma \ref{lem:formal_md_sp_full} with $\beta=\gamma=\rho=1$:
\begin{align*}
    \|\pi^{\mu,\sigma^k} -  \sigma^{k+1}\|^2 \leq \|\pi^{\mu,\sigma^k} - \sigma^k\|^2\left(1 - \frac{\eta \mu}{2}\right)^{T_{\sigma}}.
\end{align*}
\end{proof}

\subsection{Proof of Lemma \ref{lem:formal_gap_fn_slingshot_strategy} (Formal Version of Lemma \ref{lem:gap_fn_slingshot_strategy})}
\begin{proof}[Proof of Lemma \ref{lem:formal_gap_fn_slingshot_strategy}]
First, we have for any $\pi\in \mathcal{X}$:
\begin{align}
    \mathrm{GAP}(\pi) &= \max_{\tilde{\pi}_i\in \mathcal{X}_i}\sum_{i=1}^N\langle \nabla_{\pi_i}v_i(\pi), \tilde{\pi}_i - \pi_i\rangle \nonumber\\
    &= \max_{\tilde{\pi}\in \mathcal{X}} \sum_{i=1}^N \left(\langle \nabla_{\pi_i}v_i(\pi'), \tilde{\pi}_i - \pi_i'\rangle - \langle \nabla_{\pi_i}v_i(\pi'), \pi_i - \pi_i'\rangle + \langle \nabla_{\pi_i}v_i(\pi) - \nabla_{\pi_i}v_i(\pi'), \tilde{\pi}_i - \pi_i\rangle\right).
\label{eq:decompose_exploitability}
\end{align}

Here, we introduce the following lemma from \citet{cai2022finite}:
\begin{lemma}[Lemma 2 of \citet{cai2022finite}]
\label{lem:exploit_by_tangent}
For any $\pi \in \mathcal{X}$, we have:
\begin{align*}
    \max_{\tilde{\pi}\in \mathcal{X}}\sum_{i=1}^N \langle \nabla_{\pi_i} v_i(\pi), \tilde{\pi}_i - \pi_i\rangle \leq \mathrm{diam}(\mathcal{X}) \cdot \min_{(a_i)\in N_{\mathcal{X}}(\pi)}\sqrt{\sum_{i=1}^N\|-\nabla_{\pi_i}v_i(\pi) + a_i\|^2},
\end{align*}
where $N_{\mathcal{X}}(\pi) = \{(a_i)_{i\in [N]}\in \prod_{i=1}^N\mathbb{R}^{d_i} ~|~ \sum_{i=1}^N\langle a_i, \pi_i' - \pi_i\rangle \leq 0, ~\forall \pi'\in \mathcal{X}\}$.
\end{lemma}

From Lemma \ref{lem:exploit_by_tangent}, the first term of \eqref{eq:decompose_exploitability} can be upper bounded as:
\begin{align}
    \max_{\tilde{\pi}\in \mathcal{X}} \sum_{i=1}^N \langle \nabla_{\pi_i}v_i(\pi'), \tilde{\pi}_i - \pi_i'\rangle \leq \mathrm{diam}(\mathcal{X})\cdot \min_{(a_i)\in N_{\mathcal{X}}(\pi')} \sqrt{\sum_{i=1}^N \|-\nabla_{\pi_i} v_i(\pi') + a_i\|^2}
    \label{eq:source_exploitability}
\end{align}

Next, from Cauchy-Schwarz inequality, the second term of \eqref{eq:decompose_exploitability} can be upper bounded as:
\begin{align}
    - \sum_{i=1}^N \langle \nabla_{\pi_i}v_i(\pi'), \pi_i - \pi_i'\rangle &\leq \|\pi - \pi'\| \sqrt{\sum_{i=1}^N \|\nabla_{\pi_i}v_i(\pi')\|^2} \nonumber\\
    &\leq \zeta\|\pi - \pi'\|.
    \label{eq:strategy_diff}
\end{align}

Again from Cauchy-Schwarz inequality, the third term of \eqref{eq:decompose_exploitability} can be upper bounded as:
\begin{align}
    \sum_{i=1}^N \langle \nabla_{\pi_i}v_i(\pi) - \nabla_{\pi_i}v_i(\pi'), \tilde{\pi}_i - \pi_i\rangle &\leq \|\tilde{\pi} - \pi\| \sqrt{\sum_{i=1}^N \|\nabla_{\pi_i}v_i(\pi) - \nabla_{\pi_i}v_i(\pi')\|^2} \nonumber\\
    &\leq \mathrm{diam}(\mathcal{X}) \sqrt{\sum_{i=1}^N \|\nabla_{\pi_i}v_i(\pi) - \nabla_{\pi_i}v_i(\pi')\|^2} \nonumber\\
    &\leq L \cdot \mathrm{diam}(\mathcal{X}) \cdot \|\pi - \pi'\|.
    \label{eq:gradient_diff}
\end{align}

By combining \eqref{eq:decompose_exploitability}, \eqref{eq:source_exploitability}, \eqref{eq:strategy_diff}, and \eqref{eq:gradient_diff}, we get for any $\pi, \pi'\in \mathcal{X}$:
\begin{align*}
    \mathrm{GAP}(\pi) &\leq \mathrm{diam}(\mathcal{X})\cdot \min_{(a_i)\in N_{\mathcal{X}}(\pi')} \sqrt{\sum_{i=1}^N \|-\nabla_{\pi_i} v_i(\pi') + a_i\|^2} + \left(L\cdot \mathrm{diam}(\mathcal{X}) + \zeta\right)\|\pi - \pi'\|.
\end{align*}
Thus, letting $\pi = \sigma^{k+1}$ and $\pi' = \pi^{\mu, \sigma^k}$, we have:
\begin{align}
    \mathrm{GAP}(\sigma^{k+1}) \leq \mathrm{diam}(\mathcal{X})\cdot \min_{(a_i)\in N_{\mathcal{X}}(\pi^{\mu, \sigma^k})} \sqrt{\sum_{i=1}^N \|-\nabla_{\pi_i} v_i(\pi^{\mu, \sigma^k}) + a_i\|^2} + \left(L\cdot \mathrm{diam}(\mathcal{X}) + \zeta\right)\|\pi^{\mu, \sigma^k} - \sigma^{k+1}\|.
\label{eq:decompose_exploitability_constrained}
\end{align}

On the other hand, from the first-order optimality condition for $\pi^{\mu, \sigma^k}$, we have for any $\pi \in \mathcal{X}$:
\begin{align*}
    \sum_{i=1}^N \langle \nabla_{\pi_i}v_i(\pi^{\mu, \sigma^k}) - \mu \left(\pi_i^{\mu, \sigma^k} - \sigma_i^k\right), \pi_i - \pi_i^{\mu, \sigma^k}\rangle \leq 0,
\end{align*}
and then $\left(\nabla_{\pi_i}v_i(\pi^{\mu, \sigma^k}) - \mu \left(\pi_i^{\mu, \sigma^k} - \sigma_i^k\right)\right)_{i\in [N]}\in N_{\mathcal{X}}(\pi^{\mu, \sigma^k})$.
Thus, the first term of \eqref{eq:decompose_exploitability_constrained} can be bounded as:
\begin{align}
    &\min_{(a_i)\in N_{\mathcal{X}}(\pi^{\mu, \sigma^k})} \sqrt{\sum_{i=1}^N \|-\nabla_{\pi_i} v_i(\pi^{\mu, \sigma^k}) + a_i\|^2} \nonumber\\
    &\leq \sqrt{\sum_{i=1}^N \|-\nabla_{\pi_i} v_i(\pi^{\mu, \sigma^k}) + \nabla_{\pi_i}v_i(\pi^{\mu, \sigma^k}) - \mu \left(\pi_i^{\mu, \sigma^k} - \sigma_i^k\right)\|^2} \nonumber\\
    &= \mu \|\pi^{\mu, \sigma^k} - \sigma^k\|.
\label{eq:upper_bound_on_tangent_residual}
\end{align}

Combining \eqref{eq:decompose_exploitability_constrained} and \eqref{eq:upper_bound_on_tangent_residual}, we have:
\begin{align*}
    \mathrm{GAP}(\sigma^{k+1}) \leq \mu  \cdot \mathrm{diam}(\mathcal{X})\cdot \|\pi^{\mu, \sigma^k} - \sigma^k\| + \left(L\cdot \mathrm{diam}(\mathcal{X}) + \zeta\right)\|\pi^{\mu, \sigma^k} - \sigma^{k+1}\|
\end{align*}
\end{proof}

\subsection{Proof of Lemma \ref{lem:formal_slingshot_diff_full} (Formal Version of Lemma \ref{lem:slingshot_diff_full})}
\begin{proof}[Proof of Lemma \ref{lem:formal_slingshot_diff_full}]
First, we prove the following lemma:
\begin{lemma}
\label{lem:contraction_constrained}
Assume that $\sqrt{\sum_{i=1}^N \|\nabla_{\pi_i} v_i(\pi)\|^2} \leq \zeta$ for any $\pi \in \mathcal{X}$.
If $G$ is set as the squared $\ell^2$-distance, we have for any $k\in [K]$:
\begin{align*}
    \|\pi^{\mu, \sigma^k} - \sigma^k\|^2 &\leq \|\sigma^k - \sigma^{k-1}\|^2 + \frac{2\zeta}{\mu} \|\pi^{\mu, \sigma^{k-1}} - \sigma^k\|.
\end{align*}
\end{lemma}

From Lemma \ref{lem:contraction_constrained}, we can bound $\|\pi^{\mu, \sigma^k} - \sigma^k\|^2$ as:
\begin{align}
    &\|\pi^{\mu, \sigma^k} - \sigma^k\|^2 \nonumber \\
    &\leq \|\sigma^k - \sigma^{k-1}\|^2 + \frac{2\zeta}{\mu} \|\pi^{\mu, \sigma^{k-1}} - \sigma^k\| \nonumber \\
    &\leq \|\pi^{\mu, \sigma^{k-1}} - \sigma^{k-1}\|^2 + \|\sigma^k - \pi^{\mu, \sigma^{k-1}}\|^2 + 2\|\sigma^k - \pi^{\mu, \sigma^{k-1}}\|\|\pi^{\mu, \sigma^{k-1}} - \sigma^{k-1}\| + \frac{2\zeta}{\mu} \|\pi^{\mu, \sigma^{k-1}} - \sigma^k\|.
\label{eq:diff_p_sigma_by_norm}
\end{align}

Next, we upper bound $\|\pi^{\mu, \sigma^k} - \sigma^k\|^2$ using the following lemma:
\begin{lemma}
\label{lem:diff_p_sigma_sum_constrained}
Assume that $T_{\sigma} \geq \max(\frac{6}{\ln 2 - \ln (2 - \eta \mu) }\ln T + \frac{2 \ln 64}{\ln 2 - \ln (2 - \eta \mu)}, 1)$.
In the same setup of Theorem \ref{thm:lic_rate_full}, we have for any Nash equilibrium $\pi^{\ast} \in \Pi^{\ast}$:
\begin{align*}
    \sum_{k=0}^{K-1}\|\pi^{\mu, \sigma^k} - \sigma^k\|^2 &\leq 16\|\pi^{\ast} - \sigma^0\|^2.
\end{align*}
\end{lemma}

Using Lemma \ref{lem:diff_p_sigma_sum_constrained}, we have:
\begin{align*}
    \|\pi^{\mu, \sigma^k} - \sigma^k\|^2 \leq \sum_{k=0}^{K-1}\|\pi^{\mu, \sigma^k} - \sigma^k\|^2 \leq 16\|\pi^{\ast} - \sigma^0\|^2,
\end{align*}
and then from Lemma \ref{lem:formal_gd_sp_full}, we get:
\begin{align}
    \|\pi^{\mu,\sigma^k} -  \sigma^{k+1}\| \leq \|\pi^{\mu,\sigma^k} - \sigma^k\|\left(1 - \frac{\eta \mu}{2}\right)^{\frac{T_{\sigma}}{2}} \leq 4\|\pi^{\ast} - \sigma^0\|\left(1 - \frac{\eta \mu}{2}\right)^{\frac{T_{\sigma}}{2}}.
\label{eq:diff_p_sigma_exponential}
\end{align}

By combining \eqref{eq:diff_p_sigma_by_norm} and \eqref{eq:diff_p_sigma_exponential}, we have:
\begin{align*}
    &\|\pi^{\mu, \sigma^k} - \sigma^k\|^2 \\
    &\leq \|\pi^{\mu, \sigma^{k-1}} - \sigma^{k-1}\|^2 + 16\|\pi^{\ast} - \sigma^0\|\left(1 - \frac{\eta \mu}{2}\right)^{T_{\sigma}} + 32\|\pi^{\ast} - \sigma^0\|^2\left(1 - \frac{\eta \mu}{2}\right)^{\frac{T_{\sigma}}{2}} + \frac{8\zeta}{\mu} \|\pi^{\ast} - \sigma^0\|\left(1 - \frac{\eta \mu}{2}\right)^{\frac{T_{\sigma}}{2}} \\
    &\leq \|\pi^{\mu, \sigma^{k-1}} - \sigma^{k-1}\|^2 + 48\|\pi^{\ast} - \sigma^0\|^2\left(1 - \frac{\eta \mu}{2}\right)^{\frac{T_{\sigma}}{2}} + \frac{8\zeta}{\mu} \|\pi^{\ast} - \sigma^0\|\left(1 - \frac{\eta \mu}{2}\right)^{\frac{T_{\sigma}}{2}} \\
    &\leq \|\pi^{\mu, \sigma^{k-1}} - \sigma^{k-1}\|^2 + 8\|\pi^{\ast} - \sigma^0\|\left(1 - \frac{\eta \mu}{2}\right)^{\frac{T_{\sigma}}{2}}\left(6\|\pi^{\ast} - \sigma^0\| + \frac{\zeta}{\mu}\right).
\end{align*}
Therefore, we get:
\begin{align}
    \|\pi^{\mu, \sigma^{K-1}} - \sigma^{K-1}\|^2 &\leq \|\pi^{\mu, \sigma^{K-2}} - \sigma^{K-2}\|^2 + 8\|\pi^{\ast} - \sigma^0\|\left(1 - \frac{\eta \mu}{2}\right)^{\frac{T_{\sigma}}{2}}\left(6\|\pi^{\ast} - \sigma^0\| + \frac{\zeta}{\mu}\right) \nonumber \\
    &\leq \|\pi^{\mu, \sigma^k} - \sigma^k\|^2 + 8K\|\pi^{\ast} - \sigma^0\|\left(1 - \frac{\eta \mu}{2}\right)^{\frac{T_{\sigma}}{2}}\left(6\|\pi^{\ast} - \sigma^0\| + \frac{\zeta}{\mu}\right).
\label{eq:diff_p_sigma_constrained}
\end{align}

By combining \eqref{eq:diff_p_sigma_constrained} and Lemma \ref{lem:diff_p_sigma_sum_constrained}, we have:
\begin{align*}
    K\|\pi^{\mu, \sigma^{K-1}} - \sigma^{K-1}\|^2 &\leq \sum_{k=0}^{K-1}\|\pi^{\mu, \sigma^k} - \sigma^k\|^2 + 8K^2\|\pi^{\ast} - \sigma^0\|\left(1 - \frac{\eta \mu}{2}\right)^{\frac{T_{\sigma}}{2}}\left(6\|\pi^{\ast} - \sigma^0\| + \frac{\zeta}{\mu}\right) \\
    &\leq 16\|\pi^{\ast} - \sigma^0\|^2 + 8K^2\|\pi^{\ast} - \sigma^0\|\left(1 - \frac{\eta \mu}{2}\right)^{\frac{T_{\sigma}}{2}}\left(6\|\pi^{\ast} - \sigma^0\| + \frac{\zeta}{\mu}\right) \\
    &\leq 16\|\pi^{\ast} - \sigma^0\|^2 + 8T^2\|\pi^{\ast} - \sigma^0\|\left(1 - \frac{\eta \mu}{2}\right)^{\frac{T_{\sigma}}{2}}\left(6\|\pi^{\ast} - \sigma^0\| + \frac{\zeta}{\mu}\right).
\end{align*}
Under the assumption that $T_{\sigma} \geq \max(\frac{6}{\ln 2 - \ln (2 - \eta \mu) }\ln T + \frac{2 \ln 64}{\ln 2 - \ln (2 - \eta \mu)}, 1)$, we get:
\begin{align*}
    \left(1-\frac{\eta\mu}{2}\right)^{-\frac{T_{\sigma}}{2}} \geq \left(1-\frac{\eta\mu}{2}\right)^{-\frac{3}{\ln 2 - \ln (2 - \eta \mu) }\ln T} \left(1-\frac{\eta\mu}{2}\right)^{\frac{\ln 64}{\ln \left(1 - \frac{\eta\mu}{2}\right)}} = 64\left(1-\frac{\eta\mu}{2}\right)^{\frac{\ln T^3}{\ln (2 - \eta \mu)  - \ln 2}}=64T^3.
\end{align*}
Thus, we have:
\begin{align*}
    K\|\pi^{\mu, \sigma^{K-1}} - \sigma^{K-1}\|^2 &\leq 16\|\pi^{\ast} - \sigma^0\|^2 + 8\|\pi^{\ast} - \sigma^0\|\left(6\|\pi^{\ast} - \sigma^0\| + \frac{\zeta}{\mu}\right) \\
    &= 8\|\pi^{\ast} - \sigma^0\|\left(8\|\pi^{\ast} - \sigma^0\| + \frac{\zeta}{\mu}\right) \\
    &\leq 8\cdot \mathrm{diam}(\mathcal{X})\left(8\cdot \mathrm{diam}(\mathcal{X}) + \frac{\zeta}{\mu}\right).
\end{align*}
\end{proof}

\subsection{Proof of Theorem \ref{thm:formal_lic_rate_noisy} (Formal Version of Theorem \ref{thm:lic_rate_noisy})}
\begin{proof}[Proof of Theorem \ref{thm:formal_lic_rate_noisy}]
First, from Lemma \ref{lem:formal_gap_fn_slingshot_strategy}, we have for any $k\geq 0$:
\begin{align*}
    \mathbb{E}\left[\mathrm{GAP}(\sigma^{k+1})\right] \leq \mu  \cdot \mathrm{diam}(\mathcal{X})\cdot \mathbb{E}\left[\|\pi^{\mu, \sigma^k} - \sigma^k\|\right] + \left(L\cdot \mathrm{diam}(\mathcal{X}) + \zeta\right) \cdot \mathbb{E}\left[\|\pi^{\mu, \sigma^k} - \sigma^{k+1}\|\right].
\end{align*}

Using Lemma \ref{lem:formal_gd_sp_noisy}, we can upper bound the term of $\mathbb{E}\left[\|\pi^{\mu,\sigma^k} - \sigma^{k+1}\|^2\right]$ as follows:
\begin{align*}
\mathbb{E}\left[\|\pi^{\mu,\sigma^k} - \sigma^{k+1}\|^2\right] \le \frac{2 \theta - \kappa}{\kappa (T_{\sigma} - 1) + 2 \theta} \mathbb{E}\left[\|\pi^{\mu,\sigma^k}-\sigma^k\|^2\right] + \frac{N C^2 }{\rho (\kappa (T_{\sigma} - 1) + 2 \theta)} \left( \frac{1}{\kappa }\ln \left(\frac{\kappa}{2 \theta} (T_{\sigma} - 1) + 1\right) +  \frac{1}{2 \theta}\right).
\end{align*}

Combining these inequalities, we have for any $k\geq 0$:
\begin{align*}
    &\mathbb{E}\left[\mathrm{GAP}(\sigma^{k+1})\right] \\
    &\leq \mu  \cdot \mathrm{diam}(\mathcal{X})\cdot \mathbb{E}\left[\|\pi^{\mu, \sigma^k} - \sigma^k\|\right] + \left(L\cdot \mathrm{diam}(\mathcal{X}) + \zeta\right) \cdot \sqrt{\frac{\rho(2 \theta - \kappa)\mathrm{diam}(\mathcal{X})^2 + N C^2\left( \frac{1}{\kappa }\ln \left(\frac{\kappa}{2 \theta} T + 1\right) +  \frac{1}{2 \theta}\right)}{\rho\kappa T^{4/5}}},
\end{align*}
where we use $T_{\sigma} \geq \max(T^{4/5} + 2, 3)$.
Let us denote $K:= \lfloor T/T_{\sigma}\rfloor$ as the total number of the slingshot strategy updates over the entire $T$ iterations.
By letting $k=K-1$ in the above inequality, we get:
\begin{align}
\mathbb{E}\left[\mathrm{GAP}(\sigma^K)\right] &\leq \mu  \cdot \mathrm{diam}(\mathcal{X})\cdot \mathbb{E}\left[\|\pi^{\mu, \sigma^{K-1}} - \sigma^{K-1}\|\right] \nonumber\\
&+ \left(L\cdot \mathrm{diam}(\mathcal{X}) + \zeta\right) \cdot \sqrt{\frac{\rho(2 \theta - \kappa)\mathrm{diam}(\mathcal{X})^2 + N C^2\left( \frac{1}{\kappa }\ln \left(\frac{\kappa}{2 \theta} T + 1\right) +  \frac{1}{2 \theta}\right)}{\rho\kappa T^{4/5}}},
\label{eq:gap_fn_last_slingshot_strategy_noisy}
\end{align}

Next, we derive the following upper bound on $\mathbb{E}\left[\|\pi^{\mu, \sigma^{K-1}} - \sigma^{K-1}\|\right]$ from Lemma \ref{lem:formal_slingshot_diff_noisy}:
\begin{align}
&\mathbb{E}\left[\|\pi^{\mu,\sigma^{K-1}} - \sigma^{K-1}\|\right] \nonumber\\
&\leq \sqrt{\frac{\|\pi^{\ast} - \sigma^0\|^2  + \frac{3}{\rho\kappa}\left(\mathrm{diam}(\mathcal{X}) + \frac{\zeta}{\mu} + 1\right)\left(\rho(2 \theta - \kappa)\mathrm{diam}(\mathcal{X})^2 + N C^2\left( \frac{1}{\kappa }\ln \left(\frac{\kappa}{2 \theta} T + 1\right) +  \frac{1}{2 \theta}\right)\right)}{K}}.
\label{eq:slingshot_diff_noisy}
\end{align}

By combining \eqref{eq:gap_fn_last_slingshot_strategy_noisy} and \eqref{eq:slingshot_diff_noisy}, we get:
\begin{align*}
&\mathbb{E}\left[\mathrm{GAP}(\sigma^K)\right] \\
&\leq \mu  \cdot \mathrm{diam}(\mathcal{X})\cdot \sqrt{\frac{\|\pi^{\ast} - \sigma^0\|^2  + \frac{3}{\rho\kappa}\left(\mathrm{diam}(\mathcal{X}) + \frac{\zeta}{\mu} + 1\right)\left(\rho(2 \theta - \kappa)\mathrm{diam}(\mathcal{X})^2 + N C^2\left( \frac{1}{\kappa }\ln \left(\frac{\kappa}{2 \theta} T + 1\right) +  \frac{1}{2 \theta}\right)\right)}{K}} \nonumber\\
&+ \left(L\cdot \mathrm{diam}(\mathcal{X}) + \zeta\right) \cdot \sqrt{\frac{\rho(2 \theta - \kappa)\mathrm{diam}(\mathcal{X})^2 + N C^2\left( \frac{1}{\kappa }\ln \left(\frac{\kappa}{2 \theta} T + 1\right) +  \frac{1}{2 \theta}\right)}{\rho\kappa T^{4/5}}}.
\end{align*}

Finally, since $\pi^T = \sigma^K$, $K=\lfloor T/T_{\sigma}\rfloor$, and $T_{\sigma} = c \cdot \max(T^{4/5} + 2, 3)$, we have:
\begin{align*}
&\mathbb{E}\left[\mathrm{GAP}(\pi^T)\right] \\
&\leq \mu  \cdot \mathrm{diam}(\mathcal{X})\cdot \sqrt{\frac{\|\pi^{\ast} - \sigma^0\|^2  + \frac{3}{\rho\kappa}\left(\mathrm{diam}(\mathcal{X}) + \frac{\zeta}{\mu} + 1\right)\left(\rho(2 \theta - \kappa)\mathrm{diam}(\mathcal{X})^2 + N C^2\left( \frac{1}{\kappa }\ln \left(\frac{\kappa}{2 \theta} T + 1\right) +  \frac{1}{2 \theta}\right)\right)}{T/T_{\sigma}}} \\
&+ \left(L\cdot \mathrm{diam}(\mathcal{X}) + \zeta\right) \cdot \sqrt{\frac{\rho(2 \theta - \kappa)\mathrm{diam}(\mathcal{X})^2 + N C^2\left( \frac{1}{\kappa }\ln \left(\frac{\kappa}{2 \theta} T + 1\right) +  \frac{1}{2 \theta}\right)}{\rho\kappa T^{4/5}}} \\
&\leq \mu  \cdot \mathrm{diam}(\mathcal{X})\cdot \sqrt{c (T^{4/5} + 5)} \\
&\cdot \sqrt{\frac{\|\pi^{\ast} - \sigma^0\|^2  + \frac{3}{\rho\kappa}\left(\mathrm{diam}(\mathcal{X}) + \frac{\zeta}{\mu} + 1\right)\left(\rho(2 \theta - \kappa)\mathrm{diam}(\mathcal{X})^2 + N C^2\left( \frac{1}{\kappa }\ln \left(\frac{\kappa}{2 \theta} T + 1\right) +  \frac{1}{2 \theta}\right)\right)}{T}} \\
&+ \left(L\cdot \mathrm{diam}(\mathcal{X}) + \zeta\right) \cdot \sqrt{\frac{\rho(2 \theta - \kappa)\mathrm{diam}(\mathcal{X})^2 + N C^2\left( \frac{1}{\kappa }\ln \left(\frac{\kappa}{2 \theta} T + 1\right) +  \frac{1}{2 \theta}\right)}{\rho\kappa T^{4/5}}} \\
&\leq \mu  \cdot \mathrm{diam}(\mathcal{X})\cdot \sqrt{ \frac{6c\|\pi^{\ast} - \sigma^0\|^2  + \frac{18c}{\rho\kappa}\left(\mathrm{diam}(\mathcal{X}) + \frac{\zeta}{\mu} + 1\right)\left(\rho(2 \theta - \kappa)\mathrm{diam}(\mathcal{X})^2 + N C^2\left( \frac{1}{\kappa }\ln \left(\frac{\kappa}{2 \theta} T + 1\right) +  \frac{1}{2 \theta}\right)\right)}{T^{1/5}}} \\
&+ \left(L\cdot \mathrm{diam}(\mathcal{X}) + \zeta\right) \cdot \sqrt{\frac{\rho(2 \theta - \kappa)\mathrm{diam}(\mathcal{X})^2 + N C^2\left( \frac{1}{\kappa }\ln \left(\frac{\kappa}{2 \theta} T + 1\right) +  \frac{1}{2 \theta}\right)}{\rho\kappa T^{4/5}}} \\
&\leq \mu  \cdot \mathrm{diam}(\mathcal{X})^2\cdot \sqrt{ \frac{6c}{T^{1/5}}} \\
&+ \mu  \cdot \mathrm{diam}(\mathcal{X})\cdot \sqrt{ \frac{18c\left(\mathrm{diam}(\mathcal{X}) + \frac{\zeta}{\mu} + 1\right)\left(\rho(2 \theta - \kappa)\mathrm{diam}(\mathcal{X})^2 + N C^2\left( \frac{1}{\kappa }\ln \left(\frac{\kappa}{2 \theta} T + 1\right) +  \frac{1}{2 \theta}\right)\right)}{\rho\kappa T^{1/5}}} \\
&+ \left(L\cdot \mathrm{diam}(\mathcal{X}) + \zeta\right) \cdot \sqrt{\frac{\rho(2 \theta - \kappa)\mathrm{diam}(\mathcal{X})^2 + N C^2\left( \frac{1}{\kappa }\ln \left(\frac{\kappa}{2 \theta} T + 1\right) +  \frac{1}{2 \theta}\right)}{\rho\kappa T^{1/5}}} \\
&\leq \frac{\sqrt{6c} \mu \cdot \mathrm{diam}(\mathcal{X})^2}{T^{1/10}} \\
&+ \frac{L \cdot \mathrm{diam}(\mathcal{X}) + \zeta + \mu \cdot \mathrm{diam}(\mathcal{X})\sqrt{18c\left(\mathrm{diam}(\mathcal{X}) + \frac{\zeta}{\mu} + 1\right)}}{T^{1/10}}\sqrt{\frac{\rho(2 \theta - \kappa)\mathrm{diam}(\mathcal{X})^2 + N C^2\left( \frac{1}{\kappa }\ln \left(\frac{\kappa}{2 \theta} T + 1\right) +  \frac{1}{2 \theta}\right)}{\rho \kappa}}.
\end{align*}
\end{proof}
This concludes the statement of the theorem.

\subsection{Proof of Lemma \ref{lem:formal_gd_sp_noisy} (Formal Version of Lemma \ref{lem:gd_sp_noisy})}
\label{sec:appx_gd_sp_noisy}
\begin{proof}[Proof of Lemma \ref{lem:formal_gd_sp_noisy}]
Assuming that $G$ is identical to $D_{\psi}$, Assumption \ref{asm:rel_smooth} is satisfied with $\beta=\gamma=1$.
Furthermore, since $\psi(x)=\frac{1}{2}\left\|x\right\|^2$, both $\rho=1$ and $\mathrm{int}(\mathrm{dom}~\psi)=\mathbb{R}^{d_i}$ hold.
Therefore, Assumption \ref{asm:well_defined} is also satisfied.
Consequently, we can apply Lemma \ref{lem:formal_md_sp_noisy} with $\beta=\gamma=\rho=1$ and obtain:
\begin{align*}
\mathbb{E}\left[\|\pi^{\mu,\sigma^k} - \sigma^{k+1}\|^2\right] \le \frac{2 \theta - \kappa}{\kappa (T_{\sigma} - 1) + 2 \theta} \mathbb{E}\left[\|\pi^{\mu,\sigma^k}-\sigma^k\|^2\right] + \frac{N C^2 }{\rho (\kappa (T_{\sigma} - 1) + 2 \theta)} \left( \frac{1}{\kappa }\ln \left(\frac{\kappa}{2 \theta} (T_{\sigma} - 1) + 1\right) +  \frac{1}{2 \theta}\right).
\end{align*}
\end{proof}

\subsection{Proof of Lemma \ref{lem:formal_slingshot_diff_noisy} (Formal Version of Lemma \ref{lem:slingshot_diff_noisy})}
\begin{proof}[Proof of Lemma \ref{lem:formal_slingshot_diff_noisy}]
From Lemmas \ref{lem:contraction_constrained} and \ref{lem:formal_gd_sp_noisy}, we have:
\begin{align*}
    &\mathbb{E}\left[\|\pi^{\mu, \sigma^k} - \sigma^k\|^2\right] \\
    &\leq \mathbb{E}\left[\|\sigma^k - \sigma^{k-1}\|^2\right] + \frac{2\zeta}{\mu} \mathbb{E}\left[\|\pi^{\mu, \sigma^{k-1}} - \sigma^k\|\right] \\
    &\leq \mathbb{E}\left[\|\sigma^k - \pi^{\mu, \sigma^{k-1}}\|^2\right] + \mathbb{E}\left[\|\pi^{\mu, \sigma^{k-1}} - \sigma^{k-1}\|^2\right] + \mathbb{E}\left[2\|\sigma^k - \pi^{\mu, \sigma^{k-1}}\|\|\pi^{\mu, \sigma^{k-1}} - \sigma^{k-1}\| + \frac{2\zeta}{\mu} \|\pi^{\mu, \sigma^{k-1}} - \sigma^k\|\right] \\
    &\leq \frac{2 \theta - \kappa}{\kappa (T_{\sigma} - 2) + 2 \theta} \mathbb{E}\left[\|\pi^{\mu,\sigma^k}-\sigma^k\|^2\right] + \frac{N C^2 }{\rho (\kappa (T_{\sigma} - 2) + 2 \theta)} \left( \frac{1}{\kappa }\ln \left(\frac{\kappa}{2 \theta} (T_{\sigma} - 2) + 1\right) +  \frac{1}{2 \theta}\right) \\
    & + 2 \left(\mathrm{diam}(\mathcal{X}) + \frac{\zeta}{\mu} \right)\sqrt{\frac{2 \theta - \kappa}{\kappa (T_{\sigma} - 2) + 2 \theta} \mathbb{E}\left[\|\pi^{\mu,\sigma^{k-1}}-\sigma^{k-1}\|^2\right] + \frac{N C^2 }{\rho (\kappa (T_{\sigma} - 2) + 2 \theta)} \left( \frac{1}{\kappa }\ln \left(\frac{\kappa}{2 \theta} (T_{\sigma} - 2) + 1\right) +  \frac{1}{2 \theta}\right)} \\
    & + \mathbb{E}\left[\|\pi^{\mu, \sigma^{k-1}} - \sigma^{k-1}\|^2\right] \\
    &\leq \mathbb{E}\left[\|\pi^{\mu, \sigma^{k-1}} - \sigma^{k-1}\|^2\right] + \frac{\rho(2 \theta - \kappa)\mathrm{diam}(\mathcal{X})^2 + N C^2\left( \frac{1}{\kappa }\ln \left(\frac{\kappa}{2 \theta} (T_{\sigma} - 2) + 1\right) +  \frac{1}{2 \theta}\right)}{\rho(\kappa (T_{\sigma} - 2) + 2 \theta)} \\
    & + 2 \left(\mathrm{diam}(\mathcal{X}) + \frac{\zeta}{\mu} \right)\sqrt{\frac{\rho(2 \theta - \kappa)\mathrm{diam}(\mathcal{X})^2 + N C^2\left( \frac{1}{\kappa }\ln \left(\frac{\kappa}{2 \theta} (T_{\sigma} - 2) + 1\right) +  \frac{1}{2 \theta}\right)}{\rho(\kappa (T_{\sigma} - 2) + 2 \theta)}}.
\end{align*}
Under the assumption that $T_{\sigma} \geq \max(T^{4/5} + 2, 3)$, we get:
\begin{align*}
    \mathbb{E}\left[\|\pi^{\mu, \sigma^k} - \sigma^k\|^2\right] &\leq \mathbb{E}\left[\|\pi^{\mu, \sigma^{k-1}} - \sigma^{k-1}\|^2\right] + \frac{\rho(2 \theta - \kappa)\mathrm{diam}(\mathcal{X})^2 + N C^2\left( \frac{1}{\kappa }\ln \left(\frac{\kappa}{2 \theta} T + 1\right) +  \frac{1}{2 \theta}\right)}{\rho(\kappa T^{4/5} + 2 \theta)} \\
    & + 2 \left(\mathrm{diam}(\mathcal{X}) + \frac{\zeta}{\mu} \right)\sqrt{\frac{\rho(2 \theta - \kappa)\mathrm{diam}(\mathcal{X})^2 + N C^2\left( \frac{1}{\kappa }\ln \left(\frac{\kappa}{2 \theta} T + 1\right) +  \frac{1}{2 \theta}\right)}{\rho(\kappa T^{4/5} + 2 \theta)}}.
\end{align*}
Therefore, we get:
\begin{align}
    \mathbb{E}[\|\pi^{\mu, \sigma^{K-1}} - \sigma^{K-1}\|^2] &\leq \mathbb{E}[\|\pi^{\mu, \sigma^{K-2}} - \sigma^{K-2}\|^2] + \frac{\rho(2 \theta - \kappa)\mathrm{diam}(\mathcal{X})^2 + N C^2\left( \frac{1}{\kappa }\ln \left(\frac{\kappa}{2 \theta} T + 1\right) +  \frac{1}{2 \theta}\right)}{\rho(\kappa T^{4/5} + 2 \theta)} \nonumber\\
    & + 2 \left(\mathrm{diam}(\mathcal{X}) + \frac{\zeta}{\mu} \right)\sqrt{\frac{\rho(2 \theta - \kappa)\mathrm{diam}(\mathcal{X})^2 + N C^2\left( \frac{1}{\kappa }\ln \left(\frac{\kappa}{2 \theta} T + 1\right) +  \frac{1}{2 \theta}\right)}{\rho(\kappa T^{4/5} + 2 \theta)}} \nonumber \\
    &\leq \mathbb{E}[ \|\pi^{\mu, \sigma^k} - \sigma^k\|^2] + K\frac{\rho(2 \theta - \kappa)\mathrm{diam}(\mathcal{X})^2 + N C^2\left( \frac{1}{\kappa }\ln \left(\frac{\kappa}{2 \theta} T + 1\right) +  \frac{1}{2 \theta}\right)}{\rho(\kappa T^{4/5} + 2 \theta)} \nonumber\\
    & + 2 K\left(\mathrm{diam}(\mathcal{X}) + \frac{\zeta}{\mu} \right)\sqrt{\frac{\rho(2 \theta - \kappa)\mathrm{diam}(\mathcal{X})^2 + N C^2\left( \frac{1}{\kappa }\ln \left(\frac{\kappa}{2 \theta} T + 1\right) +  \frac{1}{2 \theta}\right)}{\rho(\kappa T^{4/5} + 2 \theta)}}.
\label{eq:diff_p_sigma_noisy}
\end{align}

Here, we derive the following upper bound in terms of $\mathbb{E}[ \|\pi^{\mu, \sigma^k} - \sigma^k\|^2]$:
\begin{lemma}
\label{lem:diff_p_sigma_sum_noisy}
Assume that $T_{\sigma} \geq \max(T^{4/5} + 2, 3)$.
In the same setup of Theorem \ref{thm:lic_rate_noisy}, we have for any Nash equilibrium $\pi^{\ast} \in \Pi^{\ast}$:
\begin{align*}
    &\mathbb{E}\left[\sum_{k=0}^{K-1}\|\pi^{\mu, \sigma^k} - \sigma^k\|^2\right] \leq \|\pi^{\ast} - \sigma^0\|^2 + K\cdot \mathrm{diam}(\mathcal{X}) \cdot \sqrt{\frac{\rho(2 \theta - \kappa)\mathrm{diam}(\mathcal{X})^2 + N C^2\left( \frac{1}{\kappa }\ln \left(\frac{\kappa}{2 \theta} T + 1\right) +  \frac{1}{2 \theta}\right)}{\rho(\kappa T^{4/5} + 2 \theta)}}.
\end{align*}
\end{lemma}

By combining \eqref{eq:diff_p_sigma_noisy}, Lemma \ref{lem:diff_p_sigma_sum_noisy}, and the assumption that $T_{\sigma} \geq \max(T^{4/5} + 2, 3)$, we have:
\begin{align*}
    &K \mathbb{E}[\|\pi^{\mu, \sigma^{K-1}} - \sigma^{K-1}\|^2] \\
    &\leq \mathbb{E}\left[\sum_{k=0}^{K-1}\|\pi^{\mu, \sigma^k} - \sigma^k\|^2\right] + K^2\frac{\rho(2 \theta - \kappa)\mathrm{diam}(\mathcal{X})^2 + N C^2\left( \frac{1}{\kappa }\ln \left(\frac{\kappa}{2 \theta} T + 1\right) +  \frac{1}{2 \theta}\right)}{\rho(\kappa T^{4/5} + 2 \theta)} \nonumber\\
    & + 2 K^2\left(\mathrm{diam}(\mathcal{X}) + \frac{\zeta}{\mu} \right)\sqrt{\frac{\rho(2 \theta - \kappa)\mathrm{diam}(\mathcal{X})^2 + N C^2\left( \frac{1}{\kappa }\ln \left(\frac{\kappa}{2 \theta} T + 1\right) +  \frac{1}{2 \theta}\right)}{\rho(\kappa T^{4/5} + 2 \theta)}} \\
    &\leq \|\pi^{\ast} - \sigma^0\|^2 + K\cdot \mathrm{diam}(\mathcal{X}) \cdot \sqrt{\frac{\rho(2 \theta - \kappa)\mathrm{diam}(\mathcal{X})^2 + N C^2\left( \frac{1}{\kappa }\ln \left(\frac{\kappa}{2 \theta} (T - 1) + 1\right) +  \frac{1}{2 \theta}\right)}{\rho(\kappa T^{4/5} + 2 \theta)}} \\
    &+ K^2\frac{\rho(2 \theta - \kappa)\mathrm{diam}(\mathcal{X})^2 + N C^2\left( \frac{1}{\kappa }\ln \left(\frac{\kappa}{2 \theta} T + 1\right) +  \frac{1}{2 \theta}\right)}{\rho(\kappa T^{4/5} + 2 \theta)} \nonumber\\
    & + 2 K^2\left(\mathrm{diam}(\mathcal{X}) + \frac{\zeta}{\mu} \right)\sqrt{\frac{\rho(2 \theta - \kappa)\mathrm{diam}(\mathcal{X})^2 + N C^2\left( \frac{1}{\kappa }\ln \left(\frac{\kappa}{2 \theta} T + 1\right) +  \frac{1}{2 \theta}\right)}{\rho(\kappa T^{4/5} + 2 \theta)}} \\ \\
    &\leq \|\pi^{\ast} - \sigma^0\|^2 + T^{1/5}\cdot \mathrm{diam}(\mathcal{X}) \cdot \sqrt{\frac{\rho(2 \theta - \kappa)\mathrm{diam}(\mathcal{X})^2 + N C^2\left( \frac{1}{\kappa }\ln \left(\frac{\kappa}{2 \theta} (T - 1) + 1\right) +  \frac{1}{2 \theta}\right)}{\rho(\kappa T^{4/5} + 2 \theta)}} \\
    &+ T^{2/5}\frac{\rho(2 \theta - \kappa)\mathrm{diam}(\mathcal{X})^2 + N C^2\left( \frac{1}{\kappa }\ln \left(\frac{\kappa}{2 \theta} T + 1\right) +  \frac{1}{2 \theta}\right)}{\rho(\kappa T^{4/5} + 2 \theta)} \nonumber\\
    & + 2T^{2/5}\left(\mathrm{diam}(\mathcal{X}) + \frac{\zeta}{\mu} \right)\sqrt{\frac{\rho(2 \theta - \kappa)\mathrm{diam}(\mathcal{X})^2 + N C^2\left( \frac{1}{\kappa }\ln \left(\frac{\kappa}{2 \theta} T + 1\right) +  \frac{1}{2 \theta}\right)}{\rho(\kappa T^{4/5} + 2 \theta)}} \\ \\
    &\leq \|\pi^{\ast} - \sigma^0\|^2  + 3\left(\mathrm{diam}(\mathcal{X}) + \frac{\zeta}{\mu} + 1\right)\left(\frac{\rho(2 \theta - \kappa)\mathrm{diam}(\mathcal{X})^2 + N C^2\left( \frac{1}{\kappa }\ln \left(\frac{\kappa}{2 \theta} T + 1\right) +  \frac{1}{2 \theta}\right)}{\rho\kappa}\right).
\end{align*}

\end{proof}

\section{Proofs for Section \ref{sec:discussion}}

\subsection{Proof of Lemma \ref{lem:formal_md_sp_full} (Formal Version of Lemma \ref{lem:md_sp_full})}
\begin{proof}[Proof of Lemma \ref{lem:formal_md_sp_full}]
From the definition of the Bregman divergence, we have for any $t\in \{kT_{\sigma}, kT_{\sigma}+1, \cdots, (k+1)T_{\sigma} - 1\}$:
\begin{align}
    &D_{\psi}(\pi_i^{\mu,\sigma^k}, \pi_i^{t+1}) - D_{\psi}(\pi_i^{\mu,\sigma^k}, \pi_i^t) + D_{\psi}(\pi_i^{t+1}, \pi_i^t) \nonumber\\
    &= \psi(\pi_i^{\mu,\sigma^k}) - \psi(\pi_i^{t+1}) - \langle \nabla \psi(\pi_i^{t+1}), \pi_i^{\mu,\sigma^k} - \pi_i^{t+1}\rangle \nonumber\\
    &- \psi(\pi_i^{\mu,\sigma^k}) + \psi(\pi_i^t) + \langle \nabla\psi(\pi_i^t), \pi_i^{\mu,\sigma^k} - \pi_i^t\rangle \nonumber\\
    &+ \psi(\pi_i^{t+1}) - \psi(\pi_i^t) - \langle \nabla\psi(\pi_i^t), \pi_i^{t+1} - \pi_i^t\rangle \nonumber\\
    &= \langle \nabla\psi(\pi_i^t) - \nabla\psi(\pi_i^{t+1}), \pi_i^{\mu,\sigma^k} - \pi_i^{t+1}\rangle.
    \label{eq:three_point_bregman_div}
\end{align}
From the first-order optimality condition for $\pi_i^{t+1}$, we get for any $t\in \{kT_{\sigma}, kT_{\sigma}+1, \cdots, (k+1)T_{\sigma} - 1\}$:
\begin{align}
    \langle \eta(\nabla_{\pi_i}v_i(\pi_i^t, \pi_{-i}^t) - \mu\nabla_{\pi_i}G(\pi_i^t, \sigma_i^k)) - \nabla \psi(\pi_i^{t+1}) + \nabla \psi(\pi_i^t), \pi_i^{\mu,\sigma^k} - \pi_i^{t+1}\rangle \leq 0.
    \label{eq:opt_cond_t_md}
\end{align}
Note that $\nabla_{\pi_i}G(\pi_i^t, \sigma_i^k)$ and $\nabla_{\pi_i}\psi(\pi_i^t)$ are well-defined because of Assumptions \ref{asm:rel_smooth} and \ref{asm:well_defined}.
By combining \eqref{eq:three_point_bregman_div} and \eqref{eq:opt_cond_t_md}, we have:
\begin{align}
D_{\psi}(\pi_i^{\mu,\sigma^k}, \pi_i^{t+1}) - D_{\psi}(\pi_i^{\mu,\sigma^k}, \pi_i^t) + D_{\psi}(\pi_i^{t+1}, \pi_i^t) \leq \eta\langle \nabla_{\pi_i}v_i(\pi_i^t, \pi_{-i}^t) - \mu\nabla_{\pi_i}G(\pi_i^t, \sigma_i^k), \pi_i^{t+1} - \pi_i^{\mu,\sigma^k}\rangle.
\label{eq:three_point_eq_md_full}
\end{align}

Next, we derive the following convergence result for $\pi^t$:
\begin{lemma}
\label{lem:grad_breg_full}
Suppose that Assumption \ref{asm:rel_smooth} holds with $\beta, \gamma\in (0, \infty)$, and the updated strategy profile $\pi^t$ satisfies the following condition: for any $t\in \{kT_{\sigma}, kT_{\sigma}+1, \cdots, (k+1)T_{\sigma} - 1\}$,
\begin{align*}
D_{\psi}(\pi^{\mu,\sigma^k}, \pi^{t+1}) - D_{\psi}(\pi^{\mu,\sigma^k}, \pi^t) + D_{\psi}(\pi^{t+1}, \pi^t) \leq \eta\sum_{i=1}^N \langle \nabla_{\pi_i}v_i(\pi_i^t, \pi_{-i}^t) - \mu\nabla_{\pi_i}G(\pi_i^t, \sigma_i^k), \pi_i^{t+1} - \pi_i^{\mu,\sigma^k}\rangle.
\end{align*}
Then, for any $t\in \{kT_{\sigma}, kT_{\sigma}+1, \cdots, (k+1)T_{\sigma} - 1\}$:
\begin{align*}
    D_{\psi}(\pi^{\mu,\sigma^k}, \pi^{t+1}) \leq D_{\psi}(\pi^{\mu,\sigma^k}, \pi^{kT_{\sigma}})\left(1 - \frac{\eta \mu\gamma}{2}\right)^{t - kT_{\sigma}+1}.
\end{align*}
\end{lemma}

It is easy to confirm that \eqref{eq:three_point_eq_md_full} satisfies the assumption in Lemma \ref{lem:grad_breg_full}.
Thus, taking $t=(k+1)T_{\sigma}-1$, we have:
\begin{align*}
    D_{\psi}(\pi^{\mu,\sigma^k}, \pi^{(k+1)T_{\sigma}}) \leq D_{\psi}(\pi^{\mu,\sigma^k}, \pi^{kT_{\sigma}})\left(1 - \frac{\eta \mu\gamma}{2}\right)^{T_{\sigma}}.
\end{align*}
Since $\sigma^k = \pi^{kT_{\sigma}}$ and $\sigma^{k+1} = \pi^{(k+1)T_{\sigma}}$, we conclude the statement.
\end{proof}

\subsection{Proof of Lemma~\ref{lem:formal_md_sp_noisy} (Formal Version of Lemma~\ref{lem:md_sp_noisy})}
\label{app:prf_gm_md_noisy}
\begin{proof}[Proof of Lemma~\ref{lem:formal_md_sp_noisy}]
Writing ${g}_i^t = \nabla_{\pi_i}v_i(\pi_i^t, \pi_{-i}^t) - \mu\nabla_{\pi_i}G(\pi_i^t, \sigma_i^k)$, from the first-order optimality condition for $\pi_i^{t+1}$, we get for any $t\in \{kT_{\sigma}, kT_{\sigma}+1, \cdots, (k+1)T_{\sigma} - 1\}$:
\begin{align}
    \langle \eta_t({g}_i^t +\xi_i^t) - \nabla \psi(\pi_i^{t+1}) + \nabla \psi(\pi_i^t), \pi_i^{\mu,\sigma^k} - \pi_i^{t+1}\rangle \leq 0.
    \label{eq:opt_cond_t_md_noise}
\end{align}
Note that $\nabla_{\pi_i}G(\pi_i^t, \sigma_i^k)$ and $\nabla_{\pi_i}\psi(\pi_i^t)$ are well-defined because of Assumptions \ref{asm:rel_smooth} and \ref{asm:well_defined}.
By combining \eqref{eq:three_point_bregman_div} and \eqref{eq:opt_cond_t_md_noise}, we have for any $t\in \{kT_{\sigma}, kT_{\sigma}+1, \cdots, (k+1)T_{\sigma} - 1\}$:
\begin{align}
    D_{\psi}(\pi_i^{\mu,\sigma^k}, \pi_i^{t+1}) - D_{\psi}(\pi_i^{\mu,\sigma^k}, \pi_i^t) + D_{\psi}(\pi_i^{t+1}, \pi_i^t) \leq \eta_t\langle g_i^t + \xi_i^t, \pi_i^{t+1} - \pi_i^{\mu,\sigma^k}\rangle.
\label{eq:three_point_eq_md_noisy}
\end{align}

We have the following Lemma that replaces the gradient with Bregman divergences:
\begin{lemma}
\label{lem:grad_breg_noise}
Under the noisy feedback setting, suppose that Assumption \ref{asm:rel_smooth} holds with $\beta, \gamma\in (0, \infty)$, and the updated strategy profile $\pi^t$ satisfies the following condition: for any $t\in \{kT_{\sigma}, kT_{\sigma}+1, \cdots, (k+1)T_{\sigma} - 1\}$,
\begin{align*}
    & D_{\psi}(\pi^{\mu,\sigma^k}, \pi^{t+1}) - D_{\psi}(\pi^{\mu,\sigma^k}, \pi^t) + D_{\psi}(\pi^{t+1}, \pi^t) \\
    & \quad \le \eta_t \sum_{i=1}^N \langle \nabla_{\pi_i}v_i(\pi_i^t, \pi_{-i}^t) - \mu\nabla_{\pi_i}G(\pi_i^t, \sigma_i^k) + \xi_i^t,  \pi_i^{t+1} - \pi_i^{\mu,\sigma^k}\rangle.
\end{align*}
Then, for any $t\in \{kT_{\sigma}, kT_{\sigma}+1, \cdots, (k+1)T_{\sigma} - 1\}$:
\begin{align*}
& D_{\psi}(\pi^{\mu,\sigma^k}, \pi^{t+1}) - D_{\psi}(\pi^{\mu,\sigma^k}, \pi^t) + D_{\psi}(\pi^{t+1}, \pi^t) \\
& \le \eta_t \left(\left( \frac{\mu^2 \gamma \rho^2 (\gamma + 2\beta) + 8 L^2}{2 \mu \gamma \rho^2} \right) D_{\psi}(\pi^{t+1}, \pi^t)  - \frac{\mu\gamma}{2} D_{\psi}(\pi^{\mu,\sigma^k}, \pi^t) \right)+ \eta_t \sum_{i=1}^N\langle \xi_i^t, \pi_i^{t+1} - \pi_i^{\mu,\sigma^k}\rangle.
\end{align*}
\end{lemma}
It is easy to confirm that \eqref{eq:three_point_eq_md_noisy} satisfies the assumption in Lemma \ref{lem:grad_breg_noise}.
Thus, for any $t\in \{kT_{\sigma}, kT_{\sigma}+1, \cdots, (k+1)T_{\sigma} - 1\}$:
\begin{align*}
& D_{\psi}(\pi^{\mu,r}, \pi^{t+1}) - D_{\psi}(\pi^{\mu,\sigma^k}, \pi^t) + D_{\psi}(\pi^{t+1}, \pi^t) \\
& \le \eta_t \left(\left( \frac{\mu^2 \gamma \rho^2 (\gamma + 2\beta) + 8 L^2}{2 \mu \gamma \rho^2} \right) D_{\psi}(\pi^{t+1}, \pi^t)  - \frac{\mu\gamma}{2} D_{\psi}(\pi^{\mu,\sigma^k}, \pi^t) \right)+ \eta_t \sum_{i=1}^N\langle \xi_i^t, \pi_i^{t+1} - \pi_i^{\mu,\sigma^k}\rangle
\\
& = \eta_t(\theta D_\psi(\pi^{t+1}, \pi^{t}) - \kappa D_\psi(\pi^{\mu,\sigma^k}, \pi^{t})) + \eta_t \sum_{i=1}^N\langle \xi_i^t,  \pi_i^{t+1} - \pi_i^{\mu,\sigma^k}\rangle.
\end{align*}

Then, using this inequality, we can bound the expected value of $D_{\psi}(\pi^{\mu, \sigma^k}, \pi^{t+1})$ as follows:
\begin{lemma}
\label{lem:noisy_tel_sum}
    Suppose that with some constants $\theta > \kappa>0$, for all $t\in \{kT_{\sigma}, kT_{\sigma}+1, \cdots, (k+1)T_{\sigma} - 1\}$, the following inequality holds:
\begin{align*}
    & D_{\psi}(\pi^{\mu,\sigma^k}, \pi^{t+1}) - D_{\psi}(\pi^{\mu,\sigma^k}, \pi^t) + D_{\psi}(\pi^{t+1}, \pi^t)  
    \\
    & \quad \le \eta_t(\theta D_\psi(\pi^{t+1}, \pi^{t}) - \kappa D_\psi(\pi^{\mu,\sigma^k}, \pi^{t})) + \eta_t \sum_{i=1}^N\langle \xi_i^t,  \pi_i^{t+1} - \pi_i^{\mu,\sigma^k}\rangle.
\end{align*}
Then, under Assumption~\ref{asm:noise}, for any $t\in \{kT_{\sigma}, kT_{\sigma}+1, \cdots, (k+1)T_{\sigma} - 1\}$,
\begin{align*}
   &\mathbb{E}[D_{\psi}(\pi^{\mu,\sigma^k}, \pi^{t+1})] \\
    &\le \frac{2 \theta - \kappa}{\kappa (t- kT_{\sigma}) + 2 \theta} \mathbb{E}[D_\psi(\pi^{\mu,\sigma^k}, \pi^{kT_{\sigma}})] + \frac{N C^2 }{\rho (\kappa (t- kT_{\sigma}) + 2 \theta)} \left( \frac{1}{\kappa }\ln \left(\frac{\kappa}{2 \theta} (t- kT_{\sigma}) + 1\right) +  \frac{1}{2 \theta}\right).
\end{align*}
\end{lemma}

Taking $t=(k+1)T_{\sigma}-1$, we have:
\begin{align*}
&\mathbb{E}[D_{\psi}(\pi^{\mu,\sigma^k}, \pi^{(k+1)T_{\sigma}})] \\
&\le \frac{2 \theta - \kappa}{\kappa (T_{\sigma} - 1) + 2 \theta} \mathbb{E}[D_\psi(\pi^{\mu,\sigma^k}, \pi^{kT_{\sigma}})] + \frac{N C^2 }{\rho (\kappa (T_{\sigma} - 1) + 2 \theta)} \left( \frac{1}{\kappa }\ln \left(\frac{\kappa}{2 \theta} (T_{\sigma} - 1) + 1\right) +  \frac{1}{2 \theta}\right).
\end{align*}
Since $\sigma^k = \pi^{kT_{\sigma}}$ and $\sigma^{k+1} = \pi^{(k+1)T_{\sigma}}$, we conclude the statement.
\end{proof}

\subsection{Proof of Theorem \ref{thm:exact_conv_bregman}}
\label{sec:appx_exact_conv_bregman}
\begin{proof}[Proof of Theorem \ref{thm:exact_conv_bregman}]
When $G(\pi_i, \pi_i')=D_{\psi'}(\pi_i, \pi_i')$ for all $i\in [N]$ and $\pi,\pi'\in \mathcal{X}$, we can show that the Bregman divergence from a Nash equilibrium $\pi^{\ast}\in \Pi^{\ast}$ to $\sigma^{k+1}$ monotonically decreases:
\begin{lemma}
\label{lem:exact_conv_bregman}
Assume that $G$ is a Bregman divergence $D_{\psi'}$ for some strongly convex function $\psi'$.
Then, for any Nash equilibrium $\pi^{\ast} \in \Pi^{\ast}$ of the underlying game, we have for any $k\geq 0$:
\begin{align*}
    D_{\psi'}(\pi^{\ast}, \sigma^{k+1}) - D_{\psi'}(\pi^{\ast}, \sigma^k) &\leq -D_{\psi'}(\sigma^{k+1}, \sigma^k).
\end{align*}
\end{lemma}
By summing the inequality in Lemma \ref{lem:exact_conv_bregman} from $k=0$ to $K$, we have:
\begin{align*}
     D_{\psi'}(\pi^{\ast}, \sigma^0) \geq \sum_{k=0}^{K} D_{\psi'}(\sigma^{k+1}, \sigma^k) \geq \frac{\rho}{2}\sum_{k=0}^K \|\sigma^{k+1} - \sigma^k\|^2,
\end{align*}
where the second inequality follows from the strong convexity of $\psi'$.
Therefore, $\sum_{k=0}^{\infty}\|\sigma^{k+1} - \sigma^k\|^2 < \infty$, which implies that $\|\sigma^{k+1} - \sigma^k\| \to 0$ as $k\to \infty$.

By the compactness of $\mathcal{X}$ and Bolzano–Weierstrass theorem, there exists a subsequence $k_n$ and a limit point $\hat{\sigma}\in \mathcal{X}$ such that $\sigma^{k_n}\to \hat{\sigma}$ as $n\to \infty$.
Since $\|\sigma^{k_n+1} - \sigma^{k_n}\| \to 0$ as $n\to \infty$, we have $\sigma^{k_n+1} \to \hat{\sigma}$ as $n\to \infty$.
Thus, the limit point $\hat{\sigma}$ is the fixed point of the updating rule.
From the following lemma, we show that the fixed point $\hat{\sigma}$ is a Nash equilibrium of the underlying game:
\begin{lemma}
\label{lem:stop_condition_bregman}
Assume that $G$ is a Bregman divergence $D_{\psi'}$ for some strongly convex function $\psi'$, and $\sigma^{k+1}=\pi^{\mu, \sigma^k}$ for $k\geq 0$.
If $\sigma^{k+1} = \sigma^k$, then $\sigma^k$ is a Nash equilibrium of the underlying game.
\end{lemma}

On the other hand, by summing the inequality in Lemma \ref{lem:exact_conv_bregman} from $k=k_n$ to $k=K-1$ for $K\geq k_n + 1$, we have:
\begin{align*}
    0 \leq D_{\psi'}(\hat{\sigma}, \sigma^K) &\leq D_{\psi'}(\hat{\sigma}, \sigma^{k_n}).
\end{align*}
Since $\sigma^{k_n}\to \hat{\sigma}$ as $n\to \infty$, we have $\sigma^K \to \hat{\sigma}$ as $K\to \infty$.
Since $\hat{\sigma}$ is a Nash equilibrium of the underlying game, we conclude the first statement of the theorem.

\end{proof}

\subsection{Proof of Theorem \ref{thm:exact_conv_alpha}}
\label{sec:appx_exact_conv_alpha}
\begin{proof}[Proof of Theorem \ref{thm:exact_conv_alpha}]
We first show that the divergence between $\Pi^{\ast}$ and $\sigma^k$ decreases monotonically as $k$ increases:
\begin{lemma}
\label{lem:exact_conv_min_alpha}
Suppose that the same assumptions in Theorem \ref{thm:exact_conv_alpha} hold.
For any $k\geq 0$, if $\sigma^k\in \mathcal{X} \setminus \Pi^{\ast}$, then:
\begin{align*}
    \min_{\pi^{\ast}\in \Pi^{\ast}} \mathrm{KL}(\pi^{\ast}, \sigma^{k+1}) < \min_{\pi^{\ast}\in \Pi^{\ast}} \mathrm{KL}(\pi^{\ast}, \sigma^k).
\end{align*}
Otherwise, if $\sigma^k\in \Pi^{\ast}$, then $\sigma^{k+1}=\sigma^k \in \Pi^{\ast}$.
\end{lemma}

From Lemma \ref{lem:exact_conv_min_alpha}, the sequence $\{\min_{\pi^{\ast}\in \Pi^{\ast}} \mathrm{KL}(\pi^{\ast}, \sigma^k)\}_{k\geq 0}$ is a monotonically decreasing sequence and is bounded from below by zero.
Thus, $\{\min_{\pi^{\ast}\in \Pi^{\ast}} \mathrm{KL}(\pi^{\ast}, \sigma^k)\}_{k\geq 0}$ converges to some constant $b\geq 0$.
We show that $b=0$ by a contradiction argument.

Suppose $b>0$ and let us define $B=\min_{\pi^{\ast}\in \Pi^{\ast}} \mathrm{KL}(\pi^{\ast}, \sigma^0)$.
Since $\min_{\pi^{\ast}\in \Pi^{\ast}} \mathrm{KL}(\pi^{\ast}, \sigma^k)$ monotonically decreases, $\sigma^k$ is in the set $\Omega_{b, B} = \{\sigma\in \mathcal{X} ~|~ b \leq \min_{\pi^{\ast}\in \Pi^{\ast}} \mathrm{KL}(\pi^{\ast}, \sigma) \leq B\}$ for all $k\geq 0$.
Since $\min_{\pi^{\ast}\in \Pi^{\ast}} \mathrm{KL}(\pi^{\ast}, \cdot)$ is a continuous function on $\mathcal{X}$, the preimage $\Omega_{b, B}$ of the closed set $[b, B]$ is also closed.
Furthermore, since $\mathcal{X}$ is compact and then bounded, $\Omega_{b, B}$ is a bounded set.
Thus, $\Omega_{b, B}$ is a compact set.

Next, we show that the function which maps the slingshot strategies $\sigma$ to the associated stationary point $\pi^{\mu,\sigma}$ is continuous:
\begin{lemma}
\label{lem:map_continuity}
Let $F(\sigma):\mathcal{X}\to \mathcal{X}$ be a function that maps the slingshot strategies $\sigma$ to the stationary point $\pi^{\mu,\sigma}$ defined by \eqref{eq:perturbed_nash}.
In the same setup of Theorem \ref{thm:exact_conv_alpha}, $F(\cdot)$ is a continuous function on $\mathcal{X}$.
\end{lemma}

From Lemma \ref{lem:map_continuity}, $\min_{\pi^{\ast}\in \Pi^{\ast}} \mathrm{KL}(\pi^{\ast}, F(\sigma)) - \min_{\pi^{\ast}\in \Pi^{\ast}} \mathrm{KL}(\pi^{\ast}, \sigma)$ is also a continuous function.
Since a continuous function has a maximum over a compact set $\Omega_{b, B}$, the maximum $\delta = \max_{\sigma\in \Omega_{b, B}}\left\{\min_{\pi^{\ast}\in \Pi^{\ast}} \mathrm{KL}(\pi^{\ast}, F(\sigma)) - \min_{\pi^{\ast}\in \Pi^{\ast}} \mathrm{KL}(\pi^{\ast}, \sigma)\right\}$ exists.
From Lemma \ref{lem:exact_conv_min_alpha} and the assumption that $b>0$, we have $\delta < 0$.
It follows that:
\begin{align*}
    \min_{\pi^{\ast}\in \Pi^{\ast}} \mathrm{KL}(\pi^{\ast}, \sigma^k) &= \min_{\pi^{\ast}\in \Pi^{\ast}} \mathrm{KL}(\pi^{\ast}, \sigma^0) + \sum_{l=0}^{k-1}\left(\min_{\pi^{\ast}\in \Pi^{\ast}} \mathrm{KL}(\pi^{\ast}, \sigma^{l+1}) - \min_{\pi^{\ast}\in \Pi^{\ast}} \mathrm{KL}(\pi^{\ast}, \sigma^l)\right) \\
    &\leq B + \sum_{l=0}^{k-1}\delta = B + k\delta.
\end{align*}
This implies that $\min_{\pi^{\ast}\in \Pi^{\ast}} \mathrm{KL}(\pi^{\ast}, \sigma^k) < 0$ for $k > \frac{B}{-\delta}$, which contradicts $\min_{\pi^{\ast}\in \Pi^{\ast}} \mathrm{KL}(\pi^{\ast}, \sigma) \geq 0$.
Therefore, the sequence $\{\min_{\pi^{\ast}\in \Pi^{\ast}} \mathrm{KL}(\pi^{\ast}, \sigma^k)\}_{k\geq 0}$ converges to $0$, and $\sigma^k$ converges to $\Pi^{\ast}$.
\end{proof}

\subsection{Proof of Lemma \ref{lem:formal_ftrl_sp_full}}
\label{sec:appx_frtl_sp_full}
\begin{proof}[Proof of Lemma \ref{lem:formal_ftrl_sp_full}]
First, we introduce the following lemma:
\begin{lemma}
\label{lem:legendre}
Let us define $T(y_i) = \argmax_{x\in \mathcal{X}_i}\{\langle y_i, x\rangle -\psi(x)\}$.
Assuming $\psi: \mathcal{X}_i\to \mathbb{R}$ be a convex function of the Legendre type, we have for any $\pi_i\in \mathcal{X}_i$:
\begin{align*}
    D_{\psi}(\pi_i, T(y_i)) = \psi(\pi_i) - \psi(T(y_i)) - \langle y_i, \pi_i - T(y_i)\rangle.
\end{align*}
\end{lemma}
Defining $y_i^t=\eta \sum_{s=0}^t\left(\nabla_{\pi_i}v_i(\pi_i^s, \pi_{-i}^s) - \mu \nabla_{\pi_i}G(\pi_i^s, \sigma_i^k)\right)$ and letting $\pi_i=\pi_i^{\mu,\sigma^k}$, $y_i=y_i^t$ in Lemma \ref{lem:legendre}, we have:
\begin{align*}
    D_{\psi}(\pi_i^{\mu,\sigma^k}, \pi_i^{t+1}) = \psi(\pi_i^{\mu, \sigma^k}) - \psi(\pi_i^{t+1}) - \langle y_i^t, \pi_i^{\mu,\sigma^k} - \pi_i^{t+1}\rangle.
\end{align*}
Note that $\nabla_{\pi_i}G(\pi_i^s, \sigma_i^k)$ is well-defined because of Assumptions \ref{asm:rel_smooth} and \ref{asm:well_defined_ftrl}.
Using this equation, we get for any $t\in \{kT_{\sigma}, kT_{\sigma}+1, \cdots, (k+1)T_{\sigma} - 1\}$:
\begin{align}
    &D_{\psi}(\pi_i^{\mu,\sigma^k}, \pi_i^{t+1}) - D_{\psi}(\pi_i^{\mu,\sigma^k}, \pi_i^t) + D_{\psi}(\pi_i^{t+1}, \pi_i^t) \nonumber\\
    &= \psi(\pi_i^{\mu,\sigma^k}) - \psi(\pi_i^{t+1}) - \langle y_i^t, \pi_i^{\mu,\sigma^k} - \pi_i^{t+1}\rangle - \psi(\pi_i^{\mu,\sigma^k}) + \psi(\pi_i^t) + \langle y_i^{t-1}, \pi_i^{\mu,\sigma^k} - \pi_i^t\rangle \nonumber\\
    &+ \psi(\pi_i^{t+1}) - \psi(\pi_i^t) - \langle y_i^{t-1}, \pi_i^{t+1} - \pi_i^t\rangle \nonumber\\
    &= - \langle y_i^t, \pi_i^{\mu,\sigma^k} - \pi_i^{t+1}\rangle + \langle y_i^{t-1}, \pi_i^{\mu,\sigma^k} - \pi_i^{t+1}\rangle \nonumber\\
    &= \langle y_i^t - y_i^{t-1}, \pi_i^{t+1} - \pi_i^{\mu,\sigma^k}\rangle \nonumber\\
    &= \eta \langle \nabla_{\pi_i}v_i(\pi_i^t, \pi_{-i}^t) - \mu \nabla_{\pi_i}G(\pi_i^t, \sigma_i^k), \pi_i^{t+1} - \pi_i^{\mu,\sigma^k}\rangle.
    \label{eq:three_point_eq_ftrl}
\end{align}

It is easy to confirm that \eqref{eq:three_point_eq_ftrl} satisfies the assumption in Lemma \ref{lem:grad_breg_full}.
Thus, taking $t=(k+1)T_{\sigma}-1$ in Lemma \ref{lem:grad_breg_full}, we have:
\begin{align*}
    D_{\psi}(\pi^{\mu,\sigma^k}, \pi^{(k+1)T_{\sigma}}) \leq D_{\psi}(\pi^{\mu,\sigma^k}, \pi^{kT_{\sigma}})\left(1 - \frac{\eta \mu\gamma}{2}\right)^{T_{\sigma}}.
\end{align*}
Since $\sigma^k = \pi^{kT_{\sigma}}$ and $\sigma^{k+1} = \pi^{(k+1)T_{\sigma}}$, we conclude the statement.
\end{proof}

\subsection{Proof of Lemma~\ref{lem:formal_ftrl_sp_noisy}}
\label{sec:appx_frtl_sp_noisy}
\begin{proof}[Proof of Lemma~\ref{lem:formal_ftrl_sp_noisy}]
Writing  $y_i^t= \sum_{s=0}^t\eta_s(\nabla_{\pi_i}v_i(\pi_i^s, \pi_{-i}^s) + \xi^s_i - \mu \nabla_{\pi_i}G(\pi_i^s, r_i))$ and 
using Lemma \ref{lem:legendre} in Appendix~\ref{sec:appx_frtl_sp_full}, we have:
\begin{align*}
    &D_{\psi}(\pi_i^{\mu,\sigma^k}, \pi_i^{t+1}) - D_{\psi}(\pi_i^{\mu,\sigma^k}, \pi_i^t) + D_{\psi}(\pi_i^{t+1}, \pi_i^t) \\
    &= \langle y_i^t - y_i^{t-1}, \pi_i^{t+1} - \pi_i^{\mu,\sigma^k}\rangle \\
    &= \eta_t \langle \nabla_{\pi_i}v_i(\pi_i^t, \pi_{-i}^t) - \mu \nabla_{\pi_i}G(\pi_i^t, \sigma_i^k) + \xi_i^t, \pi_i^{t+1} - \pi_i^{\mu,\sigma^k}\rangle.
\end{align*}
Note that $\nabla_{\pi_i}G(\pi_i^t, \sigma_i^k)$ is well-defined because of Assumptions \ref{asm:rel_smooth} and \ref{asm:well_defined_ftrl}.
Thus, we can apply Lemma~\ref{lem:grad_breg_noise}, and we have for any $t\in \{kT_{\sigma}, kT_{\sigma}+1, \cdots, (k+1)T_{\sigma} - 1\}$:
\begin{align*}
& D_{\psi}(\pi^{\mu,r}, \pi^{t+1}) - D_{\psi}(\pi^{\mu,\sigma^k}, \pi^t) + D_{\psi}(\pi^{t+1}, \pi^t) \\
& \le \eta_t \left(\left( \frac{\mu^2 \gamma \rho^2 (\gamma + 2\beta) + 8 L^2}{2 \mu \gamma \rho^2} \right) D_{\psi}(\pi^{t+1}, \pi^t)  - \frac{\mu\gamma}{2} D_{\psi}(\pi^{\mu,\sigma^k}, \pi^t) \right)+ \eta_t \sum_{i=1}^N\langle \xi_i^t, \pi_i^{t+1} - \pi_i^{\mu,\sigma^k}\rangle
\\
& = \eta_t(\theta D_\psi(\pi^{t+1}, \pi^{t}) - \kappa D_\psi(\pi^{\mu,\sigma^k}, \pi^{t})) + \eta_t \sum_{i=1}^N\langle \xi_i^t,  \pi_i^{t+1} - \pi_i^{\mu,\sigma^k}\rangle.
\end{align*}

Therefore, the assumption in Lemma \ref{lem:noisy_tel_sum} is satisfied, and we get:
\begin{align*}
   &\mathbb{E}[D_{\psi}(\pi^{\mu,\sigma^k}, \pi^{t+1})] \\
    &\le \frac{2 \theta - \kappa}{\kappa (t- kT_{\sigma}) + 2 \theta} \mathbb{E}[D_\psi(\pi^{\mu,\sigma^k}, \pi^{kT_{\sigma}})] + \frac{N C^2 }{\rho (\kappa (t- kT_{\sigma}) + 2 \theta)} \left( \frac{1}{\kappa }\ln \left(\frac{\kappa}{2 \theta} (t- kT_{\sigma}) + 1\right) +  \frac{1}{2 \theta}\right).
\end{align*}
Taking $t=(k+1)T_{\sigma}-1$, we have:
\begin{align*}
&\mathbb{E}[D_{\psi}(\pi^{\mu,\sigma^k}, \pi^{(k+1)T_{\sigma}})] \\
&\le \frac{2 \theta - \kappa}{\kappa (T_{\sigma} - 1) + 2 \theta} \mathbb{E}[D_\psi(\pi^{\mu,\sigma^k}, \pi^{kT_{\sigma}})] + \frac{N C^2 }{\rho (\kappa (T_{\sigma} - 1) + 2 \theta)} \left( \frac{1}{\kappa }\ln \left(\frac{\kappa}{2 \theta} (T_{\sigma} - 1) + 1\right) +  \frac{1}{2 \theta}\right).
\end{align*}
Since $\sigma^k = \pi^{kT_{\sigma}}$ and $\sigma^{k+1} = \pi^{(k+1)T_{\sigma}}$, we conclude the statement.
\end{proof}

\section{Proofs for Section \ref{sec:approximate_nash_conv_full}}
\label{sec:appx_approximate_nash_conv_full}

\subsection{Proof of Theorem \ref{thm:exploitability_upper_bound}}
\begin{proof}[Proof of Theorem \ref{thm:exploitability_upper_bound}]
Since $v_i(\cdot, \pi_{-i}^t)$ is concave, we can upper bound the gap function for $\pi^t$ as:
\begin{align}
&\mathrm{GAP}(\pi^t) \nonumber\\
&= \sum_{i=1}^N\max_{\tilde{\pi}_i\in \mathcal{X}_i}\langle \nabla_{\pi_i}v_i(\pi^t), \tilde{\pi}_i - \pi_i^t\rangle \nonumber\\
&= \max_{\tilde{\pi}\in \mathcal{X}} \sum_{i=1}^N \left(\langle \nabla_{\pi_i}v_i(\pi^{\mu,\sigma}), \tilde{\pi}_i - \pi_i^{\mu,\sigma}\rangle - \langle \nabla_{\pi_i}v_i(\pi^{\mu,\sigma}), \pi_i^t - \pi_i^{\mu,\sigma}\rangle + \langle \nabla_{\pi_i}v_i(\pi^t) - \nabla_{\pi_i}v_i(\pi^{\mu,\sigma}), \tilde{\pi}_i - \pi_i^t\rangle\right).
\label{eq:decompose_exploitability_}
\end{align}

From Lemma \ref{lem:exploit_by_tangent}, the first term of \eqref{eq:decompose_exploitability_} can be upper bounded as:
\begin{align*}
    \max_{\tilde{\pi}\in \mathcal{X}}\sum_{i=1}^N \langle \nabla_{\pi_i}v_i(\pi^{\mu,\sigma}), \tilde{\pi}_i - \pi_i^{\mu,\sigma}\rangle \leq \mathrm{diam}(\mathcal{X}) \cdot \min_{(a_i)\in N_{\mathcal{X}}(\pi^{\mu,\sigma})}\sqrt{\sum_{i=1}^N\|-\nabla_{\pi_i}v_i(\pi^{\mu,\sigma}) + a_i\|^2}.
\end{align*}
From the first-order optimality condition for $\pi^{\mu,\sigma}$, we have for any $\pi\in \mathcal{X}$:
\begin{align*}
    \sum_{i=1}^N\langle \nabla_{\pi_i} v_i(\pi^{\mu,\sigma}) - \mu \nabla_{\pi_i}G(\pi_i^{\mu,\sigma}, \sigma_i), \pi_i - \pi_i^{\mu,\sigma}\rangle \leq 0,
\end{align*}
and then $(\nabla_{\pi_i} v_i(\pi^{\mu,\sigma}) - \mu \nabla_{\pi_i}G(\pi_i^{\mu,\sigma}, \sigma_i))_{i\in [N]} \in N_{\mathcal{X}}(\pi^{\mu,\sigma})$.
Thus,
\begin{align}
    &\max_{\tilde{\pi}\in \mathcal{X}}\sum_{i=1}^N \langle \nabla_{\pi_i}v_i(\pi^{\mu,\sigma}), \tilde{\pi}_i - \pi_i^{\mu,\sigma}\rangle \nonumber\\
    &\leq \mathrm{diam}(\mathcal{X})\sqrt{\sum_{i=1}^N\|-\nabla_{\pi_i}v_i(\pi^{\mu,\sigma}) +  \nabla_{\pi_i} v_i(\pi^{\mu,\sigma}) - \mu \nabla_{\pi_i}G(\pi_i^{\mu,\sigma}, \sigma_i)\|^2} \nonumber\\
    &= \mu \cdot \mathrm{diam}(\mathcal{X})\sqrt{\sum_{i=1}^N\| \nabla_{\pi_i}G(\pi_i^{\mu,\sigma}, \sigma_i)\|^2}.
    \label{eq:exploitability_first_term}
\end{align}

Next, from Cauchy–Schwarz inequality, the second term of \eqref{eq:decompose_exploitability_} can be bounded as:
\begin{align}
    - \sum_{i=1}^N \langle \nabla_{\pi_i}v_i(\pi^{\mu,\sigma}), \pi_i^t - \pi_i^{\mu,\sigma}\rangle \leq \|\pi^t - \pi^{\mu,\sigma}\|\sqrt{\sum_{i=1}^N\|\nabla_{\pi_i}v_i(\pi^{\mu,\sigma})\|^2}.
    \label{eq:exploitability_second_term}
\end{align}

Again from Cauchy–Schwarz inequality, the third term of \eqref{eq:decompose_exploitability_} is bounded by:
\begin{align}
    \sum_{i=1}^N\langle \nabla_{\pi_i}v_i(\pi^t) - \nabla_{\pi_i}v_i(\pi^{\mu,\sigma}), \tilde{\pi}_i - \pi_i^t\rangle &\leq \|\tilde{\pi} - \pi^t\|\sqrt{\sum_{i=1}^N\|\nabla_{\pi_i}v_i(\pi^t) - \nabla_{\pi_i}v_i(\pi^{\mu,\sigma})\|^2} \nonumber\\
    &\leq \mathrm{diam}(\mathcal{X})\sqrt{\sum_{i=1}^N\|\nabla_{\pi_i}v_i(\pi^t) - \nabla_{\pi_i}v_i(\pi^{\mu,\sigma})\|^2} \nonumber\\
    &\leq L\cdot \mathrm{diam}(\mathcal{X})\|\pi^t - \pi^{\mu,\sigma}\|,
    \label{eq:exploitability_third_term}
\end{align}
where the third inequality follows from \eqref{eq:smooth}.

By combining \eqref{eq:decompose_exploitability_}, \eqref{eq:exploitability_first_term}, \eqref{eq:exploitability_second_term}, and \eqref{eq:exploitability_third_term}, we get:
\begin{align*}
&\mathrm{GAP}(\pi^t) \leq \mu \cdot \mathrm{diam}(\mathcal{X})\sqrt{\sum_{i=1}^N\| \nabla_{\pi_i}G(\pi_i^{\mu,\sigma}, \sigma_i)\|^2} +  \left(L\cdot \mathrm{diam}(\mathcal{X}) + \sqrt{\sum_{i=1}^N\|\nabla_{\pi_i}v_i(\pi^{\mu,\sigma})\|^2}\right)\|\pi^t - \pi^{\mu,\sigma}\|.
\end{align*}
Thus, from Lemma \ref{lem:md_sp_full} and the strong convexity of $\psi$, we have:
\begin{align*}
\mathrm{GAP}(\pi^t) &\leq \mu \cdot \mathrm{diam}(\mathcal{X})\sqrt{\sum_{i=1}^N\| \nabla_{\pi_i}G(\pi_i^{\mu,\sigma}, \sigma_i)\|^2} \\
&+ \left(L\cdot \mathrm{diam}(\mathcal{X}) + \sqrt{\sum_{i=1}^N\|\nabla_{\pi_i}v_i(\pi^{\mu,\sigma})\|^2}\right)\sqrt{\frac{2D_{\psi}(\pi^{\mu,\sigma}, \pi^0)}{\rho}\left(1 - \frac{\eta \mu\gamma}{2}\right)^t}.
\end{align*}
\end{proof}

\section{Proofs for Additional Lemmas}
\subsection{Proof of Lemma \ref{lem:contraction_constrained}}
\begin{proof}[Proof of Lemma \ref{lem:contraction_constrained}]
From the first-order optimality condition for $\pi^{\mu, \sigma^k}$ and $\pi^{\mu, \sigma^{k-1}}$, we have for any $k\geq 1$:
\begin{align*}
    &\langle \nabla_{\pi_i}v_i(\pi_i^{\mu, \sigma^k}, \pi_{-i}^{\mu, \sigma^k}) - \mu \left(\pi_i^{\mu, \sigma^k} - \sigma_i^k\right), \sigma_i^k - \pi_i^{\mu, \sigma^k}\rangle \leq 0, \\
    &\langle \nabla_{\pi_i}v_i(\pi_i^{\mu, \sigma^{k-1}}, \pi_{-i}^{\mu, \sigma^{k-1}}) - \mu \left(\pi_i^{\mu, \sigma^{k-1}} - \sigma_i^{k-1}\right), \pi_i^{\mu, \sigma^k} - \pi_i^{\mu, \sigma^{k-1}}\rangle \leq 0.
\end{align*}
Summing up these inequalities yields:
\begin{align*}
    0 &\geq \sum_{i=1}^N \langle \nabla_{\pi_i} v_i(\pi_i^{\mu, \sigma^k}, \pi_{-i}^{\mu, \sigma^k}) - \mu \left(\pi_i^{\mu, \sigma^k} - \sigma_i^k\right), \sigma_i^k - \pi_i^{\mu, \sigma^k}\rangle \\
    &+ \sum_{i=1}^N \langle \nabla_{\pi_i} v_i(\pi_i^{\mu, \sigma^{k-1}}, \pi_{-i}^{\mu, \sigma^{k-1}}) - \mu \left(\pi_i^{\mu, \sigma^{k-1}} - \sigma_i^{k-1}\right), \pi_i^{\mu, \sigma^k} - \pi_i^{\mu, \sigma^{k-1}}\rangle \\
    &= \sum_{i=1}^N \langle \nabla_{\pi_i} v_i(\pi_i^{\mu, \sigma^k}, \pi_{-i}^{\mu, \sigma^k}), \pi_i^{\mu, \sigma^{k-1}} - \pi_i^{\mu, \sigma^k}\rangle + \sum_{i=1}^N \langle \nabla_{\pi_i} v_i(\pi_i^{\mu, \sigma^k}, \pi_{-i}^{\mu, \sigma^k}), \sigma_i^k - \pi_i^{\mu, \sigma^{k-1}}\rangle + \mu \|\pi^{\mu, \sigma^k} - \sigma^k\|^2 \\
    &+ \sum_{i=1}^N \langle \nabla_{\pi_i} v_i(\pi_i^{\mu, \sigma^{k-1}}, \pi_{-i}^{\mu, \sigma^{k-1}}), \pi_i^{\mu, \sigma^k} - \pi_i^{\mu, \sigma^{k-1}}\rangle + \mu \sum_{i=1}^N \langle \sigma_i^{k-1} - \pi_i^{\mu, \sigma^{k-1}}, \pi_i^{\mu, \sigma^k} - \pi_i^{\mu, \sigma^{k-1}}\rangle \\
    &= \sum_{i=1}^N \langle \nabla_{\pi_i} v_i(\pi_i^{\mu, \sigma^k}, \pi_{-i}^{\mu, \sigma^k}) - \nabla_{\pi_i} v_i(\pi_i^{\mu, \sigma^{k-1}}, \pi_{-i}^{\mu, \sigma^{k-1}}), \pi_i^{\mu, \sigma^{k-1}} - \pi_i^{\mu, \sigma^k}\rangle \\
    &+ \sum_{i=1}^N \langle \nabla_{\pi_i} v_i(\pi_i^{\mu, \sigma^k}, \pi_{-i}^{\mu, \sigma^k}), \sigma_i^k - \pi_i^{\mu, \sigma^{k-1}}\rangle + \mu \|\pi^{\mu, \sigma^k} - \sigma^k\|^2 + \mu \sum_{i=1}^N \langle \sigma_i^{k-1} - \pi_i^{\mu, \sigma^{k-1}}, \pi_i^{\mu, \sigma^k} - \pi_i^{\mu, \sigma^{k-1}}\rangle \\
    &\geq \sum_{i=1}^N \langle \nabla_{\pi_i} v_i(\pi_i^{\mu, \sigma^k}, \pi_{-i}^{\mu, \sigma^k}), \sigma_i^k - \pi_i^{\mu, \sigma^{k-1}}\rangle + \mu \|\pi^{\mu, \sigma^k} - \sigma^k\|^2 + \mu \sum_{i=1}^N \langle \sigma_i^{k-1} - \pi_i^{\mu, \sigma^{k-1}}, \pi_i^{\mu, \sigma^k} - \pi_i^{\mu, \sigma^{k-1}}\rangle,
\end{align*}
where the last inequality follows from \eqref{eq:monotone}.
Then, since
\begin{align*}
    \langle \sigma_i^{k-1} - \pi_i^{\mu, \sigma^{k-1}}, \pi_i^{\mu, \sigma^{k-1}} - \pi_i^{\mu, \sigma^k}\rangle = \frac{\|\pi_i^{\mu, \sigma^k} - \sigma_i^{k-1}\|^2}{2} - \frac{\|\pi_i^{\mu, \sigma^k} - \pi_i^{\mu, \sigma^{k-1}}\|^2}{2} - \frac{\|\pi_i^{\mu, \sigma^{k-1}} - \sigma_i^{k-1}\|^2}{2},
\end{align*}
we have:
\begin{align*}
    0 &\geq \sum_{i=1}^N \langle \nabla_{\pi_i} v_i(\pi_i^{\mu, \sigma^k}, \pi_{-i}^{\mu, \sigma^k}), \sigma_i^k - \pi_i^{\mu, \sigma^{k-1}}\rangle + \mu \|\pi^{\mu, \sigma^k} - \sigma^k\|^2 \\
    &+ \frac{\mu}{2} \sum_{i=1}^N \langle \sigma_i^{k-1} - \pi_i^{\mu, \sigma^{k-1}}, \pi_i^{\mu, \sigma^k} - \pi_i^{\mu, \sigma^{k-1}}\rangle + \frac{\mu}{2} \sum_{i=1}^N \langle \sigma_i^{k-1} - \pi_i^{\mu, \sigma^{k-1}}, \pi_i^{\mu, \sigma^k} - \pi_i^{\mu, \sigma^{k-1}}\rangle \\
    &\geq \sum_{i=1}^N \langle \nabla_{\pi_i} v_i(\pi_i^{\mu, \sigma^k}, \pi_{-i}^{\mu, \sigma^k}), \sigma_i^k - \pi_i^{\mu, \sigma^{k-1}}\rangle + \mu \|\pi^{\mu, \sigma^k} - \sigma^k\|^2 \\
    &- \frac{\mu}{4}\left(\|\sigma^{k-1} - \pi^{\mu, \sigma^{k-1}}\|^2 + \|\pi^{\mu, \sigma^k} - \pi^{\mu, \sigma^{k-1}}\|^2 - \|\pi^{\mu, \sigma^k} - \pi^{\mu, \sigma^{k-1}} + \sigma^{k-1} - \pi^{\mu, \sigma^{k-1}}\|^2\right) \\
    &- \frac{\mu}{4}\left(\|\pi^{\mu, \sigma^k} - \sigma^{k-1}\|^2 - \|\pi^{\mu, \sigma^k} - \pi^{\mu, \sigma^{k-1}}\|^2 - \|\pi^{\mu, \sigma^{k-1}} - \sigma^{k-1}\|^2\right) \\
    &\geq \sum_{i=1}^N \langle \nabla_{\pi_i} v_i(\pi_i^{\mu, \sigma^k}, \pi_{-i}^{\mu, \sigma^k}), \sigma_i^k - \pi_i^{\mu, \sigma^{k-1}}\rangle + \mu \|\pi^{\mu, \sigma^k} - \sigma^k\|^2 - \frac{\mu}{4}\|\pi^{\mu, \sigma^k} - \sigma^{k-1}\|^2 \\
    &\geq \sum_{i=1}^N \langle \nabla_{\pi_i} v_i(\pi_i^{\mu, \sigma^k}, \pi_{-i}^{\mu, \sigma^k}), \sigma_i^k - \pi_i^{\mu, \sigma^{k-1}}\rangle + \mu \|\pi^{\mu, \sigma^k} - \sigma^k\|^2 - \frac{\mu}{2}\|\pi^{\mu, \sigma^k} - \sigma^k\|^2 - \frac{\mu}{2}\|\sigma^k - \sigma^{k-1}\|^2 \\
    &= \sum_{i=1}^N \langle \nabla_{\pi_i} v_i(\pi_i^{\mu, \sigma^k}, \pi_{-i}^{\mu, \sigma^k}), \sigma_i^k - \pi_i^{\mu, \sigma^{k-1}}\rangle + \frac{\mu}{2}\|\pi^{\mu, \sigma^k} - \sigma^k\|^2 - \frac{\mu}{2}\|\sigma^k - \sigma^{k-1}\|^2,
\end{align*}
where the third inequality follows from $(a+b)^2\leq 2(a^2 + b^2)$ for $a, b\in \mathbb{R}$.
Thus,
\begin{align*}
    \|\pi^{\mu, \sigma^k} - \sigma^k\|^2 &\leq \|\sigma^k - \sigma^{k-1}\|^2 + \frac{2}{\mu} \sum_{i=1}^N \langle \nabla_{\pi_i} v_i(\pi_i^{\mu, \sigma^k}, \pi_{-i}^{\mu, \sigma^k}), \pi_i^{\mu, \sigma^{k-1}} - \sigma_i^k\rangle \\
    &\leq \|\sigma^k - \sigma^{k-1}\|^2 + \frac{2}{\mu} \|\pi^{\mu, \sigma^{k-1}} - \sigma^k\|\sqrt{\sum_{i=1}^N \|\nabla_{\pi_i} v_i(\pi_i^{\mu, \sigma^k}, \pi_{-i}^{\mu, \sigma^k})\|^2} \\
    &\leq \|\sigma^k - \sigma^{k-1}\|^2 + \frac{2\zeta}{\mu} \|\pi^{\mu, \sigma^{k-1}} - \sigma^k\|.
\end{align*}
\end{proof}

\subsection{Proof of Lemma \ref{lem:diff_p_sigma_sum_constrained}}
\begin{proof}[Proof of Lemma \ref{lem:diff_p_sigma_sum_constrained}]
From the first-order optimality condition for $p^{k+1}$, we have:
\begin{align*}
    &\langle \nabla_{\pi_i}v_i(\pi_i^{\mu, \sigma^k}, \pi_{-i}^{\mu, \sigma^k}) - \mu \left(\pi_i^{\mu, \sigma^k} - \sigma_i^k\right), \pi_i^{\ast} - \pi_i^{\mu, \sigma^k}\rangle \leq 0.
\end{align*}
Thus, from the three-point identity $2 \langle x - y, z - x \rangle = \|y - z\|^2 - \|x - y\|^2 - \|x - z\|^2$ and Young's inequality:
\begin{align*}
    &\sum_{i=1}^N \langle \nabla_{\pi_i}v_i(\pi_i^{\mu, \sigma^k}, \pi_{-i}^{\mu, \sigma^k}), \pi_i^{\ast} - \pi_i^{\mu, \sigma^k}\rangle \leq \mu \sum_{i=1}^N \langle \pi_i^{\mu, \sigma^k} - \sigma_i^k, \pi_i^{\ast} - \pi_i^{\mu, \sigma^k}\rangle \\
    &= \frac{\mu}{2}\|\pi^{\ast} - \sigma^k\|^2  - \frac{\mu}{2}\|\pi^{\mu, \sigma^k} - \sigma^k\|^2 - \frac{\mu}{2}\|\pi^{\ast} - \pi^{\mu, \sigma^k}\|^2 \\
    &= \frac{\mu}{2}\|\pi^{\ast} - \sigma^k\|^2  - \frac{\mu}{2}\|\pi^{\mu, \sigma^k} - \sigma^k\|^2 - \frac{\mu}{2}\|\pi^{\ast} - \sigma^{k+1} + \sigma^{k+1} - \pi^{\mu, \sigma^k}\|^2 \\
    &= \frac{\mu}{2}\|\pi^{\ast} - \sigma^k\|^2  - \frac{\mu}{2}\|\pi^{\mu, \sigma^k} - \sigma^k\|^2 - \frac{\mu}{2}\|\pi^{\ast} - \sigma^{k+1}\|^2 - \frac{\mu}{2}\|\sigma^{k+1} - \pi^{\mu, \sigma^k}\|^2 - \mu \langle \pi^{\ast} - \sigma^{k+1}, \sigma^{k+1} - \pi^{\mu, \sigma^k}\rangle \\
    &\leq \frac{\mu}{2}\|\pi^{\ast} - \sigma^k\|^2  - \frac{\mu}{2}\|\pi^{\mu, \sigma^k} - \sigma^k\|^2 - \frac{\mu}{2}\|\pi^{\ast} - \sigma^{k+1}\|^2 + \mu \|\pi^{\ast} - \sigma^{k+1}\|\| \sigma^{k+1} - \pi^{\mu, \sigma^k}\| \\
    &\leq \frac{\mu}{2}\|\pi^{\ast} - \sigma^k\|^2  - \frac{\mu}{2}\|\pi^{\mu, \sigma^k} - \sigma^k\|^2 - \frac{\mu}{2}\|\pi^{\ast} - \sigma^{k+1}\|^2 + \frac{\mu}{64T^3} \|\pi^{\ast} - \sigma^{k+1}\|^2 + \frac{32\mu T^3}{2} \| \sigma^{k+1} - \pi^{\mu, \sigma^k}\|^2.
\end{align*}
Here, since $T_{\sigma} \geq \max(\frac{6}{\ln 2 - \ln (2 - \eta \mu) }\ln T + \frac{2\ln 64}{\ln 2 - \ln (2 - \eta \mu)}, 1)$, we have:
\begin{align*}
    \left(1-\frac{\eta\mu}{2}\right)^{-T_{\sigma}} \geq \left(1-\frac{\eta\mu}{2}\right)^{-\frac{3}{\ln 2 - \ln (2 - \eta \mu) }\ln T} \left(1-\frac{\eta\mu}{2}\right)^{\frac{\ln 64}{\ln \left(1 - \frac{\eta\mu}{2}\right)}} = 64\left(1-\frac{\eta\mu}{2}\right)^{\frac{\ln T^3}{\ln (2 - \eta \mu)  - \ln 2}}=64T^3.
\end{align*}
Therefore, we get from Lemma \ref{lem:formal_gd_sp_full}:
\begin{align*}
    &\sum_{i=1}^N \langle \nabla_{\pi_i}v_i(\pi_i^{\mu, \sigma^k}, \pi_{-i}^{\mu, \sigma^k}), \pi_i^{\ast} - \pi_i^{\mu, \sigma^k}\rangle \\
    &\leq \frac{\mu}{2}\|\pi^{\ast} - \sigma^k\|^2  - \frac{\mu}{2}\|\pi^{\mu, \sigma^k} - \sigma^k\|^2 - \frac{\mu}{2}\|\pi^{\ast} - \sigma^{k+1}\|^2 + \frac{\mu}{64T^3} \|\pi^{\ast} - \sigma^{k+1}\|^2 + \frac{32\mu T^3}{2} \| \sigma^k - \pi^{\mu, \sigma^k}\|^2\left(1-\frac{\eta\mu}{2}\right)^{T_{\sigma}} \\
    &\leq \frac{\mu}{2}\|\pi^{\ast} - \sigma^k\|^2  - \frac{\mu}{2}\|\pi^{\mu, \sigma^k} - \sigma^k\|^2 - \frac{\mu}{2}\|\pi^{\ast} - \sigma^{k+1}\|^2 + \frac{\mu}{64T^3} \|\pi^{\ast} - \sigma^{k+1}\|^2 + \frac{\mu}{4} \| \sigma^k - \pi^{\mu, \sigma^k}\|^2 \\
    &= \frac{\mu}{2}\|\pi^{\ast} - \sigma^k\|^2  - \frac{\mu}{2}\|\pi^{\ast} - \sigma^{k+1}\|^2 - \frac{\mu}{4}\|\pi^{\mu, \sigma^k} - \sigma^k\|^2 + \frac{\mu}{64T^3} \|\pi^{\ast} - \sigma^{k+1}\|^2.
\end{align*}
Summing up this inequality from $k=0$ to $K-1$ yields:
\begin{align*}
    &\frac{\mu}{2}\|\pi^{\ast} - \sigma^0\|^2 + \frac{\mu}{64T^3} \sum_{k=0}^{K-1}\|\pi^{\ast} - \sigma^{k+1}\|^2 - \frac{\mu}{4}\sum_{k=0}^{K-1}\|\pi^{\mu, \sigma^k} - \sigma^k\|^2 \\
    &\geq \sum_{k=0}^{K-1}\sum_{i=1}^N \langle \nabla_{\pi_i}v_i(\pi_i^{\mu, \sigma^k}, \pi_{-i}^{\mu, \sigma^k}), \pi_i^{\ast} - \pi_i^{\mu, \sigma^k}\rangle \\
    &\geq \sum_{k=0}^{K-1}\sum_{i=1}^N \langle \nabla_{\pi_i}v_i(\pi_i^{\ast}, \pi_{-i}^{\ast}), \pi_i^{\ast} - \pi_i^{\mu, \sigma^k}\rangle \\
    &\geq 0.
\end{align*}
Then, from Cauchy–Schwarz inequality, we have:
\begin{align*}
    \frac{\mu}{4}\sum_{k=0}^{K-1}\|\pi^{\mu, \sigma^k} - \sigma^k\|^2 &\leq \frac{\mu}{2}\|\pi^{\ast} - \sigma^0\|^2 + \frac{\mu}{64T^3} \sum_{k=0}^{K-1}\|\pi^{\ast} - \sigma^{k+1}\|^2 \\
    &\leq \frac{\mu}{2}\|\pi^{\ast} - \sigma^0\|^2 + \frac{\mu}{64T^3} \sum_{k=0}^{K-1}\left(\sum_{\tau=0}^k\|\sigma^{\tau+1} - \sigma^{\tau}\| + \|\pi^{\ast} - \sigma^0\|\right)^2 \\
    &\leq \frac{\mu}{2}\|\pi^{\ast} - \sigma^0\|^2 + \frac{\mu}{32T^3} \sum_{k=0}^{K-1}\left(\left(\sum_{\tau=0}^k\|\sigma^{\tau+1} - \sigma^{\tau}\|\right)^2 + \|\pi^{\ast} - \sigma^0\|^2\right) \\
    &\leq \frac{\mu}{2}\|\pi^{\ast} - \sigma^0\|^2 + \frac{\mu}{32T^3} \sum_{k=0}^{K-1}\left(K\sum_{\tau=0}^k\|\sigma^{\tau+1} - \sigma^{\tau}\|^2 + \|\pi^{\ast} - \sigma^0\|^2\right) \\
    &\leq \frac{\mu}{2}\|\pi^{\ast} - \sigma^0\|^2 + \frac{\mu}{32T^3} \left(K^2\sum_{k=0}^{K-1}\|\sigma^{k+1} - \sigma^k\|^2 + K\|\pi^{\ast} - \sigma^0\|^2\right) \\
    &\leq \frac{\mu}{2}\|\pi^{\ast} - \sigma^0\|^2 + \frac{\mu}{32T^3} \left(K^2\sum_{k=0}^{K-1}\left(\|\sigma^{k+1} - \pi^{\mu, \sigma^k}\| + \|\pi^{\mu, \sigma^k} - \sigma^k\|\right)^2 + K\|\pi^{\ast} - \sigma^0\|^2\right).
\end{align*}
By applying Lemma \ref{lem:formal_gd_sp_full} to the above inequality, we get:
\begin{align*}
    &\frac{\mu}{4}\sum_{k=0}^{K-1}\|\pi^{\mu, \sigma^k} - \sigma^k\|^2 \\
    &\leq \frac{\mu}{2}\|\pi^{\ast} - \sigma^0\|^2 + \frac{\mu}{32T^3} \left(K^2\sum_{k=0}^{K-1}\left(\| \sigma^k - \pi^{\mu, \sigma^k}\|^2\left(1-\frac{\eta\mu}{2}\right)^{\frac{T_{\sigma}}{2}} + \|\pi^{\mu, \sigma^k} - \sigma^k\|\right)^2 + K\|\pi^{\ast} - \sigma^0\|^2\right) \\
    &\leq \frac{\mu}{2}\|\pi^{\ast} - \sigma^0\|^2 + \frac{\mu}{32T^3} \left(K^2\sum_{k=0}^{K-1}\left(2\|\pi^{\mu, \sigma^k} - \sigma^k\|\right)^2 + K\|\pi^{\ast} - \sigma^0\|^2\right) \\
    &= \frac{\mu}{2}\|\pi^{\ast} - \sigma^0\|^2 + \frac{\mu}{32T^3} \left(4K^2\sum_{k=0}^{K-1}\|\pi^{\mu, \sigma^k} - \sigma^k\|^2 + K\|\pi^{\ast} - \sigma^0\|^2\right) \\
    &\leq \frac{\mu}{2}\|\pi^{\ast} - \sigma^0\|^2 + \frac{\mu}{32T^3} \left(4T^2\sum_{k=0}^{K-1}\|\pi^{\mu, \sigma^k} - \sigma^k\|^2 + T\|\pi^{\ast} - \sigma^0\|^2\right) \\
    &= \mu\left(\frac{1}{2} + \frac{1}{T^2}\right)\|\pi^{\ast} - \sigma^0\|^2 + \frac{\mu}{8T} \sum_{k=0}^{K-1}\|\pi^{\mu, \sigma^k} - \sigma^k\|^2 \\
    &\leq 2\mu\|\pi^{\ast} - \sigma^0\|^2 + \frac{\mu}{8} \sum_{k=0}^{K-1}\|\pi^{\mu, \sigma^k} - \sigma^k\|^2.
\end{align*}
Therefore, for $K\geq 1$, we get:
\begin{align*}
    \sum_{k=0}^{K-1}\|\pi^{\mu, \sigma^k} - \sigma^k\|^2 &\leq 16\|\pi^{\ast} - \sigma^0\|^2.
\end{align*}
\end{proof}

\subsection{Proof of Lemma \ref{lem:diff_p_sigma_sum_noisy}}
\begin{proof}[Proof of Lemma \ref{lem:diff_p_sigma_sum_noisy}]
From the first-order optimality condition for $p^{k+1}$, we have:
\begin{align*}
    &\langle \nabla_{\pi_i}v_i(\pi_i^{\mu, \sigma^k}, \pi_{-i}^{\mu, \sigma^k}) - \mu \left(\pi_i^{\mu, \sigma^k} - \sigma_i^k\right), \pi_i^{\ast} - \pi_i^{\mu, \sigma^k}\rangle \leq 0.
\end{align*}
Thus, from the three-point identity $2 \langle a - b, c - a \rangle = \|b - c\|^2 - \|a - b\|^2 - \|a - c\|^2$ and Young's inequality:
\begin{align*}
    &\sum_{i=1}^N \langle \nabla_{\pi_i}v_i(\pi_i^{\mu, \sigma^k}, \pi_{-i}^{\mu, \sigma^k}), \pi_i^{\ast} - \pi_i^{\mu, \sigma^k}\rangle \leq \mu \sum_{i=1}^N \langle \pi_i^{\mu, \sigma^k} - \sigma_i^k, \pi_i^{\ast} - \pi_i^{\mu, \sigma^k}\rangle \\
    &= \frac{\mu}{2}\|\pi^{\ast} - \sigma^k\|^2  - \frac{\mu}{2}\|\pi^{\mu, \sigma^k} - \sigma^k\|^2 - \frac{\mu}{2}\|\pi^{\ast} - \pi^{\mu, \sigma^k}\|^2 \\
    &= \frac{\mu}{2}\|\pi^{\ast} - \sigma^k\|^2  - \frac{\mu}{2}\|\pi^{\mu, \sigma^k} - \sigma^k\|^2 - \frac{\mu}{2}\|\pi^{\ast} - \sigma^{k+1} + \sigma^{k+1} - \pi^{\mu, \sigma^k}\|^2 \\
    &= \frac{\mu}{2}\|\pi^{\ast} - \sigma^k\|^2  - \frac{\mu}{2}\|\pi^{\mu, \sigma^k} - \sigma^k\|^2 - \frac{\mu}{2}\|\pi^{\ast} - \sigma^{k+1}\|^2 - \frac{\mu}{2}\|\sigma^{k+1} - \pi^{\mu, \sigma^k}\|^2 - \mu \langle \pi^{\ast} - \sigma^{k+1}, \sigma^{k+1} - \pi^{\mu, \sigma^k}\rangle \\
    &\leq \frac{\mu}{2}\|\pi^{\ast} - \sigma^k\|^2  - \frac{\mu}{2}\|\pi^{\mu, \sigma^k} - \sigma^k\|^2 - \frac{\mu}{2}\|\pi^{\ast} - \sigma^{k+1}\|^2 + \mu \|\pi^{\ast} - \sigma^{k+1}\|\| \sigma^{k+1} - \pi^{\mu, \sigma^k}\| \\
    &\leq \frac{\mu}{2}\|\pi^{\ast} - \sigma^k\|^2  - \frac{\mu}{2}\|\pi^{\mu, \sigma^k} - \sigma^k\|^2 - \frac{\mu}{2}\|\pi^{\ast} - \sigma^{k+1}\|^2 + \mu \cdot \mathrm{diam}(\mathcal{X})\| \sigma^{k+1} - \pi^{\mu, \sigma^k}\| \\
\end{align*}
Since $T_{\sigma} \geq \max(T^{4/5} + 2, 3)$, we get from Lemma \ref{lem:formal_gd_sp_noisy}:
\begin{align*}
    &\mathbb{E}\left[\sum_{i=1}^N \langle \nabla_{\pi_i}v_i(\pi_i^{\mu, \sigma^k}, \pi_{-i}^{\mu, \sigma^k}), \pi_i^{\ast} - \pi_i^{\mu, \sigma^k}\rangle \right] \\
    &\leq \mathbb{E}\left[\frac{\mu}{2}\|\pi^{\ast} - \sigma^k\|^2  - \frac{\mu}{2}\|\pi^{\mu, \sigma^k} - \sigma^k\|^2 - \frac{\mu}{2}\|\pi^{\ast} - \sigma^{k+1}\|^2 \right] \\
    &+ \mu \cdot \mathrm{diam}(\mathcal{X}) \cdot \sqrt{\frac{\rho(2 \theta - \kappa)\mathrm{diam}(\mathcal{X})^2 + N C^2\left( \frac{1}{\kappa }\ln \left(\frac{\kappa}{2 \theta} (T_{\sigma} - 1) + 1\right) +  \frac{1}{2 \theta}\right)}{\rho(\kappa (T_{\sigma} - 1) + 2 \theta)}} \\
    &\leq \mathbb{E}\left[\frac{\mu}{2}\|\pi^{\ast} - \sigma^k\|^2  - \frac{\mu}{2}\|\pi^{\mu, \sigma^k} - \sigma^k\|^2 - \frac{\mu}{2}\|\pi^{\ast} - \sigma^{k+1}\|^2 \right] \\
    &+ \mu \cdot \mathrm{diam}(\mathcal{X}) \cdot \sqrt{\frac{\rho(2 \theta - \kappa)\mathrm{diam}(\mathcal{X})^2 + N C^2\left( \frac{1}{\kappa }\ln \left(\frac{\kappa}{2 \theta} T + 1\right) +  \frac{1}{2 \theta}\right)}{\rho(\kappa T^{4/5} + 2 \theta)}}.
\end{align*}
Summing up this inequality from $k=0$ to $K-1$ and taking its expectation yields:
\begin{align*}
    &\frac{\mu}{2}\|\pi^{\ast} - \sigma^0\|^2 - \mathbb{E}\left[\frac{\mu}{2}\sum_{k=0}^{K-1}\|\pi^{\mu, \sigma^k} - \sigma^k\|^2\right] \\
    &+ K\mu \cdot \mathrm{diam}(\mathcal{X}) \cdot \sqrt{\frac{\rho(2 \theta - \kappa)\mathrm{diam}(\mathcal{X})^2 + N C^2\left( \frac{1}{\kappa }\ln \left(\frac{\kappa}{2 \theta} T + 1\right) +  \frac{1}{2 \theta}\right)}{\rho(\kappa T^{4/5} + 2 \theta)}} \\
    &\geq \mathbb{E}\left[\sum_{k=0}^{K-1}\sum_{i=1}^N \langle \nabla_{\pi_i}v_i(\pi_i^{\mu, \sigma^k}, \pi_{-i}^{\mu, \sigma^k}), \pi_i^{\ast} - \pi_i^{\mu, \sigma^k}\rangle \right] \\
    &\geq \mathbb{E}\left[\sum_{k=0}^{K-1}\sum_{i=1}^N \langle \nabla_{\pi_i}v_i(\pi_i^{\ast}, \pi_{-i}^{\ast}), \pi_i^{\ast} - \pi_i^{\mu, \sigma^k}\rangle\right] \\
    &\geq 0.
\end{align*}
Therefore, for $K\geq 1$, we get:
\begin{align*}
    &\mathbb{E}\left[\sum_{k=0}^{K-1}\|\pi^{\mu, \sigma^k} - \sigma^k\|^2\right] \leq \|\pi^{\ast} - \sigma^0\|^2 + K\cdot \mathrm{diam}(\mathcal{X}) \cdot \sqrt{\frac{\rho(2 \theta - \kappa)\mathrm{diam}(\mathcal{X})^2 + N C^2\left( \frac{1}{\kappa }\ln \left(\frac{\kappa}{2 \theta} T + 1\right) +  \frac{1}{2 \theta}\right)}{\rho(\kappa T^{4/5} + 2 \theta)}}.
\end{align*}
\end{proof}

\subsection{Proof of Lemma \ref{lem:grad_breg_full}}
\label{sec:appx_convergence_rate_full}
\begin{proof}[Proof of Lemma \ref{lem:grad_breg_full}]
We first decompose the inequality in the assumption as follows:
\begin{align}
    &D_{\psi}(\pi_i^{\mu,\sigma^k}, \pi_i^{t+1}) - D_{\psi}(\pi_i^{\mu,\sigma^k}, \pi_i^t) + D_{\psi}(\pi_i^{t+1}, \pi_i^t) \nonumber\\
    &\leq \eta \langle \nabla_{\pi_i}v_i(\pi_i^t, \pi_{-i}^t) - \mu \nabla_{\pi_i}G(\pi_i^t, \sigma_i^k), \pi_i^{t+1} - \pi_i^{\mu,\sigma^k}\rangle \nonumber
    \\
    &= \eta \langle \nabla_{\pi_i}v_i(\pi_i^t, \pi_{-i}^t), \pi_i^{t+1} - \pi_i^{\mu,\sigma^k}\rangle +  \eta \mu \langle \nabla_{\pi_i}G(\pi_i^t, \sigma_i^k), \pi_i^t - \pi_i^{t+1}\rangle + \eta \mu \langle \nabla_{\pi_i}G(\pi_i^t, \sigma_i^k), \pi_i^{\mu,\sigma^k} - \pi_i^t\rangle.
    \label{eq:three_point_eq}
\end{align}
From the relative smoothness in Assumption \ref{asm:rel_smooth} and the convexity of $G(\cdot, \sigma_i^k)$:
\begin{align}
    &\langle \nabla_{\pi_i}G(\pi_i^t, \sigma_i^k), \pi_i^t - \pi_i^{t+1}\rangle \nonumber\\
    &\leq G(\pi_i^t, \sigma_i^k) - G(\pi_i^{t+1}, \sigma_i^k) + \beta D_{\psi}(\pi_i^{t+1}, \pi_i^t) \nonumber\\
    &\leq G(\pi_i^t, \sigma_i^k) - G(\pi_i^{\mu,\sigma^k}, \sigma_i^k) + \langle \nabla_{\pi_i}G(\pi_i^{\mu,\sigma^k}, \sigma_i^k), \pi_i^{\mu,\sigma^k} - \pi_i^{t+1}\rangle + \beta D_{\psi}(\pi_i^{t+1}, \pi_i^t).
    \label{eq:rel_smooth}
\end{align}
Also, from the relative strong convexity in Assumption \ref{asm:rel_smooth}:
\begin{align}
    G(\pi_i^t, \sigma_i^k) - G(\pi_i^{\mu,\sigma^k}, \sigma_i^k) \leq \langle \nabla_{\pi_i}G(\pi_i^t, \sigma_i^k), \pi_i^t - \pi_i^{\mu,\sigma^k}\rangle - \gamma D_{\psi}(\pi_i^{\mu,\sigma^k}, \pi_i^t).
    \label{eq:rel_strongly_convex}
\end{align}
By combining \eqref{eq:three_point_eq}, \eqref{eq:rel_smooth}, and \eqref{eq:rel_strongly_convex}, we have:
\begin{align*}
    &D_{\psi}(\pi_i^{\mu,\sigma^k}, \pi_i^{t+1}) - D_{\psi}(\pi_i^{\mu,\sigma^k}, \pi_i^t) + D_{\psi}(\pi_i^{t+1}, \pi_i^t) \\
    &\leq \eta \langle \nabla_{\pi_i}v_i(\pi_i^t, \pi_{-i}^t), \pi_i^{t+1} - \pi_i^{\mu,\sigma^k}\rangle + \eta\mu\langle \nabla_{\pi_i}G(\pi_i^{\mu,\sigma^k}, \sigma_i^k), \pi_i^{\mu,\sigma^k} - \pi_i^{t+1}\rangle \\
    &- \eta\mu\gamma D_{\psi}(\pi_i^{\mu,\sigma^k}, \pi_i^t) + \eta\mu\beta D_{\psi}(\pi_i^{t+1}, \pi_i^t),
\end{align*}
and then:
\begin{align*}
    &D_{\psi}(\pi_i^{\mu,\sigma^k}, \pi_i^{t+1}) - (1 - \eta\mu \gamma) D_{\psi}(\pi_i^{\mu,\sigma^k}, \pi_i^t) + (1 - \eta\mu \beta) D_{\psi}(\pi_i^{t+1}, \pi_i^t) \\
    &\leq \eta \langle \nabla_{\pi_i}v_i(\pi_i^t, \pi_{-i}^t), \pi_i^{t+1} - \pi_i^{\mu,\sigma^k}\rangle + \eta\mu\langle \nabla_{\pi_i}G(\pi_i^{\mu,\sigma^k}, \sigma_i^k), \pi_i^{\mu,\sigma^k} - \pi_i^{t+1}\rangle \\
    &= \eta \langle \nabla_{\pi_i}v_i(\pi_i^{t+1}, \pi_{-i}^{t+1}), \pi_i^{t+1} - \pi_i^{\mu,\sigma^k}\rangle + \eta\mu\langle \nabla_{\pi_i}G(\pi_i^{\mu,\sigma^k}, \sigma_i^k), \pi_i^{\mu,\sigma^k} - \pi_i^{t+1}\rangle \\
    &+ \eta \langle \nabla_{\pi_i}v_i(\pi_i^t, \pi_{-i}^t) - \nabla_{\pi_i}v_i(\pi_i^{t+1}, \pi_{-i}^{t+1}), \pi_i^{t+1} - \pi_i^{\mu,\sigma^k}\rangle.
\end{align*}
Summing this inequality from $i=1$ to $N$ implies that:
\begin{align}
    &D_{\psi}(\pi^{\mu,\sigma^k}, \pi^{t+1}) - (1 - \eta\mu \gamma) D_{\psi}(\pi^{\mu,\sigma^k}, \pi^t) + (1 - \eta\mu \beta) D_{\psi}(\pi^{t+1}, \pi^t) \nonumber\\
    &\leq \eta \sum_{i=1}^N\langle \nabla_{\pi_i}v_i(\pi_i^{t+1}, \pi_{-i}^{t+1}), \pi_i^{t+1} - \pi_i^{\mu,\sigma^k}\rangle + \eta\mu \sum_{i=1}^N\langle \nabla_{\pi_i}G(\pi_i^{\mu,\sigma^k}, \sigma_i^k), \pi_i^{\mu,\sigma^k} - \pi_i^{t+1}\rangle \nonumber\\
    &+ \eta \sum_{i=1}^N\langle \nabla_{\pi_i}v_i(\pi_i^t, \pi_{-i}^t) - \nabla_{\pi_i}v_i(\pi_i^{t+1}, \pi_{-i}^{t+1}), \pi_i^{t+1} - \pi_i^{\mu,\sigma^k}\rangle \nonumber\\
    &\leq \eta \sum_{i=1}^N\langle \nabla_{\pi_i}v_i(\pi_i^{\mu,\sigma^k}, \pi_{-i}^{\mu,\sigma^k}) - \mu \nabla_{\pi_i}G(\pi_i^{\mu,\sigma^k}, \sigma_i^k), \pi_i^{t+1} - \pi_i^{\mu,\sigma^k}\rangle \nonumber\\
    &+ \eta \sum_{i=1}^N\langle \nabla_{\pi_i}v_i(\pi_i^t, \pi_{-i}^t) - \nabla_{\pi_i}v_i(\pi_i^{t+1}, \pi_{-i}^{t+1}), \pi_i^{t+1} - \pi_i^{\mu,\sigma^k}\rangle \nonumber\\
    &\leq \eta \sum_{i=1}^N\langle \nabla_{\pi_i}v_i(\pi_i^t, \pi_{-i}^t) - \nabla_{\pi_i}v_i(\pi_i^{t+1}, \pi_{-i}^{t+1}), \pi_i^{t+1} - \pi_i^{\mu,\sigma^k}\rangle,
    \label{eq:three_point_ineq}
\end{align}
where the second inequality follows from \eqref{eq:monotone}, and the third inequality follows from the first-order optimality condition for $\pi^{\mu,\sigma^k}$.

Here, from Young's inequality, we have for any $\lambda>0$:
\begin{align}
    &\sum_{i=1}^N\langle \nabla_{\pi_i}v_i(\pi_i^t, \pi_{-i}^t) - \nabla_{\pi_i}v_i(\pi_i^{t+1}, \pi_{-i}^{t+1}), \pi_i^{t+1} - \pi_i^{\mu,\sigma^k}\rangle \nonumber\\
    &= \sum_{i=1}^N\langle \nabla_{\pi_i}v_i(\pi_i^t, \pi_{-i}^t) - \nabla_{\pi_i}v_i(\pi_i^{t+1}, \pi_{-i}^{t+1}), \pi_i^{t+1} - \pi_i^t\rangle + \sum_{i=1}^N\langle \nabla_{\pi_i}v_i(\pi_i^t, \pi_{-i}^t) - \nabla_{\pi_i}v_i(\pi_i^{t+1}, \pi_{-i}^{t+1}), \pi_i^t - \pi_i^{\mu,\sigma^k}\rangle \nonumber\\
    &\leq \lambda \sum_{i=1}^N\|\nabla_{\pi_i}v_i(\pi_i^{t+1}, \pi_{-i}^{t+1}) - \nabla_{\pi_i}v_i(\pi_i^t, \pi_{-i}^t)\|^2 + \frac{1}{2\lambda}\|\pi^{t+1} - \pi^t\|^2 + \frac{1}{2\lambda}\|\pi^t - \pi^{\mu,\sigma^k}\|^2 \nonumber\\
    &\leq \left(L^2\lambda  + \frac{1}{2\lambda}\right)\|\pi^{t+1} - \pi^t\|^2 + \frac{1}{2\lambda}\|\pi^t - \pi^{\mu,\sigma^k}\|^2 \nonumber\\
    &\leq \frac{1}{\rho}\left(2L^2\lambda  + \frac{1}{\lambda}\right)D_{\psi}(\pi^{t+1}, \pi^t) + \frac{1}{\rho\lambda}D_{\psi}(\pi^{\mu,\sigma^k}, \pi^t).
    \label{eq:young_ineq}
\end{align}
where the second inequality follows from \eqref{eq:smooth}, and the fourth inequality follows from the strong convexity of $\psi$.

By combining \eqref{eq:three_point_ineq} and \eqref{eq:young_ineq}, we get:
\begin{align*}
    &D_{\psi}(\pi^{\mu,\sigma^k}, \pi^{t+1}) \leq \left(1 - \eta \left(\mu \gamma - \frac{1}{\rho\lambda}\right)\right) D_{\psi}(\pi^{\mu,\sigma^k}, \pi^t) - \left(1 - \eta\left(\mu \beta + \frac{2L^2\lambda}{\rho} + \frac{1}{\rho\lambda}\right)\right)D_{\psi}(\pi^{t+1}, \pi^t).
\end{align*}
By setting $\lambda=\frac{2}{\mu\gamma\rho}$,
\begin{align*}
    D_{\psi}(\pi^{\mu,\sigma^k}, \pi^{t+1}) &\leq \left(1 - \frac{\eta\mu\gamma}{2}\right) D_{\psi}(\pi^{\mu,\sigma^k}, \pi^t) - \left(1 - \eta\left(\frac{\mu(\gamma + 2\beta)}{2} + \frac{4L^2}{\mu\gamma\rho^2}\right)\right)D_{\psi}(\pi^{t+1}, \pi^t) \\
    &= \left(1 - \frac{\eta\mu\gamma}{2}\right) D_{\psi}(\pi^{\mu,\sigma^k}, \pi^t) - \left(1 - \eta\left(\frac{\mu^2\gamma\rho^2(\gamma + 2\beta) + 8L^2}{2\mu\gamma\rho^2}\right)\right)D_{\psi}(\pi^{t+1}, \pi^t).
\end{align*}
Thus, when $\eta \leq \frac{2\mu\gamma\rho^2}{\mu^2\gamma\rho^2(\gamma + 2\beta) + 8L^2} < \frac{2}{\mu\gamma}$, we have for any $t\in \{kT_{\sigma}, kT_{\sigma}+1, \cdots, (k+1)T_{\sigma} - 1\}$:
\begin{align*}
    D_{\psi}(\pi^{\mu,\sigma^k}, \pi^{t+1}) \leq D_{\psi}(\pi^{\mu,\sigma^k}, \pi^t)\left(1 - \frac{\eta\mu\gamma}{2}\right) \leq D_{\psi}(\pi^{\mu,\sigma^k}, \pi^{kT_{\sigma}})\left(1 - \frac{\eta \mu\gamma}{2}\right)^{t - kT_{\sigma}+1}.
\end{align*}
\end{proof}

\subsection{Proof of Lemma~\ref{lem:grad_breg_noise}}\label{app:prf_grad_breg_noise}
\begin{proof}[Proof of Lemma~\ref{lem:grad_breg_noise}]
We first decompose the inequality in the assumption as follows:
\begin{align}
    & D_{\psi}(\pi_i^{\mu,\sigma^k}, \pi_i^{t+1}) - D_{\psi}(\pi_i^{\mu,\sigma^k}, \pi_i^t) + D_{\psi}(\pi_i^{t+1}, \pi_i^t)  \nonumber \\
    & \le \eta_t\langle \nabla_{\pi_i}v_i(\pi_i^t, \pi_{-i}^t) - \mu\nabla_{\pi_i}G(\pi_i^t, \sigma_i^k) + \xi_i^t,  \pi_i^{t+1} - \pi_i^{\mu,\sigma^k}\rangle \nonumber \\
    & = \eta_t \langle \nabla_{\pi_i}v_i(\pi_i^t, \pi_{-i}^t), \pi_i^{t+1} - \pi_i^{\mu,\sigma^k}\rangle +  \eta_t \mu \langle \nabla_{\pi_i}G(\pi_i^t, \sigma_i^k), \pi_i^t - \pi_i^{t+1}\rangle \nonumber \\
    & \quad + \eta_t \mu \langle \nabla_{\pi_i}G(\pi_i^t, \sigma_i^k), \pi_i^{\mu,\sigma^k} - \pi_i^t\rangle +\langle \xi_i^t, \pi_i^{t+1} - \pi_i^{\mu,\sigma^k}\rangle. \label{eq:decom_noise}
\end{align}
By combining \eqref{eq:rel_smooth}, \eqref{eq:rel_strongly_convex} in Appendix \ref{sec:appx_convergence_rate_full}, and \eqref{eq:decom_noise},
\begin{align*}
    & D_{\psi}(\pi_i^{\mu,\sigma^k}, \pi_i^{t+1}) - D_{\psi}(\pi_i^{\mu,\sigma^k}, \pi_i^t) + D_{\psi}(\pi_i^{t+1}, \pi_i^t) \\
    & \le \eta_t \langle \nabla_{\pi_i}v_i(\pi_i^t, \pi_{-i}^t), \pi_i^{t+1} - \pi_i^{\mu,\sigma^k}\rangle + \eta_t\mu \langle \nabla_{\pi_i}G(\pi_i^{\mu,\sigma^k}, \sigma_i^k), \pi_i^{\mu,\sigma^k} - \pi_i^{t+1}\rangle + \eta_t\mu\beta D_{\psi}(\pi_i^{t+1}, \pi_i^t)\\
    & \quad  - \eta_t\mu\gamma D_{\psi}(\pi_i^{\mu,\sigma^k}, \pi_i^t) +\langle \xi_i^t, \pi_i^{t+1} - \pi_i^{\mu,\sigma^k}\rangle.
\end{align*}
Summing up these inequalities with respect to the player index,
\begin{align}
    & D_{\psi}(\pi^{\mu,\sigma^k}, \pi^{t+1}) - D_{\psi}(\pi^{\mu,\sigma^k}, \pi^t) + D_{\psi}(\pi^{t+1}, \pi^t) \nonumber
    \\
    & \le \eta_t \sum_{i=1}^N \langle \nabla_{\pi_i}v_i(\pi_i^t, \pi_{-i}^t) - \mu G(\pi_i^{\mu,\sigma^k}, \sigma_i^k), \pi_i^{t+1} - \pi_i^{\mu,\sigma^k}\rangle + \eta_t\mu\beta D_{\psi}(\pi^{t+1}, \pi^t)  \nonumber
    \\
    & \quad  - \eta_t\mu\gamma D_{\psi}(\pi^{\mu,\sigma^k}, \pi^t) + \sum_{i=1}^N \langle \xi_i^t, \pi_i^{t+1} - \pi_i^{\mu,\sigma^k}\rangle \nonumber
    \\
    & = \sum_{i=1}^N\eta_t \langle \nabla_{\pi_i}v_i(\pi_i^{t+1}, \pi_{-i}^{t+1}) - \mu\nabla_{\pi_i}G(\pi_i^{\mu,\sigma^k}, \sigma_i^k), \pi_i^{t+1} - \pi_i^{\mu,\sigma^k}\rangle - \eta_t\mu\gamma D_{\psi}(\pi^{\mu,\sigma^k}, \pi^t) + \eta_t\mu\beta D_{\psi}(\pi^{t+1}, \pi^t) \nonumber
    \\
    &+ \eta_t \sum_{i=1}^N\langle \nabla_{\pi_i}v_i(\pi_i^t, \pi_{-i}^t) - \nabla_{\pi_i}v_i(\pi_i^{t+1}, \pi_{-i}^{t+1}), \pi_i^{t+1} - \pi_i^{\mu,\sigma^k}\rangle + \eta_t \sum_{i=1}^N\langle \xi_i^t, \pi_i^{t+1} - \pi_i^{\mu,\sigma^k}\rangle \nonumber
    \\
    &\leq \sum_{i=1}^N\eta_t \langle \nabla_{\pi_i}v_i(\pi_i^{\mu,\sigma^k}, \pi_{-i}^{\mu,\sigma^k}) - \mu\nabla_{\pi_i}G(\pi_i^{\mu,\sigma^k}, \sigma_i^k), \pi_i^{t+1} - \pi_i^{\mu,\sigma^k}\rangle - \eta_t\mu\gamma D_{\psi}(\pi^{\mu,\sigma^k}, \pi^t) + \eta_t\mu\beta D_{\psi}(\pi^{t+1}, \pi^t) \nonumber\\
    &+ \eta_t \sum_{i=1}^N\langle \nabla_{\pi_i}v_i(\pi_i^t, \pi_{-i}^t) - \nabla_{\pi_i}v_i(\pi_i^{t+1}, \pi_{-i}^{t+1}), \pi_i^{t+1} - \pi_i^{\mu,\sigma^k}\rangle + \eta_t \sum_{i=1}^N\langle \xi_i^t, \pi_i^{t+1} - \pi_i^{\mu,\sigma^k}\rangle \nonumber\\
    &\leq - \eta_t\mu\gamma D_{\psi}(\pi^{\mu,\sigma^k}, \pi^t) + \eta_t\mu\beta D_{\psi}(\pi^{t+1}, \pi^t) \nonumber
    \\
    &+ \eta_t \sum_{i=1}^N\langle \nabla_{\pi_i}v_i(\pi_i^t, \pi_{-i}^t) - \nabla_{\pi_i}v_i(\pi_i^{t+1}, \pi_{-i}^{t+1}), \pi_i^{t+1} - \pi_i^{\mu,\sigma^k}\rangle + \eta_t \sum_{i=1}^N\langle \xi_i^t, \pi_i^{t+1} - \pi_i^{\mu,\sigma^k}\rangle,    \label{eq:three_point_ineq_noise}
\end{align}
where the second inequality follows \eqref{eq:monotone}, and the third inequality follows from the first-order optimality condition for $\pi^{\mu,\sigma^k}$.

By combining \eqref{eq:young_ineq} in Appendix \ref{sec:appx_convergence_rate_full} and \eqref{eq:three_point_ineq_noise}, we have for any $\lambda>0$:
\begin{align*}
    &D_{\psi}(\pi^{\mu,\sigma^k}, \pi^{t+1}) - D_{\psi}(\pi^{\mu,\sigma^k}, \pi^t) + D_{\psi}(\pi^{t+1}, \pi^t) \nonumber
    \\
    &\leq - \eta_t\mu\gamma D_{\psi}(\pi^{\mu,\sigma^k}, \pi^t) + \eta_t\mu\beta D_{\psi}(\pi^{t+1}, \pi^t) \nonumber
    \\
    &+ \frac{\eta_t }{\rho}\left(2L^2\lambda  + \frac{1}{\lambda}\right)D_{\psi}(\pi^{t+1}, \pi^t) + \frac{\eta_t}{\rho\lambda}D_{\psi}(\pi^{\mu,\sigma^k}, \pi^t) + \eta_t \sum_{i=1}^N\langle \xi_i^t, \pi_i^{t+1} - \pi_i^{\mu,\sigma^k}\rangle.
\end{align*}
By setting $\lambda = \frac{2}{\mu \gamma \rho}$,
\begin{align*}
    & D_{\psi}(\pi^{\mu,\sigma^k}, \pi^{t+1}) - D_{\psi}(\pi^{\mu,\sigma^k}, \pi^t) + D_{\psi}(\pi^{t+1}, \pi^t) 
    \\
    & \le - \frac{\eta_t\mu\gamma}{2} D_{\psi}(\pi^{\mu,\sigma^k}, \pi^t) + \eta_t\left( \frac{\mu^2 \gamma \rho^2 (\gamma + 2\beta) + 8 L^2}{2 \mu \gamma \rho^2} \right) D_{\psi}(\pi^{t+1}, \pi^t) + \eta_t \sum_{i=1}^N\langle \xi_i^t, \pi_i^{t+1} - \pi_i^{\mu,\sigma^k}\rangle.
\end{align*}
This concludes the proof.
\end{proof}

\subsection{Proof of Lemma~\ref{lem:noisy_tel_sum}} \label{app:prf_noisy_tel_sum}
\begin{proof}[Proof of Lemma~\ref{lem:noisy_tel_sum}]
Reforming the inequality in the assumption,
\begin{align*}
    D_{\psi}(\pi^{\mu,\sigma^k}, \pi^{t+1}) &\le (1 - \kappa \eta_t) D_\psi(\pi^{\mu,\sigma^k}, \pi^{t}) - (1- \eta_t \theta) D_\psi(\pi^{t+1}, \pi^{t}) + \eta_t \sum_{i=1}^N\langle \xi_i^t,  \pi_i^{t+1} -  \pi_i^{\mu,\sigma^k}\rangle
    \\
    & =  (1 - \kappa \eta_t) D_\psi(\pi^{\mu,\sigma^k}, \pi^{t}) - (1- \eta_t \theta) D_\psi(\pi^{t+1}, \pi^{t}) \\
    &\phantom{=} + \eta_t \sum_{i=1}^N\langle \xi_i^t,   \pi_i^{t} - \pi_i^{\mu,\sigma^k}\rangle + \eta_t \sum_{i=1}^N\langle \xi_i^t,  \pi_i^{t+1} -  \pi_i^{t}\rangle.
\end{align*}
By taking the expectation conditioned on $\mathcal{F}_t$ for both sides and using Assumption \ref{asm:noise} (a),
\begin{align*}
     \mathbb{E}[D_{\psi}(\pi^{\mu,\sigma^k}, \pi^{t+1})| \mathcal{F}_t] & \le  (1 - \kappa \eta_t) D_\psi(\pi^{\mu,\sigma^k}, \pi^{t})  - (1- \eta_t \theta) \mathbb{E}[ D_\psi(\pi^{t+1}, \pi^{t}) | \mathcal{F}_t] + \sum_{i=1}^N \mathbb{E}[  \langle \eta_t \xi_i^t,   \pi_i^{t+1} - \pi_i^{t}\rangle| \mathcal{F}_t] \\
     & =  (1 - \kappa \eta_t) D_\psi(\pi^{\mu,\sigma^k}, \pi^{t})  - (1- \eta_t \theta) \mathbb{E}[ D_\psi(\pi^{t+1}, \pi^{t}) | \mathcal{F}_t] \\
     &\phantom{=} + \sum_{i=1}^N \mathbb{E}\left[  \left\langle \frac{\eta_t \xi_i^t}{\sqrt{\rho(1 - \eta_t \theta)}},  \sqrt{\rho(1 - \eta_t \theta)}( \pi_i^{t+1} -\pi_i^{t})\right\rangle| \mathcal{F}_t\right] \\
     & \le (1 - \kappa \eta_t) D_\psi(\pi^{\mu,\sigma^k}, \pi^{t})  - (1- \eta_t \theta) \mathbb{E}[ D_\psi(\pi^{t+1}, \pi^{t}) | \mathcal{F}_t] \\
     &\phantom{=} + \frac{\eta_t^2}{2\rho(1 -\eta_t\theta) } \sum_{i=1}^N\mathbb{E}[\|\xi_i^t\|^2|\mathcal{F}_t] + \frac{\rho(1 -\eta_t\theta)}{2} \mathbb{E}[\|\pi^{t+1} -\pi^{t}\|^2| \mathcal{F}_t] \\
     & \leq  (1 - \kappa \eta_t) D_\psi(\pi^{\mu,\sigma^k}, \pi^{t}) + \frac{\eta_t^2}{2\rho(1 -\eta_t\theta) } \sum_{i=1}^N \mathbb{E}[\|\xi_i^t\|^2|\mathcal{F}_t] \\
     & \le \left(1 - \frac{1}{t-kT_{\sigma}+2\theta/\kappa} \right)D_\psi(\pi^{\mu,\sigma^k}, \pi^{t}) + \frac{\eta_t^2}{\rho} \sum_{i=1}^N \mathbb{E}[\|\xi_i^t\|^2|\mathcal{F}_t].
\end{align*}
Therefore, rearranging and taking the expectations,
\begin{align*}
    (t - kT_{\sigma} + 2\theta/\kappa)\mathbb{E}[D_{\psi}(\pi^{\mu,\sigma^k}, \pi^{t+1})] \le (t - kT_{\sigma} -1 + 2\theta/\kappa) \mathbb{E}[D_\psi(\pi^{\mu,\sigma^k}, \pi^{t}) ] + \frac{N C^2}{\rho \kappa (\kappa (t - kT_{\sigma}) + 2 \theta)}.
\end{align*}
Telescoping the sum, 
\begin{align*}
    & (t- kT_{\sigma} + 2\theta/\kappa)\mathbb{E}[D_{\psi}(\pi^{\mu,\sigma^k}, \pi^{t+1})] \le (2\theta/\kappa -1) \mathbb{E}[D_\psi(\pi^{\mu,\sigma^k}, \pi^{kT_{\sigma}})] + \frac{N C^2}{\rho \kappa }\sum_{s=kT_{\sigma}}^t \frac{1}{\kappa (s - kT_{\sigma}) + 2 \theta} \\
    \Longleftrightarrow \quad & \mathbb{E}[D_{\psi}(\pi^{\mu,\sigma^k}, \pi^{t+1})] \le \frac{2 \theta - \kappa}{\kappa (t - kT_{\sigma}) + 2 \theta} \mathbb{E}[D_\psi(\pi^{\mu,\sigma^k}, \pi^{kT_{\sigma}})] + \frac{N C^2 }{\rho(\kappa (t - kT_{\sigma}) + 2 \theta)}\sum_{s=0}^{t- kT_{\sigma}} \frac{1}{\kappa s + 2 \theta}.
\end{align*}

Here, we introduce the following lemma, whose proof is given in Appendix~\ref{app:prf_sum_eval}, for the evaluation of the sum.
\begin{lemma}\label{lem:sum_eval}
    For any $\kappa,\theta\geq 0$ and $t\ge0$, 
    \begin{align*}
        \sum_{s=0}^t \frac{1}{\kappa s + 2 \theta} \le \frac{1}{2 \theta} + \frac{1}{\kappa} \ln \left(\frac{\kappa}{2 \theta} t + 1\right).
    \end{align*}
\end{lemma}

In summary, we obtain the following inequality:
\begin{align*}
    &\mathbb{E}[D_{\psi}(\pi^{\mu,\sigma^k}, \pi^{t+1})] \\
    &\le \frac{2 \theta - \kappa}{\kappa (t- kT_{\sigma}) + 2 \theta} \mathbb{E}[D_\psi(\pi^{\mu,\sigma^k}, \pi^{kT_{\sigma}})] + \frac{N C^2 }{\rho (\kappa (t- kT_{\sigma}) + 2 \theta)} \left( \frac{1}{\kappa }\ln \left(\frac{\kappa}{2 \theta} (t- kT_{\sigma}) + 1\right) +  \frac{1}{2 \theta}\right).
\end{align*}
This concludes the proof.
\end{proof}

\subsection{Proof of Lemma \ref{lem:exact_conv_bregman}}
\begin{proof}[Proof of Lemma \ref{lem:exact_conv_bregman}]
Recall that $G(\pi_i, \pi_i')=D_{\psi'}(\pi_i, \pi_i')$ for any $i\in [N]$ and $\pi_i, \pi_i'\in \mathcal{X}_i$.
By the first-order optimality condition for $\sigma^{k+1}$, we have for all $\pi^{\ast}\in \Pi^{\ast}$:
\begin{align*}
    \sum_{i=1}^N\langle \nabla_{\pi_i}v_i(\sigma_i^{k+1}, \sigma_{-i}^{k+1}) - \mu(\nabla \psi'(\sigma_i^{k+1}) - \nabla \psi'(\sigma_i^k)), \pi_i^{\ast} - \sigma_i^{k+1}\rangle \leq 0.
\end{align*}
Then,
\begin{align*}
    \sum_{i=1}^N \langle \nabla \psi'(\sigma_i^{k+1}) - \nabla\psi'(\sigma_i^k), \sigma_i^{k+1} - \pi_i^{\ast}\rangle \leq \frac{1}{\mu} \sum_{i=1}^N \langle \sigma_i^{k+1} - \pi_i^{\ast}, \nabla_{\pi_i}v_i(\sigma_i^{k+1}, \sigma_{-i}^{k+1})\rangle.
\end{align*}
Moreover, we have for any $\pi^{\ast}\in \Pi^{\ast}$:
\begin{align*}
    &D_{\psi'}(\pi^{\ast}, \sigma^{k+1}) - D_{\psi'}(\pi^{\ast}, \sigma^k) \\
    &= \sum_{i=1}^N \left(\psi'(\pi_i^{\ast}) - \psi'(\sigma_i^{k+1}) - \langle \nabla \psi'(\sigma_i^{k+1}), \pi_i^{\ast} - \sigma_i^{k+1}\rangle - \psi'(\pi_i^{\ast}) + \psi'(\sigma_i^k) + \langle \nabla \psi'(\sigma_i^k), \pi_i^{\ast} - \sigma_i^k\rangle\right) \\
    &= \sum_{i=1}^N \left(-\psi'(\sigma_i^{k+1}) + \psi'(\sigma_i^k) + \langle \nabla \psi'(\sigma_i^k), \sigma_i^{k+1} - \sigma_i^k\rangle - \langle \nabla\psi'(\sigma_i^{k+1}) - \nabla \psi'(\sigma_i^k), \pi_i^{\ast} - \sigma_i^{k+1}\rangle\right) \\
    &= -D_{\psi'}(\sigma^{k+1}, \sigma^k) + \sum_{i=1}^N \langle \nabla \psi'(\sigma_i^{k+1}) - \nabla \psi'(\sigma_i^k), \sigma_i^{k+1} - \pi_i^{\ast}\rangle.
\end{align*}
By combining these inequalities, we get for any $\pi^{\ast}\in \Pi^{\ast}$:
\begin{align*}
    D_{\psi'}(\pi^{\ast}, \sigma^{k+1}) - D_{\psi'}(\pi^{\ast}, \sigma^k) &\leq -D_{\psi'}(\sigma^{k+1}, \sigma^k) + \frac{1}{\mu}\sum_{i=1}^N \langle \sigma_i^{k+1} - \pi_i^{\ast}, \nabla_{\pi_i}v_i(\sigma_i^{k+1}, \sigma_{-i}^{k+1})\rangle \\
    &\leq -D_{\psi'}(\sigma^{k+1}, \sigma^k) + \frac{1}{\mu}\sum_{i=1}^N \langle \sigma_i^{k+1} - \pi_i^{\ast}, \nabla_{\pi_i}v_i(\pi_i^{\ast}, \pi_{-i}^{\ast})\rangle,
\end{align*}
where the second inequality follows from \eqref{eq:monotone}.
Since $\pi^{\ast}$ is the Nash equilibrium, from the first-order optimality condition, we get:
\begin{align*}
    \sum_{i=1}^N \langle \sigma_i^{k+1} - \pi_i^{\ast}, \nabla_{\pi_i}v_i(\pi_i^{\ast}, \pi_{-i}^{\ast})\rangle \leq 0.
\end{align*}
Thus, we have for $\pi^{\ast}\in \Pi^{\ast}$:
\begin{align*}
    D_{\psi'}(\pi^{\ast}, \sigma^{k+1}) - D_{\psi'}(\pi^{\ast}, \sigma^k) &\leq -D_{\psi'}(\sigma^{k+1}, \sigma^k).
\end{align*}
\end{proof}

\subsection{Proof of Lemma \ref{lem:stop_condition_bregman}}
\begin{proof}[Proof of Lemma \ref{lem:stop_condition_bregman}]
By using the first-order optimality condition for $\sigma_i^{k+1}$, we have for all $\pi \in \mathcal{X}$:
\begin{align*}
    \sum_{i=1}^N \langle \nabla_{\pi_i}v_i(\sigma_i^{k+1}, \sigma_{-i}^{k+1}) - \mu (\nabla_{\pi_i}\psi'(\sigma_i^{k+1}) - \nabla_{\pi_i}\psi'(\sigma_i^k)), \pi_i - \sigma_i^{k+1}\rangle \leq 0,
\end{align*}
and then
\begin{align*}
    \sum_{i=1}^N \langle \nabla_{\pi_i}v_i(\sigma_i^{k+1}, \sigma_{-i}^{k+1}), \pi_i - \sigma_i^{k+1}\rangle \leq \mu \sum_{i=1}^N \langle \nabla_{\pi_i}\psi'(\sigma_i^{k+1}) - \nabla_{\pi_i}\psi'(\sigma_i^k), \pi_i - \sigma_i^{k+1}\rangle.
\end{align*}
Under the assumption that $\sigma^{k+1}=\sigma^k$, we have for all $\pi \in \mathcal{X}$:
\begin{align*}
    \sum_{i=1}^N \langle \nabla_{\pi_i}v_i(\sigma_i^{k+1}, \sigma_{-i}^{k+1}), \pi_i - \sigma_i^{k+1}\rangle \leq 0.
\end{align*}
This is equivalent to the first-order optimality condition for $\pi^{\ast} \in \Pi^{\ast}$.
Therefore, $\sigma^{k+1}=\sigma^k$ is a Nash equilibrium of the underlying game.
\end{proof}

\subsection{Proof of Lemma \ref{lem:exact_conv_min_alpha}}
\begin{proof}[Proof of Lemma \ref{lem:exact_conv_min_alpha}]
First, we prove the first statement of the lemma by using the following lemmas:
\begin{lemma}
\label{lem:stop_condition_alpha}
Assume that $\sigma^{k+1}=\pi^{\mu, \sigma^k}$ for $k\geq 0$, and $G$ is one of the following divergence: 1) $\alpha$-divergence with $\alpha\in (0, 1)$; 2) R\'{e}nyi-divergence with $\alpha\in (0, 1)$; 3) reverse KL divergence.
If $\sigma^{k+1} = \sigma^k$, then $\sigma^k$ is a Nash equilibrium of the underlying game.
\end{lemma}
\begin{lemma}
\label{lem:exact_conv_alpha}
Assume that $\sigma^{k+1}=\pi^{\mu, \sigma^k}$ for $k\geq 0$, and $G$ is one of the following divergence: 1) $\alpha$-divergence with $\alpha\in (0, 1)$; 2) R\'{e}nyi-divergence with $\alpha\in (0, 1)$; 3) reverse KL divergence.
Then, if $\sigma^{k+1}\neq \sigma^k$, we have for any $\pi^{\ast} \in \Pi^{\ast}$ and $k\geq 0$:
\begin{align*}
    \mathrm{KL}(\pi^{\ast}, \sigma^{k+1}) - \mathrm{KL}(\pi^{\ast}, \sigma^k) < 0.
\end{align*}
\end{lemma}

From Lemma \ref{lem:stop_condition_alpha}, when $\sigma^k\in \mathcal{X} \setminus \Pi^{\ast}$, $\sigma^{k+1}\neq \sigma^k$ always holds.
Let us define $\pi^{\star} = \argmin_{\pi^{\ast}\in \Pi^{\ast}} \mathrm{KL}(\pi^{\ast}, \sigma^k)$.
Since $\sigma^{k+1}\neq \sigma^k$, from Lemma \ref{lem:exact_conv_alpha}, we have:
\begin{align*}
    \min_{\pi^{\ast}\in \Pi^{\ast}} \mathrm{KL}(\pi^{\ast}, \sigma^k) = \mathrm{KL}(\pi^{\star}, \sigma^k) > \mathrm{KL}(\pi^{\star}, \sigma^{k+1}) \geq \min_{\pi^{\ast}\in \Pi^{\ast}} \mathrm{KL}(\pi^{\ast}, \sigma^{k+1}).
\end{align*}
Therefore, if $\sigma^k\in \mathcal{X} \setminus \Pi^{\ast}$ then $\min_{\pi^{\ast}\in \Pi^{\ast}} \mathrm{KL}(\pi^{\ast}, \sigma^{k+1}) < \min_{\pi^{\ast}\in \Pi^{\ast}} \mathrm{KL}(\pi^{\ast}, \sigma^k)$.

Next, we prove the second statement of the lemma.
Assume that there exists $\sigma^k\in \Pi^{\ast}$ such that $\sigma^{k+1}\neq \sigma^k$.
In this case, we can apply Lemma \ref{lem:exact_conv_alpha}, hence we have $\mathrm{KL}(\pi^{\ast}, \sigma^{k+1}) < \mathrm{KL}(\pi^{\ast}, \sigma^k)$ for all $\pi^{\ast}\in \Pi^{\ast}$.
On the other hand, since $\sigma^k\in \Pi^{\ast}$, there exists a Nash equilibrium $\pi^{\star}\in \Pi^{\ast}$ such that $\mathrm{KL}(\pi^{\star}, \sigma^k)=0$.
Therefore, we have $\mathrm{KL}(\pi^{\star}, \sigma^{k+1}) < \mathrm{KL}(\pi^{\star}, \sigma^k) = 0$, which contradicts $\mathrm{KL}(\pi^{\star}, \sigma^{k+1}) \geq 0$.
Thus, if $\sigma^k\in \Pi^{\ast}$ then $\sigma^{k+1} = \sigma^k$.
\end{proof}

\subsection{Proof of Lemma \ref{lem:map_continuity}}
\begin{proof}[Proof of Lemma \ref{lem:map_continuity}]
For a given $\sigma\in \mathcal{X}$, let us consider that $\pi^t$ follows the following continuous-time dynamics:
\begin{align}
\label{eq:continuous_time_ftrl}
    \pi_i^t &= \argmax_{\pi_i\in \mathcal{X}_i}\left\{\langle y_i^t, \pi_i\rangle - \psi(\pi_i)\right\}, \\
    y_{ij}^t &= \int_0^t \left(\frac{\partial}{\pi_{ij}}v_i(\pi_i^s, \pi_{-i}^s) - \mu \frac{\partial}{\pi_{ij}}G(\pi_i^s, \sigma_i)\right) \nonumber.
\end{align}
We assume that $\psi(\pi_i)=\sum_{j=1}^{d_i} \pi_{ij}\ln \pi_{ij}$.
Note that this dynamics is the continuous-time version of APFTRL, so clearly $\pi^{\mu,\sigma}$ defined by \eqref{eq:perturbed_nash} is the stationary point of \eqref{eq:continuous_time_ftrl}.
We have for a given $\sigma'\in \mathcal{X}$ and the associated stationary point $\pi^{\mu,\sigma'}=F(\sigma')$:
\begin{align*}
    \frac{d}{dt}\mathrm{KL}(\pi^{\mu,\sigma'}, \pi^t) &= \frac{d}{dt}D_{\psi}(\pi^{\mu,\sigma'}, \pi^t) \\
    &= \sum_{i=1}^N \frac{d}{dt}\left(\psi(\pi_i^{\mu,\sigma'}) - \psi(\pi_i^t) - \langle y_i^t, \pi_i^{\mu,\sigma'} - \pi_i^t\rangle\right) \\
    &= \sum_{i=1}^N \frac{d}{dt}\left(\langle y_i^t, \pi_i^t\rangle - \psi(\pi_i^t) - \langle y_i^t, \pi_i^{\mu,\sigma'}\rangle + \psi(\pi_i^{\mu,\sigma'}) \right) \\
    &= \sum_{i=1}^N \frac{d}{dt}\left(\psi^{\ast}(y_i^t) - \langle y_i^t, \pi_i^{\mu,\sigma'}\rangle \right) \\
    &= \sum_{i=1}^N \left(\left\langle \frac{d}{dt}y_i^t, \nabla \psi^{\ast}(y_i^t)\right\rangle - \left\langle \frac{d}{dt}y_i^t, \pi_i^{\mu,\sigma'}\right\rangle \right) \\
    &= \sum_{i=1}^N \left\langle \frac{d}{dt}y_i^t, \nabla \psi^{\ast}(y_i^t) - \pi_i^{\mu,\sigma'}\right\rangle,
\end{align*}
where $\psi^{\ast}(y_i^t) = \max_{\pi_i\in \mathcal{X}_i}\left\{\langle y_i^t, \pi_i\rangle - \psi(\pi_i)\right\}$.
When $\psi(\pi_i)=\sum_{j=1}^{d_i} \pi_{ij}\ln \pi_{ij}$, we have
\begin{align*}
    \psi^{\ast}(y_i^t) &= \sum_{j=1}^{d_i}y_{ij}^t \frac{\exp(y_{ij}^t)}{\sum_{j'=1}^{d_i}\exp(y_{ij'}^t)} - \sum_{j=1}^{d_i}\frac{\exp(y_{ij}^t)}{\sum_{j'=1}^{d_i}\exp(y_{ij'}^t)}\ln \frac{\exp(y_{ij}^t)}{\sum_{j'=1}^{d_i}\exp(y_{ij'}^t)} \\
    &= \frac{\ln \sum_{j'=1}^{d_i}\exp(y_{ij'}^t)}{\sum_{j'=1}^{d_i}\exp(y_{ij'}^t)}\sum_{j=1}^{d_i}\exp(y_{ij}^t),
\end{align*}
and then,
\begin{align*}
    \frac{\partial}{\partial y_{ij}}\psi^{\ast}(y_i^t) &= \frac{\exp(y_{ij}^t)}{(\sum_{j'=1}^{d_i}\exp(y_{ij'}^t))^2}\sum_{j'=1}^{d_i}\exp(y_{ij'}^t) - \frac{\ln \sum_{j'=1}^{d_i}\exp(y_{ij'}^t)}{(\sum_{j'=1}^{d_i}\exp(y_{ij'}^t))^2}\exp(y_{ij}^t)\sum_{j'=1}^{d_i}\exp(y_{ij'}^t) \\
    &+ \frac{\ln \sum_{j'=1}^{d_i}\exp(y_{ij'}^t)}{\sum_{j'=1}^{d_i}\exp(y_{ij'}^t)}\exp(y_{ij}^t) \\
    &= \frac{\exp(y_{ij}^t)}{\sum_{j'=1}^{d_i}\exp(y_{ij'}^t)} = \pi_{ij}^t.
\end{align*}
Therefore, we get $\nabla \psi^{\ast}(y_i^t) = \pi_i^t$.
Hence,
\begin{align}
\label{eq:time_derivative_divergence}
    \frac{d}{dt}\mathrm{KL}(\pi^{\mu,\sigma'}, \pi^t) &= \sum_{i=1}^N \left\langle \frac{d}{dt}y_i^t, \pi_i^t - \pi_i^{\mu,\sigma'}\right\rangle \nonumber\\
    &= \sum_{i=1}^N \langle \nabla_{\pi_i}v_i(\pi_i^t, \pi_{-i}^t) - \mu \nabla_{\pi_i}G(\pi_i^t, \sigma_i), \pi_i^t - \pi_i^{\mu, \sigma'}\rangle \nonumber\\
    &= \sum_{i=1}^N \langle \nabla_{\pi_i}v_i(\pi_i^t, \pi_{-i}^t) - \mu \nabla_{\pi_i}G(\pi_i^t, \sigma_i'), \pi_i^t - \pi_i^{\mu, \sigma'}\rangle \nonumber\\
    &+  \mu\sum_{i=1}^N \langle \nabla_{\pi_i}G(\pi_i^t, \sigma_i') - \nabla_{\pi_i}G(\pi_i^t, \sigma_i), \pi_i^t - \pi_i^{\mu, \sigma'}\rangle.
\end{align}

The first term of \eqref{eq:time_derivative_divergence} can be written as:
\begin{align*}
    &\sum_{i=1}^N \langle \nabla_{\pi_i}v_i(\pi_i^t, \pi_{-i}^t) - \mu \nabla_{\pi_i}G(\pi_i^t, \sigma_i'), \pi_i^t - \pi_i^{\mu, \sigma'}\rangle \nonumber\\
    &\leq \sum_{i=1}^N \langle \nabla_{\pi_i}v_i(\pi_i^{\mu,\sigma'}, \pi_{-i}^{\mu,\sigma'}) - \mu \nabla_{\pi_i}G(\pi_i^t, \sigma_i'), \pi_i^t - \pi_i^{\mu, \sigma'}\rangle \nonumber\\
    &= \sum_{i=1}^N \langle \nabla_{\pi_i}v_i(\pi_i^{\mu,\sigma'}, \pi_{-i}^{\mu,\sigma'}), \pi_i^t - \pi_i^{\mu, \sigma'}\rangle - \mu \sum_{i=1}^N \langle \nabla_{\pi_i}G(\pi_i^t, \sigma_i'), \pi_i^t - \pi_i^{\mu, \sigma'}\rangle \nonumber\\
    &= \sum_{i=1}^N \langle \nabla_{\pi_i}v_i(\pi_i^{\mu,\sigma'}, \pi_{-i}^{\mu,\sigma'}) - \mu \nabla_{\pi_i}G(\pi_i^{\mu,\sigma'}, \sigma_i'), \pi_i^t - \pi_i^{\mu, \sigma'}\rangle \nonumber\\
    &- \mu \sum_{i=1}^N \langle \nabla_{\pi_i}G(\pi_i^t, \sigma_i') - \nabla_{\pi_i}G(\pi_i^{\mu,\sigma'}, \sigma_i'), \pi_i^t - \pi_i^{\mu, \sigma'}\rangle \nonumber\\
    &\leq -\mu \sum_{i=1}^N \langle \nabla_{\pi_i}G(\pi_i^t, \sigma_i') - \nabla_{\pi_i}G(\pi_i^{\mu,\sigma'}, \sigma_i'), \pi_i^t - \pi_i^{\mu, \sigma'}\rangle.
\end{align*}
where the first inequality follows from \eqref{eq:monotone}, and the second inequality follows from the first-order optimality condition for $\pi^{\mu,\sigma}$.
When $G$ is $\alpha$-divergence, $G$ has a diagonal Hessian is given as:
\begin{align*}
\nabla^2G(\pi_i, \sigma_i') = 
\begin{bmatrix}
\frac{\sigma_{i1}'}{(\pi_{i1})^{\alpha-2}} & & \\
& \ddots & \\
& & \frac{\sigma_{id_i}'}{(\pi_{id_i})^{\alpha-2}},
\end{bmatrix}
\end{align*}
and thus, its smallest eigenvalue is lower bounded by $\min_{j\in [d_i]}\sigma_{ij}'$.
Therefore,
\begin{align}
\label{eq:time_derivative_divergence_first_term}
    &\sum_{i=1}^N \langle \nabla_{\pi_i}v_i(\pi_i^t, \pi_{-i}^t) - \mu \nabla_{\pi_i}G(\pi_i^t, \sigma_i'), \pi_i^t - \pi_i^{\mu, \sigma'}\rangle \nonumber\\
    &\leq -\mu\sum_{i=1}^N\langle \nabla_{\pi_i}G(\pi_i^t, \sigma_i') - \nabla_{\pi_i}G(\pi_i^{\mu,\sigma'}, \sigma_i'), \pi_i^t - \pi_i^{\mu, \sigma'}\rangle \nonumber\\
    &\leq -\mu \left(\min_{i\in [N], ~j\in [d_i]}\sigma_{ij}'\right)\|\pi^t - \pi^{\mu,\sigma'}\|^2.
\end{align}

On the other hand, by compactness of $\mathcal{X}_i$, the second term of \eqref{eq:time_derivative_divergence} is written as:
\begin{align}
    &\mu\sum_{i=1}^N \langle \nabla_{\pi_i}G(\pi_i^t, \sigma_i') - \nabla_{\pi_i}G(\pi_i^t, \sigma_i), \pi_i^t - \pi_i^{\mu, \sigma'}\rangle \nonumber\\
    &\leq \mu \cdot \mathrm{diam}(\mathcal{X}) \sqrt{\sum_{i=1}^N \|\nabla_{\pi_i}G(\pi_i^t, \sigma_i') - \nabla_{\pi_i}G(\pi_i^t, \sigma_i)\|^2}
\label{eq:time_derivative_divergence_second_term}
\end{align}

By combining \eqref{eq:time_derivative_divergence}, \eqref{eq:time_derivative_divergence_first_term}, and \eqref{eq:time_derivative_divergence_second_term}, we get:
\begin{align*}
    &\frac{d}{dt}\mathrm{KL}(\pi^{\mu,\sigma'}, \pi^t) \\
    &\leq -\mu \left(\min_{i\in [N], ~j\in [d_i]}\sigma_{ij}'\right)\|\pi^t - \pi^{\mu,\sigma'}\|^2 +  \mu \cdot \mathrm{diam}(\mathcal{X}) \sqrt{\sum_{i=1}^N \|\nabla_{\pi_i}G(\pi_i^t, \sigma_i') - \nabla_{\pi_i}G(\pi_i^t, \sigma_i)\|^2}.
\end{align*}
Recall that $\pi^{\mu,\sigma}$ is the stationary point of \eqref{eq:continuous_time_ftrl}.
Therefore, by setting the start point as $\pi^0 = \pi^{\mu,\sigma}$, we have for all $t\geq 0, \pi^t=\pi^{\mu,\sigma}$.
In this case, for all $t\geq 0$, $\frac{d}{dt}\mathrm{KL}(\pi^{\mu,\sigma'}, \pi^t) = 0$ and then:
\begin{align*}
    &\left(\min_{i\in [N], ~j\in [d_i]}\sigma_{ij}'\right)\|\pi^{\mu,\sigma'} - \pi^{\mu,\sigma}\|^2 \leq \mathrm{diam}(\mathcal{X}) \sqrt{\sum_{i=1}^N \|\nabla_{\pi_i}G(\pi_i^t, \sigma_i') - \nabla_{\pi_i}G(\pi_i^t, \sigma_i)\|^2}.
\end{align*}
For a given $\varepsilon > 0$, let us define $\varepsilon' = \frac{\left(\min_{i\in [N], ~j\in [d_i]}\sigma_{ij}'\right)}{\mathrm{diam}(\mathcal{X})}\varepsilon^2$.
Since $\nabla_{\pi_i}G(\pi_i^{\mu,\sigma}, \cdot)$ is continuous on $\mathcal{X}_i$, for $\varepsilon'$, there exists $\delta > 0$ such that $\|\sigma' - \sigma\| < \delta \Rightarrow \sqrt{\sum_{i=1}^N\|\nabla_{\pi_i}G(\pi_i^{\mu,\sigma}, \sigma_i') - \nabla_{\pi_i}G(\pi_i^{\mu,\sigma}, \sigma_i)\|^2} < \varepsilon' = \frac{\left(\min_{i\in [N], ~j\in [d_i]}\sigma_{ij}'\right)}{\mathrm{diam}(\mathcal{X})}\varepsilon^2$.
Thus, for every $\varepsilon > 0$, there exists $\delta>0$ such that
\begin{align*}
    &\|\sigma' - \sigma\| < \delta \\
    &\Rightarrow \|\pi^{\mu,\sigma'} - \pi^{\mu,\sigma}\| \leq \sqrt{\frac{\mathrm{diam}(\mathcal{X})}{\left(\min_{i\in [N], ~j\in [d_i]}\sigma_{ij}'\right)}}\left(\sum_{i=1}^N \|\nabla_{\pi_i}G(\pi_i^{\mu,\sigma}, \sigma_i') - \nabla_{\pi_i}G(\pi_i^{\mu,\sigma}, \sigma_i)\|^2\right)^{\frac{1}{4}} < \varepsilon.
\end{align*}
This implies that $F(\cdot)$ is a continuous function on $\mathcal{X}$ when $G$ is $\alpha$-divergence.
A similar argument can be applied to R\'{e}nyi-divergence and reverse KL divergence.
\end{proof}

\subsection{Proof of Lemma \ref{lem:legendre}}
\begin{proof}[Proof of Lemma \ref{lem:legendre}]
First, from the definition of the Bregman divergence, for any $\pi_i\in \mathcal{X}_i$:
\begin{align}
\label{eq:bregman_div}
    D_{\psi}(\pi, T(y_i)) = \psi(\pi) - \psi(T(y_i)) - \langle \nabla \psi(T(y_i)), \pi_i - T(y_i)\rangle.
\end{align}
Recall that $\mathcal{X}_i$ satisfies $A\pi_i=b$ for all $\pi_i\in \mathcal{X}_i$ for a matrix $A\in \mathbb{R}^{k_i\times d_i}$ and $b\in \mathbb{R}^{k_i}$.
From the assumption for $\psi$ and the first-order optimality condition for the optimization problem of $\argmax_{x\in \mathcal{X}}\{\langle y_i, x\rangle - \psi(x)\}$, there exists $\nu\in \mathbb{R}^{k_i}$ such that
\begin{align*}
    y_i - \nabla\psi(T(y_i)) = A^{\top}\nu.
\end{align*}
Thus, we get:
\begin{align}
    \langle y_i, \pi_i - T(y_i)\rangle &= \langle A^{\top}\nu + \nabla \psi(T(y_i)), \pi_i - T(y_i)\rangle \nonumber\\
    &= \langle \nabla \psi(T(y_i)), \pi_i - T(y_i)\rangle + \nu^{\top}A\pi_i - \nu^{\top}AT(y_i) \nonumber\\
    &= \langle \nabla \psi(T(y_i)), \pi_i - T(y_i)\rangle + \nu^{\top}b - \nu^{\top}b \nonumber\\
    &= \langle \nabla \psi(T(y_i)), \pi_i - T(y_i)\rangle.
\label{eq:kkt}
\end{align}
By combining \eqref{eq:bregman_div} and \eqref{eq:kkt}, we have:
\begin{align*}
    D_{\psi}(\pi, T(y_i)) = \psi(\pi) - \psi(T(y_i)) - \langle y_i, \pi_i - T(y_i)\rangle.
\end{align*}
\end{proof}

\subsection{Proof of Lemma~\ref{lem:sum_eval}}\label{app:prf_sum_eval}
\begin{proof}[Proof of Lemma~\ref{lem:sum_eval}]
Since $\frac{1}{\kappa s + 2\theta} $ is a decreasing function for $s\ge0$, for all $s\ge1$,
$$
\frac{1}{\kappa s + 2\theta} \le \int_{s-1}^{s} \frac{1}{\kappa x + 2\theta} dx.
$$
Using this inequality, we can upper bound the sum as follows. 
\begin{align*}
    \sum_{s=0}^t \frac{1}{\kappa s + 2 \theta} & = \frac{1}{2 \theta} + \sum_{s=1}^t \frac{1}{\kappa s + 2 \theta}
    \\
    & \le \frac{1}{2 \theta} + \sum_{s=1}^t \int_{s-1}^{s} \frac{1}{\kappa x + 2\theta} dx
    \\
    & = \frac{1}{2 \theta} + \int_{0}^{t} \frac{1}{\kappa x + 2\theta} dx
    \\
    & = \frac{1}{2 \theta} + \frac{1}{\kappa} \int_{0}^{t} \frac{1}{ x + \frac{2\theta}{\kappa}} dx
    \\
    & = \frac{1}{2 \theta} + \frac{1}{\kappa} \int_{\frac{2\theta}{\kappa}}^{t + \frac{2\theta}{\kappa}} \frac{1}{ u} du \qquad \left(u = x +  \frac{2 \theta}{\kappa}\right)
    \\
    & = \frac{1}{2 \theta} + \frac{1}{\kappa} \ln \left(\frac{\kappa}{2 \theta} t + 1\right).
\end{align*}
This concludes the proof.
\end{proof}

\subsection{Proof of Lemma \ref{lem:stop_condition_alpha}}
\begin{proof}[Proof of Lemma \ref{lem:stop_condition_alpha}]
By using the first-order optimality condition for $\sigma_i^{k+1}$, we have for all $\pi \in \mathcal{X}$:
\begin{align*}
    \sum_{i=1}^N \langle \nabla_{\pi_i}v_i(\sigma_i^{k+1}, \sigma_{-i}^{k+1}) - \mu \nabla_{\pi_i}G(\sigma_i^{k+1}, \sigma_i^k), \pi_i - \sigma_i^{k+1}\rangle \leq 0,
\end{align*}
and then
\begin{align*}
    \sum_{i=1}^N \langle \nabla_{\pi_i}v_i(\sigma_i^{k+1}, \sigma_{-i}^{k+1}), \pi_i - \sigma_i^{k+1}\rangle \leq \mu \sum_{i=1}^N \langle \nabla_{\pi_i}G(\sigma_i^{k+1}, \sigma_i^k), \pi_i - \sigma_i^{k+1}\rangle.
\end{align*}

When $G$ is $\alpha$-divergence, we have for all $\pi\in \mathcal{X}$:
\begin{align*}
    \sum_{i=1}^N\langle \nabla_{\pi_i}G(\sigma_i^{k+1}, \sigma_i^k), \pi_i - \sigma_i^{k+1}\rangle &= \frac{1}{1-\alpha}\sum_{i=1}^N\sum_{j=1}^{d_i}(\sigma_{ij}^{k+1} - \pi_{ij})\left(\frac{\sigma_{ij}^k}{\sigma_{ij}^{k+1}}\right)^{1-\alpha} \\
    &= \frac{1}{1-\alpha}\sum_{i=1}^N\sum_{j=1}^{d_i}(\sigma_{ij}^{k+1} - \pi_{ij}) = 0,
\end{align*}
where we use the assumption that $\sigma^{k+1} = \sigma^k$ and $\mathcal{X}_i = \Delta^{d_i}$.
Similarly, when $G$ is R\'{e}nyi-divergence, we have for all $\pi\in \mathcal{X}$:
\begin{align*}
    \sum_{i=1}^N\langle \nabla_{\pi_i}G(\sigma_i^{k+1}, \sigma_i^k), \pi_i - \sigma_i^{k+1}\rangle &= \frac{\alpha}{1-\alpha}\sum_{i=1}^N\frac{1}{\sum_{j=1}^{d_i}(\sigma_{ij}^{k+1})^{\alpha}(\sigma_{ij}^k)^{1-\alpha}}\sum_{j=1}^{d_i}(\sigma_{ij}^{k+1} - \pi_{ij})\left(\frac{\sigma_{ij}^k}{\sigma_{ij}^{k+1}}\right)^{1-\alpha} \\
    &= \frac{\alpha}{1-\alpha}\sum_{i=1}^N\frac{1}{\sum_{j=1}^{d_i}(\sigma_{ij}^{k+1})^{\alpha}(\sigma_{ij}^k)^{1-\alpha}}\sum_{j=1}^{d_i}(\sigma_{ij}^{k+1} - \pi_{ij}) = 0.
\end{align*}
Furthermore, if $G$ is reverse KL divergence, we have for all $\pi\in \mathcal{X}$:
\begin{align*}
    \sum_{i=1}^N\langle \nabla_{\pi_i}G(\sigma_i^{k+1}, \sigma_i^k), \pi_i - \sigma_i^{k+1}\rangle &= \sum_{i=1}^N\sum_{j=1}^{d_i}(\sigma_{ij}^{k+1} - \pi_{ij})\frac{\sigma_{ij}^k}{\sigma_{ij}^{k+1}} \\
    &= \sum_{i=1}^N\sum_{j=1}^{d_i}(\sigma_{ij}^{k+1} - \pi_{ij}) = 0,
\end{align*}

Thus, we have for all $\pi \in \mathcal{X}$:
\begin{align*}
    \sum_{i=1}^N \langle \nabla_{\pi_i}v_i(\sigma_i^{k+1}, \sigma_{-i}^{k+1}), \pi_i - \sigma_i^{k+1}\rangle \leq 0.
\end{align*}
This is equivalent to the first-order optimality condition for $\pi^{\ast} \in \Pi^{\ast}$.
Therefore, $\sigma^{k+1}=\sigma^k$ is a Nash equilibrium of the underlying game.
\end{proof}

\subsection{Proof of Lemma \ref{lem:exact_conv_alpha}}
\begin{proof}[Proof of Lemma \ref{lem:exact_conv_alpha}]
First, we prove the statement for $\alpha$-divergence: $G(\sigma_i^{k+1}, \sigma_i^k) = \frac{1}{\alpha(1-\alpha)}\left(1 - \sum_{j=1}^{d_i}(\sigma_{ij}^{k+1})^{\alpha}(\sigma_{ij}^k)^{1-\alpha}\right)$.
From the definition of $\alpha$-divergence, we have for all $\pi^{\ast}\in \Pi^{\ast}$:
\begin{align*}
    \sum_{i=1}^N\langle \nabla_{\pi_i}G(\sigma_i^{k+1}, \sigma_i^k), \sigma_i^{k+1} - \pi_i^{\ast}\rangle &= \frac{1}{1-\alpha}\sum_{i=1}^N\sum_{j=1}^{d_i}(\pi_{ij}^{\ast} - \sigma_{ij}^{k+1})\left(\frac{\sigma_{ij}^k}{\sigma_{ij}^{k+1}}\right)^{1-\alpha} \\
    &= \frac{1}{1-\alpha}\sum_{i=1}^N\sum_{j=1}^{d_i}\pi_{ij}^{\ast}\left(\frac{\sigma_{ij}^k}{\sigma_{ij}^{k+1}}\right)^{1-\alpha} - \frac{1}{1-\alpha}\sum_{i=1}^N\sum_{j=1}^{d_i}(\sigma_{ij}^{k+1})^{\alpha}(\sigma_{ij}^k)^{1-\alpha}.
\end{align*}
Here, when $\alpha\in (0, 1)$, we get $\sum_{j=1}^{d_i}(\sigma_{ij}^{k+1})^{\alpha}(\sigma_{ij}^k)^{1-\alpha} \leq 1$.
Thus, 
\begin{align*}
    \sum_{i=1}^N\langle \nabla_{\pi_i}G(\sigma_i^{k+1}, \sigma_i^k), \sigma_i^{k+1} - \pi_i^{\ast}\rangle &\geq \frac{1}{1-\alpha}\sum_{i=1}^N\sum_{j=1}^{d_i}\pi_{ij}^{\ast}\left(\frac{\sigma_{ij}^k}{\sigma_{ij}^{k+1}}\right)^{1-\alpha} - \frac{N}{1-\alpha} \\
    &= \frac{N}{1-\alpha}\exp\left(\ln\left(\frac{1}{N}\sum_{i=1}^N\sum_{j=1}^{d_i}\pi_{ij}^{\ast}\left(\frac{\sigma_{ij}^k}{\sigma_{ij}^{k+1}}\right)^{1-\alpha}\right)\right) - \frac{N}{1-\alpha} \\
    &\geq \frac{N}{1-\alpha}\exp\left(\frac{1-\alpha}{N}\sum_{i=1}^N\sum_{j=1}^{d_i}\pi_{ij}^{\ast}\ln\frac{\sigma_{ij}^k}{\sigma_{ij}^{k+1}}\right) - \frac{N}{1-\alpha} \\
    &= \frac{N}{1-\alpha}\exp\left(\frac{1-\alpha}{N}\left(\mathrm{KL}(\pi^{\ast}, \sigma^{k+1}) - \mathrm{KL}(\pi^{\ast}, \sigma^k)\right)\right) - \frac{N}{1-\alpha},
\end{align*}
where the second inequality follows from the concavity of the $\ln(\cdot)$ function and Jensen's inequality for concave functions.
Since $\ln(\cdot)$ is strictly concave, the equality holds if and only if $\sigma^{k+1}=\sigma^k$.
Therefore, under the assumption that $\sigma^{k+1}\neq \sigma^k$, we get:
\begin{align}
    \mathrm{KL}(\pi^{\ast}, \sigma^{k+1}) - \mathrm{KL}(\pi^{\ast}, \sigma^k) &< \frac{N}{1-\alpha}\ln \left(1 + \frac{1-\alpha}{N}\sum_{i=1}^N\langle \nabla_{\pi_i}G(\sigma_i^{k+1}, \sigma_i^k), \sigma_i^{k+1} - \pi_i^{\ast}\rangle\right) \nonumber\\
    &\leq \sum_{i=1}^N\langle \nabla_{\pi_i}G(\sigma_i^{k+1}, \sigma_i^k), \sigma_i^{k+1} - \pi_i^{\ast}\rangle,
    \label{eq:kl_diff_alpha}
\end{align}
where the second inequality follows from $\ln(1+x) \leq x$ for $x> -1$.
From the first-order optimality condition for $\sigma_i^{k+1}$, we have for all $\pi^{\ast}\in \Pi^{\ast}$:
\begin{align*}
    \sum_{i=1}^N \langle \nabla_{\pi_i}v_i(\sigma_i^{k+1}, \sigma_{-i}^{k+1}) - \mu \nabla_{\pi_i}G(\sigma_i^{k+1}, \sigma_i^k), \pi_i^{\ast} - \sigma_i^{k+1}\rangle \leq 0.
\end{align*}
Then,
\begin{align}
    \label{eq:nabla_g_ineq}
    \sum_{i=1}^N \langle \nabla_{\pi_i}G(\sigma_i^{k+1}, \sigma_i^k), \sigma_i^{k+1} - \pi_i^{\ast}\rangle &\leq \frac{1}{\mu}\sum_{i=1}^N\langle \nabla_{\pi_i}v_i(\sigma_i^{k+1}, \sigma_{-i}^{k+1}), \sigma_i^{k+1} - \pi_i^{\ast}\rangle \nonumber\\
    &\leq \frac{1}{\mu}\sum_{i=1}^N\langle \nabla_{\pi_i}v_i(\pi_i^{\ast}, \pi_{-i}^{\ast}), \sigma_i^{k+1} - \pi_i^{\ast}\rangle,
\end{align}
where the second inequality follows from \eqref{eq:monotone}.
Moreover, since $\pi^{\ast}$ is the Nash equilibrium, from the first-order optimality condition, we get:
\begin{align}
    \label{eq:nash_optimality_condition}
    \sum_{i=1}^N \langle \sigma_i^{k+1} - \pi_i^{\ast}, \nabla_{\pi_i}v_i(\pi_i^{\ast}, \pi_{-i}^{\ast})\rangle \leq 0.
\end{align}
By combining \eqref{eq:kl_diff_alpha}, \eqref{eq:nabla_g_ineq}, and \eqref{eq:nash_optimality_condition}, if $\sigma^{k+1}\neq \sigma^k$, we have any $\pi^{\ast}\in \Pi^{\ast}$:
\begin{align*}
    \mathrm{KL}(\pi^{\ast}, \sigma^{k+1}) - \mathrm{KL}(\pi^{\ast}, \sigma^k) < 0.
\end{align*}

Next, we prove the statement for R\'{e}nyi-divergence: $G(\sigma_i^{k+1}, \sigma_i^k) = \frac{1}{\alpha - 1}\ln\left(\sum_{j=1}^{d_i}(\sigma_{ij}^{k+1})^{\alpha}(\sigma_{ij}^k)^{1-\alpha}\right)$.
We have for all $\pi^{\ast}\in \Pi^{\ast}$:
\begin{align*}
    \sum_{i=1}^N\langle \nabla_{\pi_i}G(\sigma_i^{k+1}, \sigma_i^k), \sigma_i^{k+1} - \pi_i^{\ast}\rangle &= \frac{\alpha}{1-\alpha}\sum_{i=1}^N\frac{1}{\sum_{j=1}^{d_i}(\sigma_{ij}^{k+1})^{\alpha}(\sigma_{ij}^k)^{1-\alpha}}\sum_{j=1}^{d_i}(\pi_{ij}^{\ast} - \sigma_{ij}^{k+1})\left(\frac{\sigma_{ij}^k}{\sigma_{ij}^{k+1}}\right)^{1-\alpha} \\
    &= \frac{\alpha}{1-\alpha}\sum_{i=1}^N\frac{1}{\sum_{j=1}^{d_i}(\sigma_{ij}^{k+1})^{\alpha}(\sigma_{ij}^k)^{1-\alpha}}\sum_{j=1}^{d_i}\pi_{ij}^{\ast}\left(\frac{\sigma_{ij}^k}{\sigma_{ij}^{k+1}}\right)^{1-\alpha} - \frac{N\alpha}{1-\alpha}.
\end{align*}
Again, by using $\sum_{j=1}^{d_i}(\sigma_{ij}^{k+1})^{\alpha}(\sigma_{ij}^k)^{1-\alpha} \leq 1$ when $\alpha \in (0,1)$, we get:
\begin{align*}
    \sum_{i=1}^N\langle \nabla_{\pi_i}G(\sigma_i^{k+1}, \sigma_i^k), \sigma_i^{k+1} - \pi_i^{\ast}\rangle &\geq \frac{\alpha}{1-\alpha}\sum_{i=1}^N\sum_{j=1}^{d_i}\pi_{ij}^{\ast}\left(\frac{\sigma_{ij}^k}{\sigma_{ij}^{k+1}}\right)^{1-\alpha} - \frac{N\alpha}{1-\alpha} \\
    &= \frac{N\alpha}{1-\alpha}\exp\left(\ln\left(\frac{1}{N}\sum_{i=1}^N\sum_{j=1}^{d_i}\pi_{ij}^{\ast}\left(\frac{\sigma_{ij}^k}{\sigma_{ij}^{k+1}}\right)^{1-\alpha}\right)\right) - \frac{N\alpha}{1-\alpha} \\
    &\geq \frac{N\alpha}{1-\alpha}\exp\left(\frac{1-\alpha}{N}\sum_{i=1}^N\sum_{j=1}^{d_i}\pi_{ij}^{\ast}\ln\frac{\sigma_{ij}^k}{\sigma_{ij}^{k+1}}\right) - \frac{N\alpha}{1-\alpha} \\
    &= \frac{N\alpha}{1-\alpha}\exp\left(\frac{1-\alpha}{N}\left(\mathrm{KL}(\pi^{\ast}, \sigma^{k+1}) - \mathrm{KL}(\pi^{\ast}, \sigma^k)\right)\right) - \frac{N\alpha}{1-\alpha},
\end{align*}
where the second inequality follows from Jensen's inequality for $\ln(\cdot)$ function.
Since $\ln(\cdot)$ is strictly concave, the equality holds if and only if $\sigma^{k+1}=\sigma^k$.
Therefore, under the assumption that $\sigma^{k+1}\neq \sigma^k$, we get:
\begin{align}
    \mathrm{KL}(\pi^{\ast}, \sigma^{k+1}) - \mathrm{KL}(\pi^{\ast}, \sigma^k) &< \frac{N\alpha}{1-\alpha}\ln \left(1 + \frac{1-\alpha}{N\alpha}\sum_{i=1}^N\langle \nabla_{\pi_i}G(\sigma_i^{k+1}, \sigma_i^k), \sigma_i^{k+1} - \pi_i^{\ast}\rangle\right) \nonumber\\
    &\leq \sum_{i=1}^N\langle \nabla_{\pi_i}G(\sigma_i^{k+1}, \sigma_i^k), \sigma_i^{k+1} - \pi_i^{\ast}\rangle,
    \label{eq:kl_diff_renyi}
\end{align}
where the second inequality follows from $\ln(1+x) \leq x$ for $x> -1$.
Thus, by combining \eqref{eq:nabla_g_ineq}, \eqref{eq:nash_optimality_condition}, and \eqref{eq:kl_diff_renyi}, if $\sigma^{k+1}\neq \sigma^k$, we have any $\pi^{\ast}\in \Pi^{\ast}$:
\begin{align*}
    \mathrm{KL}(\pi^{\ast}, \sigma^{k+1}) - \mathrm{KL}(\pi^{\ast}, \sigma^k) < 0.
\end{align*}

Finally, we prove the statement for reverse KL divergence: $G(\sigma_i^{k+1}, \sigma_i^k)=\sum_{j=1}^{d_i}\sigma_{ij}^k\ln \frac{\sigma_{ij}^k}{\sigma_{ij}^{k+1}}$.
We have for all $\pi^{\ast}\in \Pi^{\ast}$:
\begin{align*}
    \sum_{i=1}^N\langle \nabla_{\pi_i}G(\sigma_i^{k+1}, \sigma_i^k), \sigma_i^{k+1} - \pi_i^{\ast}\rangle &= \sum_{i=1}^N\sum_{j=1}^{d_i}(\pi_{ij}^{\ast} - \sigma_{ij}^{k+1})\frac{\sigma_{ij}^k}{\sigma_{ij}^{k+1}} \\
    &= \sum_{i=1}^N\sum_{j=1}^{d_i}\pi_{ij}^{\ast}\frac{\sigma_{ij}^k}{\sigma_{ij}^{k+1}} - N \\
    &= N\exp\left(\ln\left(\frac{1}{N}\sum_{i=1}^N\sum_{j=1}^{d_i}\pi_{ij}^{\ast}\frac{\sigma_{ij}^k}{\sigma_{ij}^{k+1}}\right)\right) - N \\
    &\geq N\exp\left(\frac{1}{N}\sum_{i=1}^N\sum_{j=1}^{d_i}\pi_{ij}^{\ast}\ln\frac{\sigma_{ij}^k}{\sigma_{ij}^{k+1}}\right) - N \\
    &= N\exp\left(\frac{1}{N}\left(\mathrm{KL}(\pi^{\ast}, \sigma^{k+1}) - \mathrm{KL}(\pi^{\ast}, \sigma^k)\right)\right) - N,
\end{align*}
where the inequality follows from Jensen's inequality for $\ln(\cdot)$ function.
Thus, under the assumption that $\sigma^{k+1}\neq \sigma^k$, we get:
\begin{align}
    \mathrm{KL}(\pi^{\ast}, \sigma^{k+1}) - \mathrm{KL}(\pi^{\ast}, \sigma^k) &< N\ln \left(1 + \frac{1}{N}\sum_{i=1}^N\langle \nabla_{\pi_i}G(\sigma_i^{k+1}, \sigma_i^k), \sigma_i^{k+1} - \pi_i^{\ast}\rangle\right) \nonumber\\
    &\leq \sum_{i=1}^N\langle \nabla_{\pi_i}G(\sigma_i^{k+1}, \sigma_i^k), \sigma_i^{k+1} - \pi_i^{\ast}\rangle,
    \label{eq:kl_diff_reverse_kl}
\end{align}
where the second inequality follows from $\ln(1+x) \leq x$ for $x> -1$.
Thus, by combining \eqref{eq:nabla_g_ineq}, \eqref{eq:nash_optimality_condition}, and \eqref{eq:kl_diff_reverse_kl}, if $\sigma^{k+1}\neq \sigma^k$, we have any $\pi^{\ast}\in \Pi^{\ast}$:
\begin{align*}
    \mathrm{KL}(\pi^{\ast}, \sigma^{k+1}) - \mathrm{KL}(\pi^{\ast}, \sigma^k) < 0.
\end{align*}
\end{proof}

\section{Additional Experimental Results and Details}
\label{sec:appx_experimental_detail}

\subsection{Payoff Matrix in Three-Player Biased RPS Game}

\begin{table}[h!]
    \centering
    \caption{Three-Player Biased RPS game matrix.}
    \label{tab:biased-rps}
    \begin{tabular}{cccc}
    \hline
      & R  & P  & S  \\ \hline
    R & $0$  & $-1/3$  & $1$ \\
    P & $1/3$ & $0$  & $-1/3$ \\
    S & $-1$  & $1/3$ & $0$  \\ \hline
    \end{tabular}
\end{table}

\subsection{Experimental Setting for Section \ref{sec:experiments}}
The experiments in Section \ref{sec:experiments} are conducted in Ubuntu 20.04.2 LTS with Intel(R) Core(TM) i9-10850K CPU @ 3.60GHz and 64GB RAM.

In the full feedback setting, we use a constant learning rate $\eta=0.1$ for MWU and OMWU, and APMD in all three games.
For APMD, we set $\mu=0.1$ and $T_{\sigma}=100$ for KL and reverse KL divergence perturbation, and set $\mu=0.1$ and $T_{\sigma}=20$ for squared $\ell^2$-distance perturbation.
As an exception, $\eta=0.01$, $\mu=1.0$, and $T_{\sigma}=200$ are used for APMD with squared $\ell^2$-distance perturbation in the random payoff games with $50$ actions.

For the noisy feedback setting, we use the lower learning rate $\eta=0.01$ for all algorithms, except APMD with squared $\ell^2$-distance perturbation for the random payoff games with $50$ actions.
We update the slingshot strategy profile $\sigma^k$ every $T_{\sigma}=1000$ iterations in APMD with KL and reverse KL divergence perturbation, and update it every $T_{\sigma}=200$ iterations in APMD with squared $\ell^2$-distance perturbation.
For APMD with $\ell^2$-distance perturbation in the random payoff games with $50$ actions, we set $\eta=0.001$ and $T_{\sigma}=2000$.

\subsection{Additional Experiments}
\label{sec:appx_additional_experiments}
In this section, we compare the performance of APMD and APFTRL to MWU, OMWU, and optimistic gradient descent (OGD) \citep{daskalakis2017training,wei2020linear} in the full/noisy feedback setting.
The parameter settings for MWU, OMWU, and APMD are the same as Section \ref{sec:experiments}.
For APFTRL, we use the squared $\ell^2$-distance and the parameter is the same as APMD with squared $\ell^2$-distance perturbation.
For OGD, we use the same learning rate as APMD with squared $\ell^2$-distance perturbation.

Figure \ref{fig:exploitability_full_w_ogda} shows the logarithm of the gap function for $\pi^t$ averaged over $100$ instances with full feedback.
We observe that APMD and APFTRL with squared $\ell^2$-distance perturbation exhibit competitive performance to OGD.
The experimental results in the noisy feedback setting are presented in Figure \ref{fig:exploitability_noisy_w_ogda}.
Surprisingly, in the noisy feedback setting, all APMD-based algorithms and the APFTRL-based algorithm exhibit overwhelmingly superior performance to OGD in all three games.

\begin{figure}[ht!]
    \centering

    \includegraphics[width=1.0\textwidth]{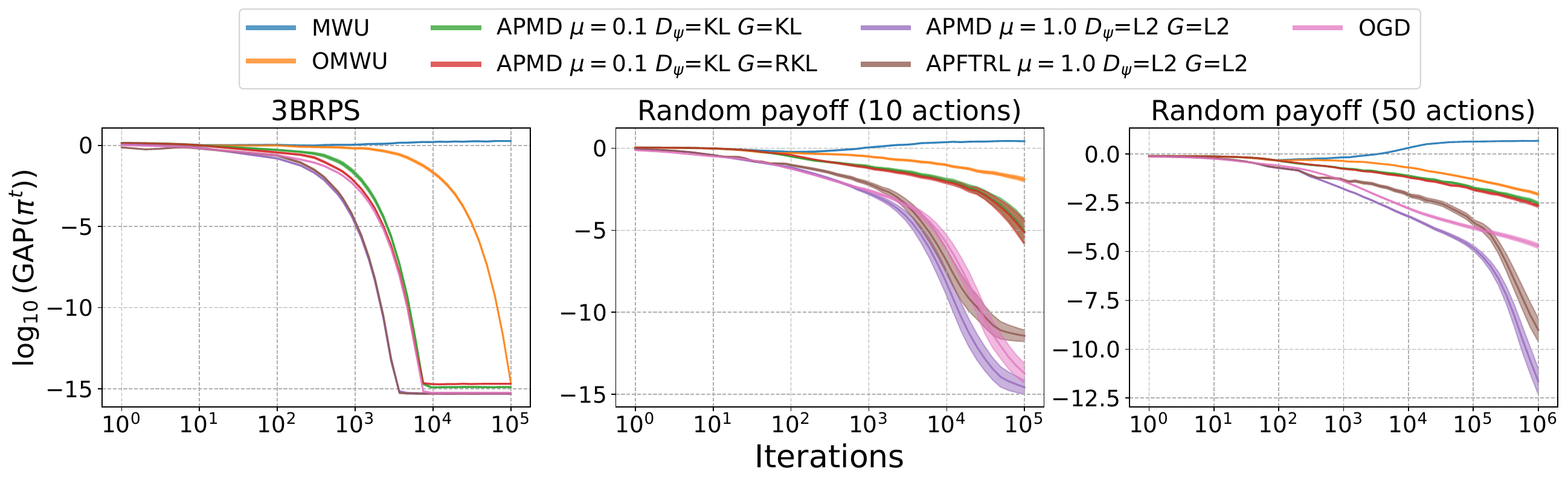}
    \caption{
    The gap function for $\pi^t$ for APMD, APFTRL, MWU, OMWU, and OGD with full feedback.
    The shaded area represents the standard errors. Note that the KL divergence, reverse KL divergence, and squared $\ell^2$-distance are abbreviated to KL, RKL, and L2, respectively.
    }
    \label{fig:exploitability_full_w_ogda}
\end{figure}

\begin{figure}[h!]
    \centering
    \includegraphics[width=1.0\textwidth]{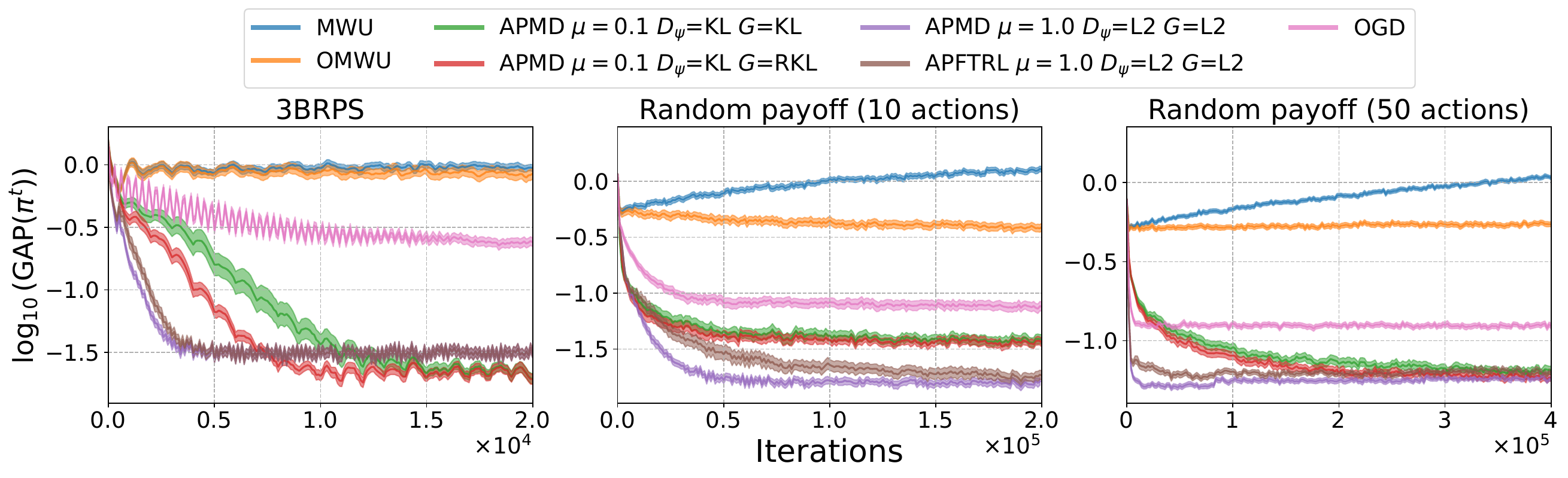}
    \caption{
    The gap function for $\pi^t$ for APMD, APFTRL, MWU, OMWU, and OGD with noisy feedback.
    The shaded area represents the standard errors.
    }
    \label{fig:exploitability_noisy_w_ogda}
\end{figure}

\subsection{Comparison with the Averaged Strategies of No-Regret Learning Algorithms}
This section compares the last-iterate strategies $\pi^t$ of APMD and APFTRL with the average of strategies $\frac{1}{t}\sum_{\tau=1}^t\pi^{\tau}$ of MWU, regret matching (RM) \citep{hart2000simple}, and regret matching plus (RM+) \citep{tammelin2014solving}.
The parameter settings for MWU, APMD, and APFTRL, as used in Section \ref{sec:appx_additional_experiments}, are maintained. 
Figure \ref{fig:exploitability_full_with_average} illustrates the logarithm of the gap function averaged over $100$ instances with full feedback.
The results show that the last-iterate strategies of APMD and APFTRL squared $\ell^2$-distance perturbation exhibit a lower gap than the averaged strategies of MWU, RM, and RM+.

\begin{figure}[h!]
    \centering
    \includegraphics[width=1.0\textwidth]{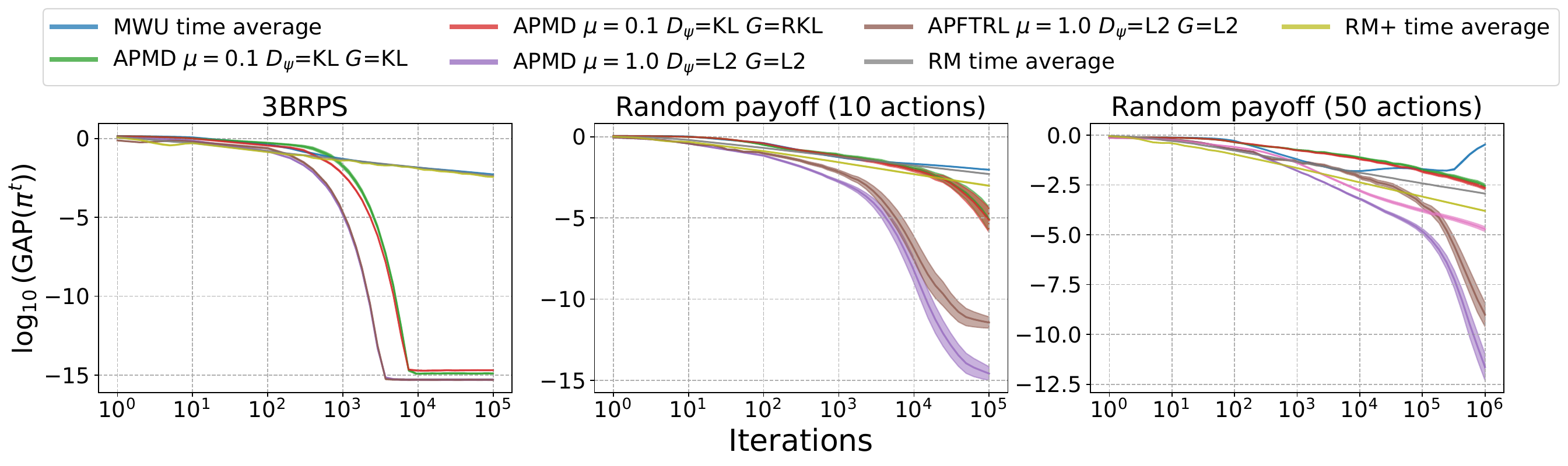}
    \caption{
    Comparison between the gap function of the last-iterate strategy profile of APMD, APFTRL, and the averaged strategy profile of MWU, RM, and RM+ with full feedback.
    The shaded area represents the standard errors.
    }
    \label{fig:exploitability_full_with_average}
\end{figure}

\subsection{Sensitivity Analysis of Update Interval for the Slingshot Strategy}
In this section, we investigate the performance when changing the update interval of the slingshot strategy.
We vary the $T_{\sigma}$ of APMD with L2 perturbation in 3BRPS with full feedback to be $T_{\sigma} \in \{10, 100, 1000, 10000\}$, and with noisy feedback to be $T_{\sigma} \in \{10, 100, 1000, 10000\}$.
All other parameters are the same as in Section ~\ref{sec:experiments}.
Figure ~\ref{fig:exploitability_for_t_sigma} shows the logarithm of the gap function for $\pi^t$ averaged over $100$ instances in 3BRPS with full/noisy feedback.
We observe that the smaller the $T_{\sigma}$, the faster $\pi^t$ converges.
However, if $T_{\sigma}$ is too small, $\pi^t$ does not converge (See $T_{\sigma} = 10$ with full feedback, and $T_{\sigma}=100$ with noisy feedback in Figure ~\ref{fig:exploitability_for_t_sigma}).

\begin{figure*}[h!]
    \centering
    \begin{minipage}[t]{0.49\textwidth}
        \centering
        \includegraphics[width=\linewidth]{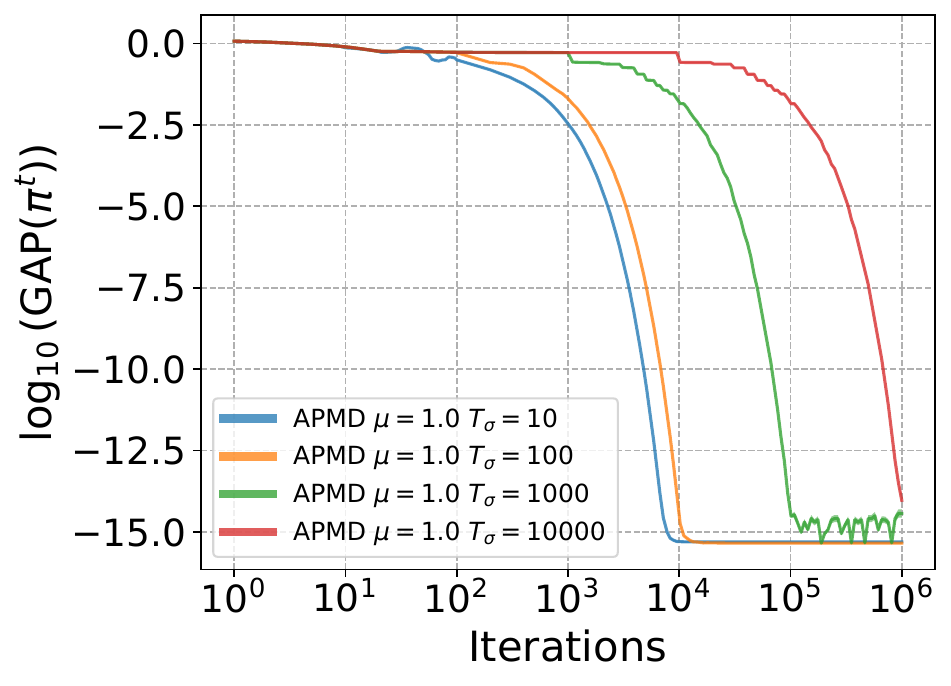}
        \subcaption{Full feedback}
    \end{minipage}
    \begin{minipage}[t]{0.49\textwidth}
        \centering
        \includegraphics[width=\linewidth]{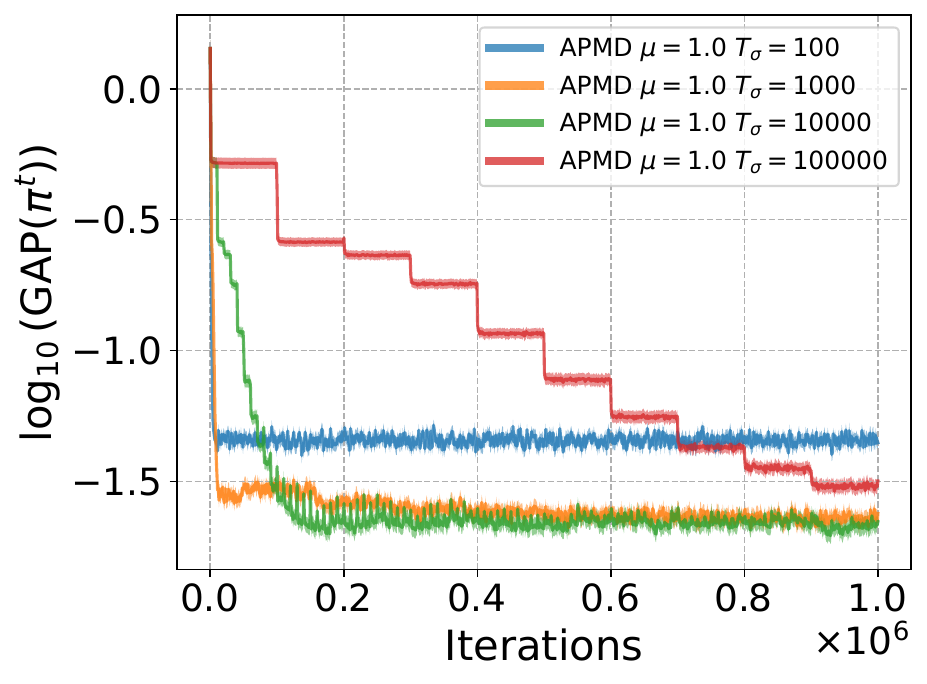}
        \subcaption{Noisy feedback}
    \end{minipage}
    \caption{The gap function for $\pi^t$ for APMD with varying $T_{\sigma}$ in 3BRPS with full/noisy feedback.
    The shaded area represents the standard errors.
    }
    \label{fig:exploitability_for_t_sigma}
\end{figure*}

\end{document}